\documentclass[11pt, fleqn]{article}
\usepackage{fullpage,xspace,amsmath,amsfonts,amssymb}
\usepackage{times,mathptmx}
\usepackage{graphicx,color,hyperref}
\usepackage{algorithmicx}
\usepackage{algorithm}
\usepackage{algpseudocode}
\usepackage{algc}
\usepackage{algcompatible} 		
\usepackage{algpascal} 	
\usepackage{cite}

\title{Time-Optimal Interactive Proofs for Circuit Evaluation}

\author{Justin Thaler \thanks{Harvard University, School of Engineering and Applied Sciences. Supported by an NSF Graduate Research Fellowship and
 NSF grants CNS-1011840 and CCF-0915922.}}

\date{}
\usepackage{amsfonts,amsmath,color,float,graphicx,verbatim, times}

\newenvironment{proof}[1][Proof: ]{\noindent \textbf{#1}}{\qed\medskip}

\newcommand{\poly}{\mathrm{poly}}

\newcommand{\polylog}{\mathrm{polylog}}

\DeclareMathAlphabet{\mathpzc}{OT1}{pzc}{m}{it}
\DeclareMathAlphabet{\mathcal}{OMS}{cmsy}{m}{n}

\newcommand{\matmult}{\textsc{matmult}\xspace}
\newcommand{\distinct}{\textsc{distinct}\xspace}

\newcommand{\eat}[1]{}

\newcommand{\cP}{\mathcal{P}}
\newcommand{\cV}{\mathcal{V}}

\newcommand{\etal}{{et al.}}

\usepackage{color}

\newcommand{\ignore}[1]{}
\newcommand{\provisionallyremove}[1]{}

\newcommand{\tV}{\tilde{V}}

\newtheorem{theorem}{Theorem}
\newtheorem{definition}{Definition}

\newtheorem{lemma}{Lemma}
\newtheorem{corollary}{Corollary}

\newtheorem{proposition}{Proposition}
\newtheorem{remark}{Remark}

\newcommand{\qed}{\hfill\rule{7pt}{7pt} \medskip}

\begin{document}
 
\maketitle


\begin{abstract}
Several research teams have recently been working toward
the development of practical general-purpose protocols for verifiable computation. 
These protocols enable a computationally weak \emph{verifier} to offload computations to a powerful but untrusted \emph{prover}, while
providing the verifier with a guarantee that the prover performed the requested computations correctly.
Despite substantial progress, existing implementations require further improvements before they become practical for most settings. The main bottleneck is typically the
extra effort required by the prover to return an answer with a guarantee of correctness, compared to returning an answer with no guarantee.

We describe a refinement of a powerful interactive proof protocol due to Goldwasser, Kalai, and Rothblum \cite{gkr}. Cormode, Mitzenmacher, and Thaler \cite{itcs}
show how to implement the prover in this protocol in time $O(S \log S)$, where $S$ is the size of an arithmetic circuit
 computing the function of interest. Our refinements apply to circuits with sufficiently ``regular'' wiring patterns; 
 for these circuits, we bring the runtime of the prover down to $O(S)$. 
 That is, our prover can evaluate the circuit with a guarantee of correctness, with only a constant-factor blowup in work compared
to evaluating the circuit with no guarantee.

We argue that our refinements capture a large class of circuits, and
we complement our theoretical results with experiments on problems such as matrix multiplication and determining the number of distinct elements in a data stream.
Experimentally, our refinements yield a 200x speedup for the prover over the implementation of Cormode \etal, and
our prover is less than 10x slower than a C++ program that simply evaluates the circuit.
Along the way, we describe a special-purpose protocol for matrix multiplication that is of interest in its own right.

Our final contribution is the design of an interactive proof protocol targeted at
general data parallel computation.  Compared to prior work, this protocol can more efficiently verify complicated computations 
as long as that computation is applied independently to many different pieces of data.



\end{abstract}

\newpage

\tableofcontents

\newpage

\section{Introduction}
Protocols for verifiable computation enable a computationally weak \emph{verifier} $\cV$ to offload computations to a powerful but untrusted \emph{prover} $\cP$. 
These protocols aim to provide the verifier with a guarantee that the prover performed the requested computations correctly, 
without requiring the verifier to perform the computations herself. 

Surprisingly powerful protocols for verifiable computation were discovered 
within the computer science theory community several decades ago, in the form of  interactive proofs (IPs) and their brethren, interactive arguments (IAs) and probabilistically checkable proofs (PCPs). 
In these protocols,
the prover $\cP$ solves a problem using her (possibly vast) computational resources, and tells $\cV$ the answer. $\cP$ and $\cV$ then have a conversation, 
i.e., they engage in a randomized protocol involving the exchange of one or more messages. 
During this conversation, $\cP$'s goal is to convince $\cV$ that the answer is correct.

Results quantifying the power of IPs, IAs, and PCPs represent some of the most celebrated results in all of computational complexity theory, 
but until recently they were mainly of theoretical interest, 
far too inefficient for actual deployment. In fact, the main applications of these results have traditionally been in negative applications -- showing that 
many problems are just as hard to approximate as they are to solve exactly.

However, the surging popularity of cloud computing has brought renewed interest in positive applications of protocols for verifiable computation. 
A typical motivating scenario is as follows. A business processes billions or trillions of transactions a day. 
The volume is sufficiently high that the business cannot or will not store and process the transactions on its own. Instead, it offloads the processing to a commercial cloud computing service.
The offloading of any computation raises issues of trust: the business may be concerned about relatively benign events like dropped transactions, buggy algorithms, or uncorrected hardware faults,
or the business may be more paranoid and fear that the cloud operator is deliberately deceptive or has been externally compromised. 
Either way, each time the business poses a query to the cloud, the business may demand that the cloud also provide a guarantee that the returned answer is correct. 

This is precisely what protocols for verifiable computation accomplish, with the cloud acting as the prover in the protocol, and the business acting as the verifier.
In this paper, we describe a refinement of an existing general-purpose protocol originally due to Goldwasser, Kalai, and Rothblum \cite{gkr, itcs}. When they are applicable, our techniques 
achieve asymptotically optimal runtime for the prover, and we demonstrate that they yield protocols
that are significantly closer to practicality than that achieved by prior work. 


We also make progress toward addressing another issue of existing interactive proof implementations: their applicability. 
The protocol of Goldwasser, Kalai, and Rothblum (henceforth the GKR protocol) applies in principle to any problem computed by a small-depth arithmetic circuit, 
but this is not the case when more fine-grained considerations of prover and verifier efficiency are taken into account. 
In brief, existing implementations of interactive proof protocols for circuit evaluation all require that the circuit have a highly regular wiring pattern \cite{itcs, allspice}. If this is not the case,
then these implementations require the verifier to perform an expensive (though data-independent) preprocessing phase to pull out information about the wiring of the circuit,
and they require a substantial factor blowup (logarithmic in the circuit size)  in runtime for the prover relative to evaluating the circuit without a guarantee of correctness. 
Developing a protocol that avoids these pitfalls and applies to more general computations
remains an important open question.

Our approach is the following. We do not have a magic bullet for dealing with irregular wiring patterns; if we want to avoid an expensive pre-processing phase 
for the verifier and minimize the blowup in runtime for the prover, we do need to make an assumption about the structure of the circuit we are verifying. 
Acknowledging this, we ask whether there is some general structure in real-world computations that we can leverage for efficiency gains. 

To this end, we design a protocol that is highly efficient for data parallel computation. By data parallel computation, we mean any
setting in which one applies the same computation independently to many pieces of data. Many outsourced computations are data parallel, with Amazon Elastic MapReduce\footnote{\texttt{http://aws.amazon.com/elasticmapreduce/}}
being one prominent example of a cloud computing service targeted specifically at data parallel computations.
Crucially, we do not want to make significant assumptions on the sub-computation that is being applied,
and in particular we want to handle sub-computations computed by circuits with highly irregular wiring patterns.

The verifier in our protocol still has to perform an offline phase to pull out information about the wiring of the circuit, but the cost of this phase 
is proportional to the size of a single instance of the sub-computation, avoiding
any dependence on the number of pieces of data to which the sub-computation is applied. 
Similarly, the blowup in runtime suffered by the prover is the same as it would be if the prover had run the basic GKR protocol on a single instance of the sub-computation.

Our final contribution is to describe a new protocol specific to matrix multiplication that is of interest in its own right. It avoids circuit evaluation entirely,
and reduces the overhead of the prover (relative to running \emph{any} unverifiable algorithm) to an additive low-order term. 

\subsection{Prior Work}  

\subsubsection{Work on Interactive Proofs.} 
\label{sec:ips}
Goldwasser, Kalai, and Rothblum described a powerful general-purpose interactive proof protocol in \cite{gkr}. 
This protocol
is framed in the context of \emph{circuit evaluation}. Given a layered arithmetic circuit $C$ of depth $d$, size $S(n)$, and fan-in 2,
the GKR protocol allows a prover to evaluate $C$ with a guarantee of correctness in time $\poly(S(n))$, 
while the verifier runs in time $\tilde{O}(n + d\log S(n))$, where $n$ is the length of the input and the $\tilde{O}$ notation
hides polylogarithmic factors in $n$. 

Cormode, Mitzenmacher, and Thaler showed how to bring the runtime of the prover in the GKR protocol down from $\poly(S(n))$ to $O(S(n) \log S(n))$ \cite{itcs}. 
They also built a full implementation of the protocol and ran it on benchmark problems. These results demonstrated that the protocol
does indeed save the verifier significant time in practice (relative to evaluating the circuit locally); they also demonstrated surprising scalability for the prover,
although the prover's runtime remained a major bottleneck. With the implementation of \cite{itcs} as a baseline, Thaler \etal\ \cite{hotcloud} described 
a parallel implementation of the GKR protocol that achieved 40x-100x speedups for the prover and 100x speedups for the (already fast) implementation of the verifier.

Vu, Setty, Blumberg, and Walfish \cite{allspice} further refine and extend the implementation of Cormode et al. \cite{itcs}. 
In particular, they combine the GKR protocol with a compiler from a high-level programming language so that programmers
do not have to explicitly express computation in the form of arithmetic circuits as was the case in the implementation of \cite{itcs}. This substantially extends
the reach of the implementation, but it should be noted that their approach generates circuits with irregular wiring patterns,
and hence only works in a \emph{batching} model, where the cost of a fairly expensive offline setup phase is amortized
by verifying many instances of a single computation in batch. They also build a hybrid system that statically evaluates whether 
it is better to use the GKR protocol or a different, cryptography-based argument system called Zaatar (see Section \ref{sec:args}), and runs the more efficient of the two protocols in an automated fashion.

A growing line of work studies protocols for verifiable computation in the context of \emph{data streaming}.
In this context, the goal is not just to save the verifier time (compared to doing the computation without a prover), but also to save the verifier space. 
The protocols developed in this line of work allow the client to make a single streaming pass over the input (which can occur, for example, while the client is uploading data to the cloud), keeping only a very small summary of the data set. 
The interactive version of this model was introduced by 
Cormode, Thaler, and Yi \cite{cty}, who observed that many protocols from the interactive proofs literature, including the GKR protocol, can be made to work in this
restrictive setting.
The observations of \cite{cty} imply that all of our protocols also work with streaming verifiers.
Non-interactive variants of the streaming interactive proofs model have also been studied in detail \cite{icalp09, esa, gur, klauck}. 

\subsubsection{Work on Argument Systems.} 
\label{sec:args} 
There has been a lot of work on the development of efficient interactive arguments,
which are essentially interactive proofs that are secure only against dishonest provers that run in polynomial time. 
A substantial body of work in this area has focused on the development of protocols targeted at specific problems (e.g. \cite{largedata, bonehfreeman, fg}). 
Other works have focused on the development of general-purpose argument systems.
Several papers in this direction (e.g. \cite{ckv, chiesabitansky,  ggp, chung}) have used fully homomorphic encryption, which unfortunately remains impractical despite substantial recent progress.
Work in this category by Chung \etal\ \cite{chung} focuses on streaming settings, and is therefore particularly relevant.

Several research teams have been pursuing the development of general-purpose argument systems that might be suitable for practical use.
Theoretical work by Ben-Sasson \etal\ \cite{pcps} focuses on the development of short PCPs that might be suitable for use in practice -- such PCPs can be compiled into efficient interactive arguments. 
As short PCPs are often a bottleneck in the development of efficient argument systems, other works have focused on avoiding their use \cite{itcs12, itcs13, bootstrap, nizk}.
In particular, Gennaro \etal\ \cite{nizk} and Bitansky \etal\ \cite{tcc} develop argument systems with a clear focus on implementation potential. 
Very recent work by Parno \etal\ \cite{pinocchio} describes
a near-practical general-purpose implementation, called Pinocchio, of an argument system based on \cite{nizk}. Pinocchio is additionally non-interactive and achieves
public verifiability.



Another line of implementation work focusing on general-purpose interactive argument systems is due to Setty \etal\  \cite{ndss, ginger, zaatar}.
This line of work begins with a base argument system due to Ishai \etal\ \cite{ishai}, and substantially refines the theory to achieve an implementation
that approaches practicality. The most recent system in this line of work is called Zaatar \cite{zaatar}, and is also based on the work of Gennaro \etal\ \cite{nizk}. 
An empirical comparison of the GKR-based approach and Zaatar performed by Vu \etal\ \cite{allspice} finds the GKR approach to be significantly more efficient for quasi-straight-line computations (e.g. programs with relatively simple control flow), while Zaatar is appropriate for programs with more complicated control flow. 

\subsection{Our Contributions}
Our primary contributions are three-fold. Our first contribution addresses
one of the biggest remaining obstacles to achieving a truly practical implementation of the GKR protocol: the logarithmic factor overhead for the prover. That is,
Cormode \etal\ show how to implement the prover in time  $O(S(n) \log S(n))$, where $S(n)$ is the size of the arithmetic circuit to which the GKR protocol is applied, down from the $\Omega(S(n)^3)$ time required for a naive implementation. The hidden constant in the Big-Oh notation is at least 3, and the $\log S(n)$ factor translates to well over an order of magnitude, even for circuits with a few million gates.

We remove this logarithmic factor, 
bringing $\cP$'s runtime down to $O(S(n))$ for a large class of circuits. Informally, our results apply to any circuit whose wiring pattern is sufficiently ``regular''. 
We formalize the class of circuits to which our results apply in Theorem \ref{thm:generaltheorem}.

We experimentally demonstrate the generality and effectiveness of Theorem \ref{thm:generaltheorem} via two case studies. 
Specifically, we apply an implementation of the protocol of Theorem \ref{thm:generaltheorem} to a circuit computing matrix multiplication ($\matmult$), as well as to a
circuit computing the number of distinct items in a data stream ($\distinct$). 
Experimentally, our refinements yield a 200x-250x speedup for the prover over the state of the art implementation of Cormode \etal\ \cite{itcs}.
A serial implementation of our prover is less than 10x slower than a C++ program that simply evaluates the circuit sequentially, a slowdown that is tolerable in realistic outsourcing scenarios
where cycles are plentiful for the prover.
Moreover, a parallel implementation 
of our prover using a graphics processing unit (GPU) is roughly 30x faster than our serial implementation, and therefore
 takes less time than that required to evaluate the circuit in serial. 

Our second contribution is to specify a highly efficient protocol for verifiably outsourcing arbitrary data parallel computation.  
Compared to prior work, this protocol  
can more efficiently verify complicated computations, 
as long as that computation is applied independently to many different pieces of data. We formalize this protocol and its efficiency guarantees
in Theorem \ref{thm:dataparallel}.

Our third contribution is to describe a new protocol specific to matrix multiplication that we believe to be of interest in its own right. 
This protocol is formalized in Theorem \ref{thm:finalfinalthm}.
Given any \emph{unverifiable} algorithm for $n \times n$ matrix multiplication that requires time $T(n)$ using space $s(n)$, Theorem \ref{thm:finalfinalthm} allows the prover
to run in time $T(n) + O(n^2)$ using space $s(n) + o(n^2)$.
Note that Theorem \ref{thm:finalfinalthm} (which is specific to matrix multiplication) is much less general than Theorem \ref{thm:generaltheorem} (which applies
to any circuit with a sufficiently regular wiring pattern). 
However, Theorem \ref{thm:finalfinalthm} achieves optimal runtime and space usage for the prover up to leading constants, assuming there is no $O(n^2)$ time algorithm
for matrix multiplication. 
While these properties are also satisfied by a classic protocol due to Freivalds \cite{freivalds}, the protocol of Theorem \ref{thm:finalfinalthm} is significantly
more amenable for use as a primitive when verifying computations
that repeatedly invoke matrix multiplication. 
For example, using the protocol of Theorem \ref{thm:finalfinalthm} as a primitive, we give a natural protocol
for computing the diameter of an unweighted directed graph $G$. 
$\cV$'s runtime in this protocol is $O(m \log n)$, where $m$ is the number of edges in $G$, 
$\cP$'s runtime matches the best known unverifiable diameter algorithm up to a low-order additive term \cite{apsp, diameterref}, and the total communication is just $\polylog(n)$. 
We know of no other protocol achieving this. 


We complement Theorem \ref{thm:finalfinalthm} with experimental results demonstrating its efficiency. 

\subsection{Roadmap}
Section \ref{sec:prelim} presents preliminaries. We give a high-level overview of the ideas underlying
our main results in Section \ref{sec:methodology}. Section \ref{sec:details} gives a detailed
overview of prior work, including the standard sum-check protocol as well as the GKR protocol. Section \ref{sec:ourrefinements} contains
 the details of our time-optimal protocol for circuit evaluation as formalized in Theorem \ref{thm:generaltheorem}. Section \ref{sec:expts} describes our experimental
cases studies of the protocol described in Theorem \ref{thm:generaltheorem}. 
Section \ref{sec:dataparallel} describes our protocol for arbitrary data parallel computation.
Section \ref{sec:finalopt} describes some additional optimizations that apply to specific important wiring patterns. In particular, this section describes
our special-purpose protocol for \matmult\ that achieves optimal prover efficiency up to leading constants.
Section \ref{sec:conclusion} concludes.

\section{Preliminaries}
\label{sec:prelim}
\subsection{Definitions}
We begin by defining a valid interactive proof protocol for a function $f$. 

\begin{definition} \label{def:bigdef} Consider a prover $\cP$ and verifier $\cV$ who both observe an input $x$ and wish to compute a function $f:\{0, 1\}^n \rightarrow \mathcal{R}$ for some set $\mathcal{R}$.
After the input is observed, 
$\cP$ and $\cV$ exchange a sequence of messages. Denote the output of $\cV$ on input $x$, given prover $\cP$ and $\cV$'s random bits $R$, by $\text{out}(\cV,x, R,\cP)$. 
$\cV$ can output $\perp$ if $\cV$ is not convinced that $\cP$'s claim is valid. 

We say $\cP$ is a \emph{valid prover} with respect to $\cV$ if for all inputs $x$, $\text{Pr}_R[\text{out}(\cV, x, R, \cP) = f(x)] = 1$. The property that there is at least one valid prover $\cP$ with respect to $\cV$ is called \emph{completeness}.
We say $\cV$ is a \emph{valid verifier} for $f$ with \emph{soundness probability} $\delta$ if there is at least one valid prover $\cP$ with respect to $\cV$, and for all provers $\cP'$ and all
inputs $x$, $\text{Pr}[\text{out}(\cV,A,R,\cP')\notin \{f(x), \perp\}]\leq \delta$. 
We say a prover-verifier pair $(\cP, \cV)$ is
a \emph{valid interactive proof protocol} for $f$ if $\cV$ is a valid verifier for $f$ with soundness probability $1/3$, and $\cP$ is a valid prover with respect to $\cV$.
If $\cP$ and $\cV$ exchange $r$ messages in total, we say the protocol has $\lceil r/2\rceil$ rounds. 
\end{definition}

Informally, the completeness property
guarantees that an honest prover will convince the verifier that the claimed answer is correct, while the soundness property
ensures that a dishonest prover will be caught with high probability. 
An \emph{interactive argument}
is an interactive proof where the soundness property holds only against polynomial-time provers $\cP'$. 
We remark that the constant $1/3$ used for the soundness probability in Definition \ref{def:bigdef} is chosen for consistency with the interactive proofs literature, 
where $1/3$ is used by convention. In our actual implementation, the soundness probability will always be less than $2^{-45}$.

\subsubsection{Cost Model}
Whenever we work over a finite field $\mathbb{F}$, we assume that a single field operation can be computed in a single machine operation. 
For example, when we say that the prover $\cP$ in our interactive protocols requires time $O(S(n))$,
we mean that $\cP$ must perform $O(S(n))$ additions and multiplications within the finite field over which the protocol is defined.

\medskip \noindent
\textbf{Input Representation.} Following prior work \cite{icalp09, itcs, cty}, all of the protocols we consider can handle inputs specified in a general data stream form. 
Each element of the stream is a tuple $(i, \delta)$, where $i \in [n]$ and $\delta$ is an integer. The $\delta$ values may be negative, thereby modeling deletions. 
The data stream implicitly defines a frequency vector $a$, where $a_i$ is the sum of all $\delta$ values associated with $i$ in the stream. For simplicity, we assume
throughout the paper that the number of stream updates $m$ is related to $n$ by a constant factor i.e., $m=\Theta(n)$.

When checking the evaluation of a circuit $C$, we consider the inputs to $C$ to be the entries of the frequency vector $a$.	
We emphasize that in all of our protocols, $\cV$ only needs to see the raw stream and not the aggregated frequency vector $a$ (see Lemma \ref{lemma:streamingv} for details).
Notice that we may interpret the frequency vector $a$ as an object other than a vector, such as a matrix or a string. For example, in \matmult, the data
stream defines two matrices to be multiplied. 

When we refer to a \emph{streaming verifier} with space usage $s(n)$, we mean that the verifier can make a single pass over the stream of tuples defining the input,
regardless of their ordering, while storing at most $s(n)$ elements in the finite field over which the protocol is defined.

\subsubsection{Problem Definitions}
To focus our discussion in this paper, we give special attention to two problems also considered in prior work \cite{itcs, hotcloud}.

\begin{enumerate}
\item In the \matmult\ problem, the input consists of two $n\times n$ matrices $A, B \in \mathbb{Z}^{n \times n}$, and the goal is to compute the matrix product $A \cdot B$. 
\item In the \distinct\ problem, also denoted $F_0$, the input is a data steam consisting of $m$ tuples $(i, \delta)$ from a universe of size $n$. The stream defines a frequency vector $a$,
and the goal is to compute $|\{i: a_i \neq 0\}|$, the number of items with 
non-zero frequency. 
\end{enumerate}

\subsubsection{Additional Notation}
\label{sec:additionalnotation}
Throughout, $[n]$ will denote the set $\{1, \dots, n\}$, while $[[n]]$ will denote the set $\{0, \dots, n-1\}$. 

Let $\mathbb{F}$ be a field, and $\mathbb{F}^* = \mathbb{F} \setminus \{0\}$ its multiplicative group. For any $d$-variate polynomial $p(x_1, \dots, x_d) : \mathbb{F}^{d} \rightarrow \mathbb{F}$, we use $\deg_i(p)$ to denote the degree of $p$ in variable $i$. A $d$-variate polynomial $p$ is said to be \emph{multilinear} if 
$\deg_i(p) \leq 1$ for all $i \in [d]$.
Given a function $V: \{0, 1\}^d \rightarrow \{0, 1\}$ whose domain is the $d$-dimensional  Boolean hypercube, the \emph{multilinear extension} (MLE)
of $V$ over $\mathbb{F}$, denoted $\tilde{V}$, is the unique multilinear polynomial $\mathbb{F}^{d} \rightarrow \mathbb{F}$ that agrees with $V$ on all Boolean-valued inputs. 
That is, $\tilde{V}$ is the unique multilinear polynomial over $\mathbb{F}$ satisfying $\tilde{V}(x) = V(x)$ for
all $x \in \{0, 1\}^{d}$.

\section{Overview of the Ideas}
\label{sec:methodology}
We begin by describing the methodology underlying the GKR protocol before summarizing the ideas underlying our improved protocols.

\subsection{The GKR Protocol From 10,000 Feet}
In the GKR protocol, $\cP$ and $\cV$ first agree on an arithmetic circuit $C$ of fan-in 2 over a finite field $\mathbb{F}$ computing the function of interest ($C$ may have multiple outputs). Each gate of $C$
performs an addition or multiplication over $\mathbb{F}$. 
$C$ is assumed to be in layered form,
meaning that the circuit can be decomposed into layers, and wires only connect gates in adjacent layers. Suppose the circuit has depth $d$; 
we will number the layers from 1 to $d$ with layer $d$ referring to the input layer, and layer $1$ referring to the output layer.

In the first message, $\cP$ tells $\cV$ the
(claimed) output of the circuit. The protocol then works its way in iterations towards the input layer, with one iteration devoted to each layer. 
The purpose of iteration $i$ is to reduce a claim about the values of the gates at layer $i$ to a claim about the values of the gates at layer $i+1$,
in the sense that it is safe for $\cV$ to assume that the first claim is true as long as the second claim is true. This reduction is accomplished by 
applying the standard \emph{sum-check protocol} \cite{lfkn} to a certain polynomial.

More concretely, the GKR protocol starts with a claim about the values of the output gates of the circuit, but $\cV$ cannot check this claim without evaluating 
the circuit herself, which is precisely what she wants to avoid. So the first iteration uses a sum-check protocol to reduce this claim about the outputs of the circuit to a claim about the
gate values at layer 2 (more specifically, to a claim about an evaluation of the multilinear extension (MLE) of the gate values at layer 2). Once again, $\cV$ cannot check this claim herself, so the second iteration uses another sum-check protocol 
to reduce the latter claim to a claim about the gate values at layer 3, and so on.
Eventually, $\cV$ is left with a claim about the inputs to the circuit, and $\cV$ can check this claim on her own. 

In summary, the GKR protocol uses a sum-check protocol at each level of the circuit to enable 
$\cV$ to go from verifying a randomly chosen evaluation of the MLE of the gate values at layer $i$ to verifying a (different) evaluation of the MLE of the gate values at layer $i+1$. 
Importantly, apart from the input layer and output layer, $\cV$ does not ever see all of the gate values at a layer (in particular, $\cP$ does not send these values in full).
Instead, $\cV$ relies on $\cP$ to do the hard work of actually evaluating the circuit, and uses the power of the sum-check protocol as the main tool to force $\cP$ to be consistent and truthful
over the course of the protocol.

\subsection{Achieving Optimal Prover Runtime for Regular Circuits} 
In Theorem \ref{thm:generaltheorem}, we describe an interactive proof protocol
for circuit evaluation that brings $\cP$'s runtime down to $O(S(n))$ for a large class of circuits, while maintaining the same verifier runtime as in prior implementations of the GKR protocol. 
Informally, Theorem \ref{thm:generaltheorem} applies to any circuit whose wiring pattern is sufficiently ``regular''. 

This protocol follows the same general outline as the GKR protocol, in that we proceed in iterations from the output layer of the circuit to the input layer,
using a sum-check protocol at iteration $i$ to reduce a claim about the gate values at layer $i$ to a claim about the gate values at layer $i+1$. 
However, at each iteration $i$ we apply the sum-check protocol to a carefully chosen polynomial that differs from the one used by GKR. 
In each round $j$ of the sum-check protocol, our choice of polynomial allows $\cP$ to reuse work from prior rounds in order to compute the prescribed message for round $j$,
allowing us to shave a $\log S(n)$ factor
from the runtime of $\cP$ relative to the $O(S(n) \log S(n))$-time implementation due to Cormode \etal\ \cite{itcs}.

Specifically, at iteration $i$, the GKR protocol uses a polynomial $f_z^{(i)}$ defined over $\log S_i + 2 \log S_{i+1}$ variables, where $S_i$ is the number of gates at layer $i$.
The ``truth table'' of $f_z^{(i)}$ is sparse on the Boolean hypercube, in the sense that 
$f_z^{(i)}(x)$ is non-zero for at most $S_i$ of the $S_i \cdot S_{i+1}^2$ inputs $x \in \{0, 1\}^{\log S_i + 2 \log S_{i+1}}$. Cormode \etal\ leverage this sparsity to 
bring the runtime of $\cP$ in iteration $i$ down to $O(S_i \log S_i)$ from a naive bound of $\Omega(S_i \cdot S_{i+1}^2)$. 
However, this same sparsity prevents $\cP$ from reusing work from prior iterations
as we seek to do.

In contrast, we use a polynomial $g_z^{(i)}$ defined over only $\log S_i$ variables rather than $\log S_i + 2 \log S_{i+1}$ variables. Moreover, the truth table of $g_z^{(i)}$  is dense on the Boolean hypercube,
in the sense that $g_z^{(i)}(x)$ may be non-zero for all of the $S_i$ Boolean inputs $x \in \{0, 1\}^{\log S_i}$. This density
allows $\cP$ to reuse work from prior iterations in order to speed up her computation in round $i$ of the sum-check protocol. 

In more detail, in each round $j$ of the sum-check protocol, the prover's prescribed message is defined via a sum over a large number of terms,
where the number of terms falls geometrically fast with the round number $j$.
Moreover, it can be shown that in each round $j$, each gate at layer $i+1$ contributes to exactly one term of this sum. 
Essentially, what we do is group the gates at layer $i+1$ by the term of the sum to which they
contribute. Each such group can be treated as a single unit, ensuring that in any round of the sum-check protocol,
the amount of work $\cP$ needs to do is proportional to the number of terms in the sum rather than the number of gates $S_i$ at layer $i$.

We remark that a similar ``reuse of work'' technique was implicit in an analysis by Cormode, Thaler, and Yi \cite[Appendix B]{cty} of an efficient protocol for a specific streaming problem
known as the second frequency moment. This frequency moment protocol was the direct inspiration for our refinements, though we require additional insights to apply the reuse of work technique in the context
of evaluating general arithmetic circuits.

It is worth clarifying why our methods do not yield savings when applied to the polynomial $f_z^{(i)}$ used in the basic GKR protocol.
The reason is that, since $f_z^{(i)}$ is defined over $\log S_i + 2 \log S_{i+1}$ variables instead of just $\log S_i$ variables, the sum defining $\cP$'s message in round $j$ 
is over a much larger number of terms when using $f_z^{(i)}$. 
It is still the case that each gate contributes to only one term of the sum, but until the number of terms in the sum falls below $S_i$
(which does not happen until round $j=\log S_i + \log S_{i+1}$ of the sum-check protocol), 
it is possible for each gate to contribute to a different term. Before this point, grouping gates by the term of the sum to which they contribute is not useful,
since each group can have size 1. 


\subsection{Verifying General Data Parallel Computations}
Theorem \ref{thm:generaltheorem} only applies to circuits with regular wiring patterns, as do other
existing implementations of interactive proof protocols for circuit evaluation \cite{itcs, allspice}. For circuits with irregular wiring patterns,
these implementations require the verifier to perform an expensive preprocessing phase (requiring time proportional to the size of the circuit)
to pull out information about the wiring of the circuit, and they require a substantial factor blowup (logarithmic in the circuit size) 
in runtime for the prover relative to evaluating the circuit without a guarantee of correctness. 

To address these bottlenecks, we do need to make an assumption about the structure of the circuit we are verifying.
Ideally our assumption will be satisfied by many real-world computations. 
To this end, Theorem \ref{thm:dataparallel} will describe a protocol that is highly efficient for any data parallel computation, by which we mean any
setting in which one applies the same computation independently to many pieces of data. See Figure \ref{fig:datapar} in Section \ref{sec:dataparallel} 
for a schematic of a data parallel computation.

The idea behind Theorem \ref{thm:dataparallel} is as follows. Let $C$ be a circuit of size $S$ with an arbitrary wiring pattern, and let $C^*$ be a ``super-circuit'' that applies $C$ independently to $B$ different inputs 
before possibly aggregating the results in some fashion. If one naively applied the basic GKR protocol to the super-circuit $C^*$, $\cV$ might have to perform a pre-processing phase
that requires time proportional to the size of $C^*$, which is $\Omega(B \cdot S)$. 
Moreover, when applying the basic GKR protocol to $C^*$, $\cP$ would require time $\Theta\left(B \cdot S \cdot \log(B \cdot S) \right)$. 

In order to improve on this, the key observation is that although each sub-computation $C$ can have a very complicated wiring pattern, 
the circuit is ``maximally regular'' between sub-computations, as the sub-computations do not interact at all. Therefore,
each time the basic GKR protocol would apply the sum-check protocol to a polynomial
derived from the wiring predicate of $C^*$, we instead use a simpler polynomial derived only from the wiring predicate of $C$.
This immediately brings the time required by $\cV$ in the pre-processing phase down to $O(S)$, which is proportional to the cost of executing
a single instance of the sub-computation.
By using the reuse of work technique underlying Theorem \ref{thm:generaltheorem}, we are also able to bring $\cP$'s
runtime down from  $\Theta\left(B \cdot S \cdot \log(B \cdot S)\right)$ to $\Theta\left(B \cdot S \cdot \log S \right)$, i.e., $\cP$'s requires 
a factor of $O(\log S)$ more time to evaluate the circuit with a guarantee of correctness, compared to evaluating the circuit without such a guarantee.
This $O(\log S)$ factor overhead does not depend on the batch size $B$.

Our improvements are most significant when $B \gg S$, i.e., when a (relatively) small but potentially complicated sub-computation is 
applied to a very large number of pieces of data. For example, given any very large database, one may ask ``How many people in the database satisfy Property $P$?'' 
Our protocol allows one to verifiably outsource
such \emph{counting} queries with overhead that depends minimally on the size of the database, but that necessarily depends on the complexity of the property $P$.

\subsection{A Special-Purpose Protocol for \matmult}
We describe a special-purpose protocol for $n \times n$ \matmult\ in Theorem \ref{thm:finalfinalthm}. The idea behind this protocol is as follows. The GKR protocol, as well the
protocols of Theorems \ref{thm:generaltheorem} and \ref{thm:dataparallel}, only make use of the multilinear
extension $\tilde{V}_i$ of the function $V_i$ mapping gate labels at  layer $i$ of the circuit to their values. 
In some cases, there is something to be gained by using a higher-degree extension
of $V_i$, and this is precisely what we exploit here.

In more detail, our special-purpose protocol can be viewed as an extension of our circuit-checking techniques applied to a circuit $C$ performing naive matrix multiplication,
but using a quadratic extension of the gate
values in this circuit. This allows us to verify the computation using a single invocation of the sum-check protocol. 
More importantly, $\cP$ can evaluate this higher-degree extension at the necessary points 
without explicitly materializing all of the gate values of $C$, which would not be possible if we had used the multilinear extension of the gate values of $C$.

In the protocol of Theorem \ref{thm:finalfinalthm}, $\cP$ just needs to compute
the correct output (possibly using an algorithm that is much more sophisticated than naive matrix multiplication), and then perform $O(n^2)$ additional work
to prove the output is correct. Since $\cP$ does not have to evaluate $C$ in full, this protocol is perhaps best viewed outside the lens of 
circuit evaluation. Still, the idea underlying Theorem \ref{thm:finalfinalthm} can be thought of as a refinement of our circuit evaluation protocols,
and we believe that similar ideas may yield further improvements to general-purpose protocols in the future.

\section{Technical Background} 
\label{sec:details}

\subsection{Schwartz-Zippel Lemma}
We will often make use of the following basic property of polynomials.

\begin{lemma}[\cite{zippel}]  \label{lemma:schwartzzippel} Let $\mathbb{F}$ be any field, and let $f: \mathbb{F}^m \rightarrow \mathbb{F}$ be a nonzero polynomial of total degree
$d$. Then on any finite set $S \subseteq \mathbb{F}$, $$\Pr_{x \leftarrow S^m}[f(x) = 0] \leq d/|S|.$$ In words, if $x$ is chosen uniformly at random from $S^m$, then the probability that $f(x)=0$
is at most $d/|S|$. In particular, any two distinct polynomials of total degree $d$
can agree on at most $d/|S|$ fraction of points in $S^m$.
\end{lemma}

\subsection{Sum-Check Protocol}
\label{sec:sumcheck}
Our main technical tool is the sum-check protocol \cite{lfkn}, and we present a full description of this protocol for completeness.
See also \cite[Chapter 8]{arorabarak} for
a complete exposition and proof of soundness. 

Suppose we are given a $v$-variate polynomial $g$ defined over a finite field $\mathbb{F}$. The purpose of the sum-check protocol is to compute the sum:
$$H:=\sum_{b_1 \in \{0,1\}} \sum_{b_2 \in \{0, 1\}} \dots \sum_{b_v \in \{0, 1\}} g(b_1, \dots, b_v).$$

In order to 
execute the protocol, the verifier needs to be able to evaluate $g(r_1, \dots, r_v)$ for a randomly chosen vector $(r_1, \dots, r_v) \in \mathbb{F}^v$ --
see the paragraph preceding Proposition \ref{prop:thepropforsumcheck} below.

The protocol proceeds in $v$ rounds as follows.
In the first round, the prover sends a polynomial $g_1(X_1)$, and claims that
$g_1(X_1) = \sum_{x_2, \dots,x_v \in \{0, 1\}^{v-1}}	g(X_1,x_2,\dots,x_v)$.	
Observe that if $g_1$ is as claimed, then  $H= g_1(0) + g_1(1)$.
Also observe that the polynomial $g_1(X_1)$ has degree $\deg_1(g)$, the degree of variable $x_1$ in $g$.
Hence $g_1$ can be specified with $\deg_1(g)+1$ field elements. In our implementation, $\cP$ will specify $g$ by 
sending the evaluation of $g$ at each point in the set $\{0, 1, \dots, \deg_1(g)\}$.

Then, in round $j > 1$, $\cV$ chooses a value $r_{j-1}$ uniformly at random from $\mathbb{F}$ and sends $r_{j-1}$ to $\cP$. 
We will often refer to this step by saying that variable $j-1$ gets \emph{bound} to value $r_{j-1}$.
In return, the prover sends a polynomial $g_j(X_j)$, and claims that 
\begin{equation} \label{messagedef} g_j(X_j) =	\sum_{(x_{j+1}, \dots, x_{v}) \in \{0, 1\}^{v-j}}	g(r_1,\dots,r_{j-1}, X_j, x_{j+1}, \dots, x_v).\end{equation}

The verifier compares the two most recent polynomials by checking that
$g_{j-1}(r_{j-1}) = g_j(0) + g_j(1)$,
and rejecting otherwise. The verifier also rejects if the degree of $g_j$ is too high: each $g_j$ should have degree $\deg_j(g)$, the degree of variable $x_j$ in $g$. 

In the final round, the prover has sent $g_v(X_v)$ which is claimed to be
$g(r_1,\dots,r_{v-1},X_v)$.
$\cV$ now checks that $g_v(r_v) = g(r_1, \dots, r_v)$ (recall that we assumed $\cV$ can evaluate $g$ at this point). If this test succeeds, and so do all previous tests, then the verifier accepts, and is
convinced that $H=g_1(0) + g_1(1)$.

\begin{proposition} \label{prop:thepropforsumcheck} Let $g$ be a $v$-variate polynomial defined over a finite field $\mathbb{F}$, and let
$(\cP, \cV)$ be the prover-verifier pair in the above description of the sum-check protocol. 
$(\cP, \cV)$ is
a valid interactive proof protocol for the function $H = \sum_{b_1 \in \{0,1\}} \sum_{b_2 \in \{0, 1\}} \dots \sum_{b_v \in \{0, 1\}} g(b_1, \dots, b_v)$.
\end{proposition}

\subsubsection{Discussion of costs.} 
\label{sec:sumcheckcosts}
Observe that 
there is one round in the sum-check protocol for each of the $v$ variables of $g$. The total communication is $\sum_{i=1}^v \deg_i(g)+1 = v + \sum_{i=1}^v \deg_i(g)$ field elements.
In all of our applications, $\deg_i(g)=O(1)$ for all $i$, and so the communication cost is $O(v)$ field elements.

The running time of the verifier over the entire execution of the protocol is proportional
to the total communication, plus
the amount of time required to compute $g(r_1, \dots, r_v)$.

Determining the running time of the prover is less straightforward. 
Recall that $\cP$ can specify $g_j$ by sending for each $i \in \{0, \dots, \deg_j(g)\}$ the value:

\begin{equation} \label{eqi} g_j(i) = \sum_{(x_{j+1}, \dots, x_{v}) \in \{0, 1\}^{v-j}} g(r_1, \dots, r_{j-1}, i, x_{j+1}, \dots, x_v).\end{equation}

An important insight is that the number of terms defining the value $g_j(i)$ in Equation \eqref{eqi} falls geometrically with $j$: in the $j$th sum, there are
only $2^{v-j}$ terms, each corresponding to a Boolean vector in $\{0, 1\}^{v-j}$. 
The total number of terms that must be evaluated over the course of the protocol is therefore $O\left(\sum_{j=1}^v 2^{v-j}\right) = O(2^v)$. 
Consequently, if $\cP$ is given oracle access to the truth table of the polynomial $g$, then $\cP$ will require just $O(2^v)$ time. 

Unfortunately, in our applications $\cP$ will not have oracle access to the truth table of $g$. The key to our results is to show that
in our applications $\cP$ can nonetheless evaluate $g$ at all of the necessary points in $O(2^v)$ total time. 


\subsection{The GKR Protocol}
\label{sec:gkr}
We describe the details of the GKR protocol for completeness,
as well as to simplify the exposition of our refinements.

\subsubsection{Notation} 
Suppose we are given a layered arithmetic circuit $C$ of size $S(n)$, depth $d(n)$, and fan-in two. Let $S_i$ denote the number of gates at layer $i$ of the circuit $C$. 
Assume $S_i$ is a power of 2 and let $S_i = 2^{s_i}$. 
In order to explain 
how each iteration of the GKR protocol proceeds, we need to introduce several functions, each of which encodes certain information about the circuit.

To this end, number the gates at layer $i$ from $0$ to $S_i-1$, and let $V_i:\{0, 1\}^{s_i} \rightarrow \mathbb{F}$ denote the function
that takes as input a binary gate label, and outputs the corresponding gate's value at layer $i$. 
The GKR protocol makes use of the multilinear extension $\tilde{V}_i$ of the function $V_i$ (see Section \ref{sec:additionalnotation}). 

The GKR protocol also makes use of the notion of a ``wiring predicate'' that encodes which pairs of wires from layer $i+1$ are connected to a given gate at layer $i$ in $C$.
We define two functions, $\text{add}_i$ and $\text{mult}_i$ mapping $\{0, 1\}^{s_i + 2s_{i+1}}$  to $\{0, 1\}$, which together constitute the
wiring predicate of layer $i$ of $C$. Specifically, these functions
take as input three gate labels $(j_1 , j_2 , j_3)$, and return 1 if gate
$j_1$ at layer $i$ is the addition (respectively, multiplication)
of gates $j_2$ and $j_3$ at layer $i+1$, and return 0 otherwise. Let $\tilde{\text{add}}_i$ and $\tilde{\text{mult}}_i$ denote the multilinear extensions of 
$\text{add}_i$ and $\text{mult}_i$ respectively. 

Finally, let $\beta_{s_i}(z, p)$ denote the function 
$$\beta_{s_i}(z, p) = \prod_{j=1}^{s_i} \left((1-z_j)(1-p_j) + z_jp_j\right).$$ 
It is straightforward to check that $\beta_{s_i}$ is the multilinear extension of the function $B(x, y) : \{0, 1\}^{s_i} \times \{0, 1\}^{s_i} \rightarrow \{0, 1\}$ that evaluates to 1 if
 $x=y$, and evaluates to 0 otherwise.

\subsubsection{Protocol Outline}
\label{sec:protocoloutline}
The GKR protocol consists of $d(n)$ iterations, one for each layer of the circuit.
Each iteration starts with $\cP$ claiming a value for $\tilde{V}_i(z)$ for some field element $z \in \mathbb{F}^{s_i}$.
In the first iteration and circuits with a single output gate, $z=0$ and $\tilde{V}_1(0)$ corresponds to the output value of the circuit.

For circuits with many output gates, Vu \etal\ \cite{allspice} observe that in the first iteration,
$\cP$ may simply send $\cV$ the (claimed) values of all output gates, thereby specifying a function $V'_1 : \{0, 1\}^{s_1} \rightarrow \mathbb{F}$ claimed
to equal $V_1$.
$\cV$ can pick a random point $z \in \mathbb{F}^{s_1}$ and evaluate $\tilde{V}'_1(z)$ on her own in $O(S_1)$ time (see Remark \ref{remark:vu} in Section \ref{sec:streamingv}). 
The Schwartz-Zippel Lemma (Lemma \ref{lemma:schwartzzippel}) implies that it is safe for $\cV$ to
believe that $V'_1$ indeed equals $V_1$ as claimed, as long as $\tilde{V}_1(z)=\tilde{V'}_1(z)$ (which will be checked in the remainder of the protocol).

The purpose of iteration $i$ is to reduce the claim about the value of $\tilde{V}_i(z)$ to a claim about $\tilde{V}_{i+1}(\omega)$ for some $\omega \in \mathbb{F}^{s_{i+1}}$, in the sense that  it is safe for $\cV$ to assume that the first claim is true 
as long as the second claim is true.
To accomplish this, the iteration applies the sum-check protocol described in Section \ref{sec:sumcheck} to a specific polynomial derived
from $\tilde{V}_{i+1}$, $\tilde{\text{add}}_i$, and $\tilde{\text{mult}}_i$, and $\beta_{s_i}$. 

\subsubsection{Details for Each Iteration}

\noindent \textbf{Applying the Sum-Check Protocol.} It can be shown 
that for any $z \in \mathbb{F}^{s_i}$,

$$\tilde{V}_i(z) = \sum_{(p, \omega_1, \omega_2) \in \{0, 1\}^{s_i + 2 s_{i+1}}} f_z^{(i)}(p, \omega_1, \omega_2),$$

where \begin{equation} \label{eqf}
 \!\!\! \!\!\!f_z^{(i)}(p, \omega_1, \omega_2) =\beta_{s_i}(z, p) \cdot \left( \tilde{\text{add}}_i(p, \omega_1, \omega_2) (\tilde{V}_{i+1}(\omega_1) + \tilde{V}_{i+1}(\omega_2)) + 
 \tilde{\text{mult}}_i(p, \omega_1, \omega_2) \tilde{V}_{i+1}(\omega_1) \cdot \tilde{V}_{i+1}(\omega_2) \right).\end{equation}
 
Iteration $i$ therefore applies the sum-check protocol of Section \ref{sec:sumcheck} to the polynomial $f_z^{(i)}$. 
There remains the issue that $\cV$ can only execute her part of the sum-check protocol if she can evaluate the polynomial $f_z^{(i)}$ at a random point $f_z^{(i)}(r_1, \dots, r_{s_i+2s_{i+1}})$. 
This is handled as follows.

Let $p^*$ denote the first $s_i$ entries of the vector $(r_1, \dots, r_{s_i+2s_{i+1}})$, $\omega_1^*$ the next $s_{i+1}$ entries, and $\omega_2^*$ the last $s_{i+1}$ entries.
Evaluating $f_z^{(i)}(p^*, \omega_1^*, \omega_2^*)$ requires evaluating $\beta(z, p^*)$, $\tilde{\text{add}}_i(p^*, \omega_1^*, \omega_2^*)$, $\tilde{\text{mult}}_i(p^*, \omega_1^*, \omega_2^*)$, $\tilde{V}_{i+1}(\omega_1^*)$,
and $\tilde{V}_{i+1}(\omega_2^*)$.

$\cV$ can easily evaluate $\beta(z, p^*)$ in $O(s_i)$ time. 
For many circuits, particularly those with ``regular'' wiring patterns, $\cV$ can evaluate  $\tilde{\text{add}}_i(p^*, \omega_1^*, \omega_2^*)$ and $\tilde{\text{mult}}_i(p^*, \omega_1^*, \omega_2^*)$
on her own in $\poly(s_i, s_{i+1})$ time as well.\footnote{Various suggestions have been put forth for what to do if this is not the case. For example, these computations can always be done by $\cV$ in $O(\log S(n))$ \emph{space}
as long as the circuit is log-space uniform,
which is sufficient in streaming applications where the space usage of the verifier is paramount \cite{itcs}. Moreover, these computations can be done offline before the input is even observed, because they only depend on the wiring of the circuit, and not on the input \cite{gkr, itcs}. Finally, \cite{allspice} notes that
the cost of this computation can be effectively amortized in a batching model, where many identical computations on different inputs are verified simultaneously. See Section \ref{sec:dataparallel}
for further discussion, and a protocol that mitigates this issue in the context of data parallel computation.}
 
$\cV$ cannot however evaluate $\tilde{V}_{i+1}(\omega_2^*)$,
and $\tilde{V}_{i+1}(\omega_1^*)$ on her own without evaluating the circuit.  Instead, $\cV$ asks $\cP$ to simply tell her these two values, and uses
iteration $i+1$ to \emph{verify} that these values are as claimed. However, one complication remains: the precondition for iteration $i+1$ is that $\cP$ claims 
 a value for $\tilde{V}_i(z)$ for a single $z \in \mathbb{F}^{s_i}$. So 
$\cV$ needs to reduce verifying both $\tilde{V}_{i+1}(\omega_2^*)$
and $\tilde{V}_{i+1}(\omega_1^*)$ to verifying $\tilde{V}_{i+1}(\omega^*)$ at a single point $\omega^* \in \mathbb{F}^{s_{i+1}}$, in the sense that
it is safe for $\cV$ to accept the claimed values of $\tilde{V}_{i+1}(\omega_1^*)$ and $\tilde{V}_{i+1}(\omega_2^*)$ as long as the value of $\tilde{V}_{i+1}(\omega^*)$ is as claimed.
This is done as follows.

\medskip \noindent \textbf{Reducing to Verification of a Single Point.}
Let $\ell: \mathbb{F} \rightarrow \mathbb{F}^{s_{i+1}}$ be some canonical line passing through $\omega_1^*$ and $\omega_2^*$. For example,
we can let $\ell$ be the unique line such that $\ell(0)= \omega_1^*$ and  $\ell(1)= \omega_2^*$.
$\cP$ sends a degree-$s_{i+1}$ polynomial $h$ claimed to be $\tilde{V}_{i+1} \circ \ell$, the restriction of $\tilde{V}_{i+1}$ to the line $\ell$.
$\cV$ checks that $h(0) = \omega_1^*$ and $h(1) = \omega_2^*$ (rejecting if this is not the case), picks a random point $r^* \in \mathbb{F}$, and 
asks $\cP$ to prove that $\tilde{V}_{i+1}(\ell(r^*))= h(r^*)$. By the Schwartz-Zippel Lemma (Lemma \ref{lemma:schwartzzippel}), as long as $\cV$ is convinced that $\tilde{V}_{i+1}(\ell(r^*))= h(r^*)$, it is safe for $\cV$
to believe that the values of $\tilde{V}_{i+1}(\omega_1^*)$ and  $\tilde{V}_{i+1}(\omega_2^*)$ are as claimed by $\cP$.
This completes iteration $i$; $\cP$ and $\cV$ then move on to the iteration for layer $i+1$ of the circuit, whose purpose is to verify that $\tilde{V}_{i+1}(\ell(r^*))$ has the claimed value.

\medskip
\noindent \textbf{The Final Iteration.}
Finally, at the final iteration $d$, 
$\cV$ must evaluate $\tilde{V}_d(\omega^*)$ on her own. But the vector of gate values at layer $d$ of  $C$ is simply the input $x$ to $C$.
It can be shown that $\cV$ can compute $\tilde{V}_d(\omega^*)$ on her own in $O(n \log n)$ time, with a single streaming pass over the input \cite{cty}.
Moreover, Vu \etal\ show how to bring $\cV$'s time cost down to $O(n)$ \cite{allspice}, but this methodology does not work in a general streaming model. For completeness,
we present details of both of these observations in Section \ref{sec:streamingv}. 

\subsubsection{Discussion of Costs.}
Observe that the polynomial $f_z^{(i)}$ defined in Equation \eqref{eqf} is an $\left(s_i + 2s_{i+1}\right)$-variate polynomial of degree at most $2$ in each variable, and so the invocation of the sum-check protocol at iteration $i$ requires $s_i + 2 s_{i+1}$ rounds, with three field elements transmitted per round.
 Thus, the total communication cost is $O(d(n)\log S(n))$ field elements, where $d(n)$ is the depth of the circuit $C$.  The time cost to $\cV$ is $O(n \log n + d(n) \log S(n))$, where the $n \log n$ 
 term is due to the time required to evaluate $\tilde{V}_d(\omega^*)$ (see Lemma \ref{lemma:streamingv} below), and the $d(n) \log S(n)$ term is the time required for $\cV$ to send messages to $\cP$ and process and check the messages from $\cP$. 

As for $\cP$'s runtime, for any iteration $i$ of the GKR protocol, a naive implementation of the prover in the corresponding instance of the sum-check protocol
would require time $\Omega(2^{s_{i} + 2 s_{i+1}})$, as the sum defining each of $\cP$'s messages is over as many as $2^{s_i + 2 s_{i+1}}$ terms. This cost can be $\Omega(S(n)^3)$,
which is prohibitively large in practice.
However, Cormode, Mitzenmacher, and Thaler showed in \cite{itcs} that each gate at layers $i$ and $i+1$ of $C$ contributes 
to only a \emph{single} term of sum, and exploit this to bring the runtime of the $\cP$ down to $O(S(n) \log S(n))$. 

\subsubsection{Making $\cV$ Fast vs. Making $\cV$ Streaming}
\label{sec:streamingv}
We describe how $\cV$ can efficiently evaluate $\tilde{V}_d(\omega^*)$ on her own, as required in the final iteration of the GKR protocol. Prior work has
identified two methods for performing this computation. The first method is due to Cormode, Thaler, and Yi \cite{cty}. It requires $O(n \log n)$ time, and allows
$\cV$ to make a single streaming pass over the input using $O(\log n)$ space.

\begin{lemma}[\cite{cty}] \label{lemma:streamingv}
Given an input $x \in \mathbb{F}^n$ and a vector $\omega^* \in \mathbb{F}^{\log n}$, $\cV$ can compute $\tilde{V}_d(\omega^*)$ in $O(n \log n)$ time and $O(\log n)$ space with a single streaming pass over the input, where $\tilde{V}_d$ is the multilinear extension of the function that maps
$i \in \{0, 1\}^{\log n}$ to the value of the $i$th entry of $x$.
\end{lemma}
\begin{proof}
We exploit the following explicit expression for $\tV_{d}$. 
For a vector $b \in\{0,1\}^{\log n}$ let $\chi_b(x_1, \dots, x_{\log n})=\prod_{k=1}^{\log n} \chi_{b_k}(x_k)$, where $\chi_{0}(x_k)=1-x_k$ and $\chi_1(x_k) = x_k$. Notice
that $\chi_b$ 
is the unique multilinear polynomial that takes $b \in \{0,1\}^{\log n}$ to 1 and all other values in $\{0,1\}^{\log n}$ to 0, i.e., it is the multilinear extension of the indicator function for boolean vector $b$. 
With this definition in hand, we may write:

\begin{equation} \label{eq:linextnew} \tV_{d}(p_1, \dots, p_{\log n}) = \sum_{b \in \{0, 1\}^{\log n}} V_{d}(b) \chi_b(p_1, \dots p_{\log n}) \end{equation}

Indeed, it is easy to check that the right hand side of Equation \eqref{eq:linextnew} is a multilinear polynomial, and that it agrees with $V_d$ 
 on all Boolean inputs. Hence, the right hand side must equal the multilinear extension of $V_d$.
 
In particular, by letting $(p_1, \dots, p_{\log n})=\omega^*$ in Equation \eqref{eq:linextnew}, we see that
\begin{equation} \label{eq:import} \tV_{d}(\omega^*) = \sum_{b \in \{0, 1\}^{\log n}} V_{d}(b) \chi_b(\omega^*).\end{equation}

Given any stream update $(i, \delta)$, let $(i_1, \dots, i_{\log n})$ denote the binary representation of $i$. 
Notice that update $(i, \delta)$ has the effect of increasing $V_d(i_1, \dots, i_{\log n})$ by $\delta$,
and does not affect $V_d(x_1, \dots x_{\log n})$ for any $(x_1, \dots, x_{\log n}) \neq (i_1, \dots, i_{\log n})$.
Thus, $\cV$ can compute $\tV_d(\omega^*)$ 
incrementally from the raw stream by initializing $\tV_d(\omega^*) \leftarrow 0$,
and processing each update $(i, \delta)$ via:
$$\tV_d(\omega^*) \leftarrow \tV_d(\omega^*) + \delta \cdot \chi_i(\omega^*).$$ 
$\cV$ only needs to store $\tV_d(\omega^*)$ and $\omega^*$, which requires $O(\log n)$ words of memory. Moreover, for any $i$, $\chi_{(i_1, \dots, i_{\log n})}(\omega^*)$ can be computed in $O(\log n)$ field operations, and thus $\cV$ can compute $\tV_d(\omega^*)$ with one pass over the raw stream, using $O(\log n)$ words of space and $O(\log n)$ field operations per update.
\end{proof}

The second method is due to Vu \etal\ \cite{allspice}. It enables $\cV$ to compute $\tilde{V}_d(\omega^*)$ in $O(n)$ time, but requires $\cV$ to use $O(n)$ space.
\begin{lemma}[Vu \etal\ \cite{allspice}] \label{lemma:vu} 
$\cV$ can compute $\tilde{V}_d(\omega^*)$ in $O(n)$ time and $O(n)$ space.
\end{lemma}
\begin{proof}
We again exploit the expression for $\tV_d(\omega^*)$ in Equation \eqref{eq:import}. 
Notice the right hand side of Equation \eqref{eq:import} expresses $\tV_d(\omega^*)$ as the inner product of two $n$-dimensional vectors,
where the $b$th entry of the first vector is $V_d(b)$ and the $b$th entry of the second vector is $\chi_b(\omega^*)$.
This inner product can be computed in $O(n)$ time given a table of size $n$
whose $b$th entry contains the quantity $\chi_b(\omega^*)$.
Vu \etal\ show how to build such a table in time $O(n)$ using memoization. 

The memoization procedure consists of $\log n$ stages, where Stage $j$ constructs a table $A^{(j)}$ of size $2^j$,
such that for any $(b_1, \dots, b_j) \in \{0, 1\}^j$, $A^{(j)}[(b_1, \dots, b_j)] = \prod_{i=1}^{j} \chi_{b_i}(\omega^*_i)$.
Notice  $A^{(j)}[(b_1, \dots, b_{j})] = A^{(j-1)}[(b_1, \dots, b_{j-1})] \cdot \chi_{b_j}(\omega^*_j)$,
and so the $j$th stage of the memoization procedure requires time $O(2^j)$. The total time across all $\log n$ stages is therefore
$O(\sum_{j=1}^{\log n} 2^{j}) = O(2^{\log n}) = O(n)$. This completes the proof.
\end{proof}

\begin{remark} \label{remark:vu} In \cite{personal}, Vu \etal\ further observe that if the input is presented in a specific order, then 
$\cV$ can evaluate $\tV_d(\omega^*)$ using $O(\log n)$ space. Compare this result to Lemma \ref{lemma:streamingv}, which requires $O(n \log n)$ time for
$\cV$, but allows $\cV$ to use $O(\log n)$ space regardless of the order in which the input is presented.
\end{remark}  


\section{Time-Optimal Protocols for Circuit Evaluation}
\label{sec:ourrefinements}
\subsection{Protocol Outline and Section Roadmap}
\label{sec:youroutline}
As with the GKR protocol, our protocol consists of $d(n)$ iterations, one for each layer of the circuit.
Each iteration starts with $\cP$ claiming a value for $\tilde{V}_i(z)$ for some value $z \in \mathbb{F}^{s_i}$.
The purpose of the iteration is to reduce this claim to a claim about $\tilde{V}_{i+1}(\omega)$ for some $\omega \in \mathbb{F}^{s_{i+1}}$, 
in the sense that  it is safe for $\cV$ to assume that the first claim is true as long as the second claim is true. As in the GKR protocol,
this is done by invoking the sum-check protocol on a certain polynomial. 

In order to improve on the costs of the GKR protocol implementation of Cormode \etal\ \cite{itcs}, we replace the polynomial $f_z^{(i)}$ in Equation \eqref{eqf}
with a different polynomial $g_z^{(i)}$ defined over a much smaller domain. Specifically, $g_z^{(i)}$ is defined over only $s_i$ variables rather than $s_i + 2 s_{i+1}$ variables as is the case of $f_z^{(i)}$. 
Using $g_z^{(i)}$ in place of $f_z^{(i)}$ allows
 $\cP$ to reuse work across iterations of the sum-check protocol, thereby reducing $\cP$'s runtime by a logarithmic factor relative to \cite{itcs},
 as formalized in Theorem \ref{thm:generaltheorem} below.
 
 
The remainder of the presentation leading up to Theorem \ref{thm:generaltheorem} proceeds as follows. 
After stating a preliminary lemma, we describe the polynomial $g_z^{(i)}$ that we use in the context of three specific circuits:
a binary tree of addition or multiplication gates, 
and a circuit computing the number of non-zero entries of an $n$-dimensional vector $a$. 
The purpose of this exposition is  
to showcase the ideas underling Theorem \ref{thm:generaltheorem} in concrete scenarios. 
Second, we explain the algorithmic insights that allow $\cP$ to reuse work across iterations of the sum-check protocol applied to $g_z^{(i)}$.
 Finally, we state and prove Theorem \ref{thm:generaltheorem}, which formalizes the class of circuits to which our methods apply.

\subsection{A Preliminary Lemma} 
We will repeatedly invoke the following lemma, which allows us to express the value $\tilde{V}_i(z)$ in a manner amenable
 to verification via the sum-check protocol. This is essentially a restatement of \cite[Lemma 3.2.1]{rothblumthesis}.
 
 \begin{lemma}
 \label{prop:big}
 Let $W$ be any polynomial $\mathbb{F}^{s_i} \rightarrow \mathbb{F}$ that extends $V_{i}$, in the sense that for all $p \in \{0, 1\}^{s_i}$,
 $W(p) = V_i(p)$. Then for any $z \in \mathbb{F}^{s_i}$,
 \begin{equation} \label{eqprop} \tilde{V}_i(z) = \sum_{p \in \{0, 1\}^{s_i}} \beta_{s_i}(z, p) W(p).\end{equation}
 \end{lemma}
 \begin{proof}
 It is easy to check that the right hand side of Equation \eqref{eqprop} is a multilinear polynomial in $z$, and that it agrees with $V_i$ 
 on all Boolean inputs. Thus, the right hand side of Equation \eqref{eqprop}, viewed as a polynomial in $z$, must be the multilinear extension $\tilde{V}_i$ of 
 $V_i$. This completes the proof.
 \end{proof}
 


\subsection{Polynomials for Specific Circuits}
\label{sec:polynomials}

\subsubsection{The Polynomial for a Binary Tree}
\label{sec:polyforbintree}
Consider a circuit $C$ that computes the product of all $n$ of its inputs by multiplying them together via a binary tree.
Label the gates at layers $i$ and $i+1$ in the natural way, so that the first input to the gate labelled $p=(p_1, \dots, p_{s_i}) \in \{0, 1\}^{s_i}$ at layer $i$ is the gate with label $(p, 0)$ at
layer $i-1$, and the second input to gate $p$ has label $(p, 1)$. 
Here and throughout, $(p, 0)$ denotes the $s_i+1$-dimensional vector obtained by concatenating the entry 0 to the end of the vector $p$.
Interpreting $p=(p_1, \dots, p_{s_i}) \in \{0, 1\}^{s_i}$ as an integer between $0$ and $2^{s_i}-1$ with $p_1$ as the high-order
bit and $p_{s_i}$ as the low-order bit, this says that
the first in-neighbor of $p$ is $2p$ and the second is $2p+1$.
It follows immediately that for any gate $p \in \{0, 1\}^{s_i}$ at layer $i$,
$V_i(p) = \tV_{i+1}(p, 0) \cdot \tV_{i+1}(p, 1)$. Invoking Lemma \ref{prop:big}, we obtain the following proposition.

\begin{proposition}
\label{prop:bintreemult}
Let $C$ be a circuit consisting of a binary tree of multiplication gates. Then $\tilde{V}_i(z) = \sum_{p \in \{0, 1\}^{s_i}} g_z^{(i)}(p)$, where  $g_z^{(i)}(p) = \beta_{s_i}(z, p)\cdot  \tV_{i+1}(p, 0) \cdot \tV_{i+1}(p, 1).$
\end{proposition}

\begin{remark} 
Notice that the polynomial $g_z^{(i)}$ in Proposition \ref{prop:bintreemult} is a degree three polynomial in each variable of $p$. When applying the sum-check protocol to $g_z^{(i)}$, the prover therefore needs to send 4 field elements per round.

 In the case of Proposition \ref{prop:bintreemult}, the line $\ell: \mathbb{F} \rightarrow \mathbb{F}^{2_{i+1}}$ in the ``Reducing to Verification of a Single Point'' step has an especially simple expression.
Let $r \in \mathbb{F}^{s_i}$ be the vector of random field elements chosen by $\cV$ over the execution of the sum-check protocol. Notice that $\ell(0)$ must equal the point $(r, 0)\in \mathbb{F}^{s_i + 1}$
 i.e., the point whose first $s_i$ coordinates equal $r$ and whose last coordinate equals 0. Similarly, $\ell(1)$ must equal $(r, 1)$. 
 We may therefore express the line $\ell$ via the equation
$\ell(t) = (r, t)$.
In this case, $\tilde{V}_{i+1} \circ \ell$ has degree 1 and is implicitly specified when $\cP$ sends the claimed values of $\tV_{i}(r, 0)$ and $\tV_{i}(r, 1)$.
\end{remark}

The case of a binary tree of addition gates is similar to the case of multiplication gates. 

\begin{proposition} \label{prop:bintreeadd}
Let $C$ be a circuit consisting of a binary tree of addition gates. Then $\tilde{V}_i(z) = \sum_{p \in \{0, 1\}^{s_i}} g_z^{(i)}(p)$, where  $g_z^{(i)}(p) = \beta_{s_i}(z, p) \left(\tV_{i+1}(p, 0) + \tV_{i+1}(p, 1)\right).$
\end{proposition}

\begin{remark}
The polynomial $g_z^{(i)}$ of Proposition \ref{prop:bintreeadd} has degree 2 in all variables, rather than degree 3 as in Proposition \ref{prop:bintreemult}.
\end{remark}

\subsubsection{The Polynomials for \distinct}
\label{sec:distinctpolys}
We now describe a circuit $C$ for computing the number of non-zero entries of a vector $a \in \mathbb{F}^n$ (this vector should be interpreted as the \emph{frequency vector} of a data stream). 
A similar circuit was used in conjunction with the GKR protocol in \cite{itcs}  
to yield an efficient protocol with a streaming verifier for \distinct, and we borrow heavily from
the presentation there. 
We remark that our refinements  enable us to slightly simplify the circuit used in \cite{itcs} by 
avoiding the awkward use of a constant-valued input wire with value set to 1. This causes some gates in our circuit to have fan-in 1 rather than fan-in 2, which is easily supported by our protocol.

The circuit $C$ is tailored for use over the field of cardinality equal to a Mersenne prime $q=2^{k}-1$ for some $k$. Fields of cardinality
equal to a Mersenne prime can support extremely fast arithmetic, and
as discussed later in Section \ref{sec:experimentaldetails}, there are several Mersenne primes of appropriate magnitude for use within our protocols.

The circuit $C$ exploits Fermat's Little Theorem, computing $a_i^{q-1}$ for each input entry $a_i$ before summing the results. As described in \cite{itcs}, verifying
the summation sub-circuit can be handled with a one invocation of
the sum-check protocol, or less efficiently by running our protocol for a
binary tree of addition gates described in Proposition \ref{prop:bintreeadd}.

We now turn to describing the part of the circuit computing $a_i^{q-1}$ for each input entry $a_i$. We may write $q-1 = 2^{k}-2$, whose binary representation is $k-1$ 1s followed by a 0. Thus, $a_i^{q-1} = \prod_{j=1}^{k-1} a_i^{2^j}$.
To compute $a_i^{q-1}$, the circuit repeatedly squares $a$, and multiplies together the results ``as it goes''. In more detail, for $j>2$ there are two multiplication gates at each layer $d(n)-j$ of the circuit for computing $a_i^{q-1}$; the first computes $a^{2^j}$ by squaring the corresponding gate at layer $j-1$, and the second computes $\prod_{\ell=1}^{j-1} a_i^{2^{\ell-1}}$. See Figure \ref{fig:F0circ} for a depiction.

\begin{figure}
\centering
\includegraphics[width=2in]{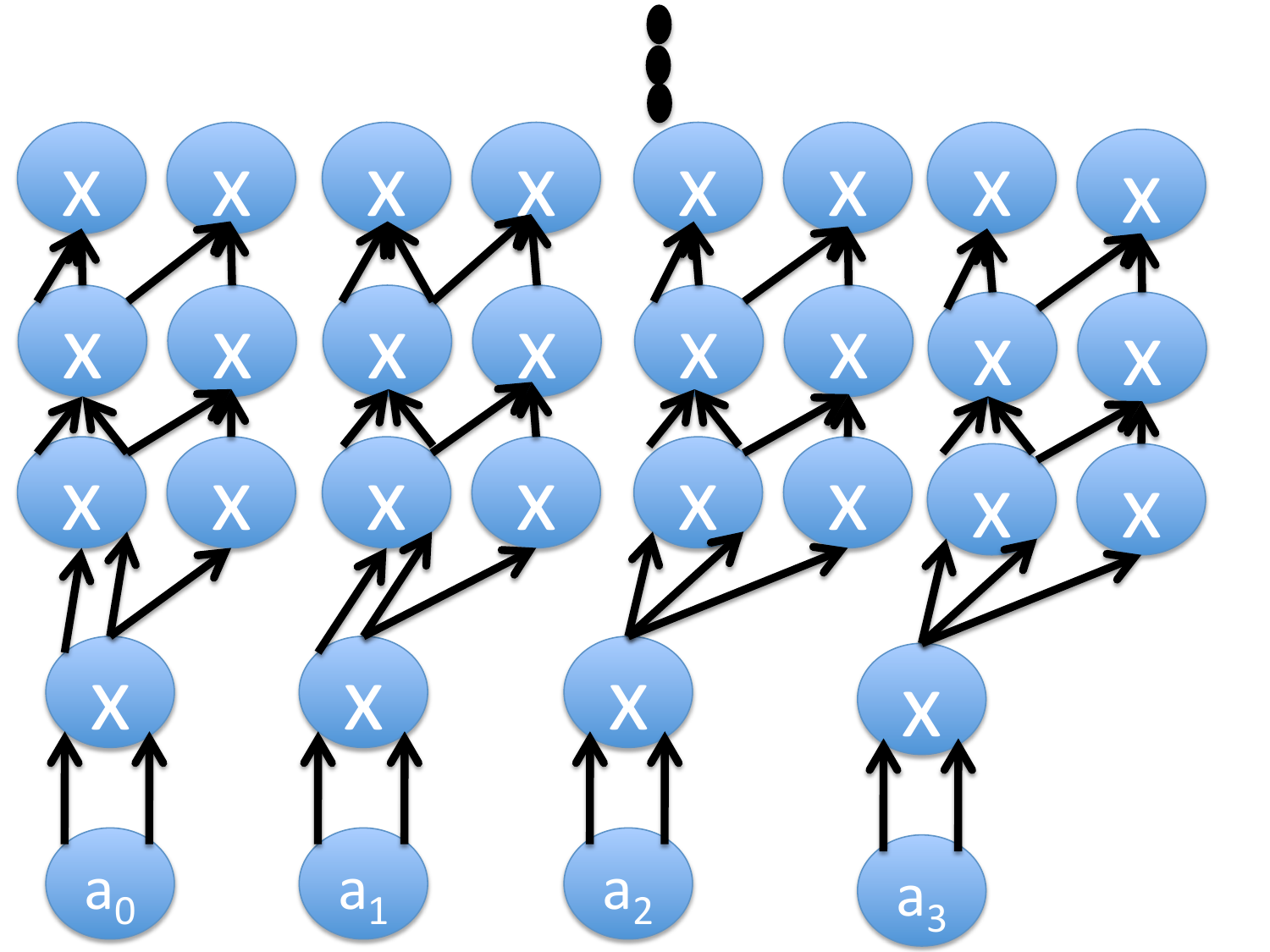}
\caption{The first several layers of a circuit for $F_0$ on four inputs over the field $\mathbb{F}$ with $q=2^{k}-1$ elements.
The first layer from the bottom computes $a_i^2$ for each input entry $a_i$. The second layer from the bottom computes $a_i^4$ and $a_i^2$ for all $i$.
The third layer computes $a_i^8$ and $a_i^6 = a_i^4 \times a_i^2$, while the fourth layer computes $a_i^{16}$ and $a_i^{14} = a_i^8 \times a_i^6$. The remaining layers (not shown) have structure identical to the third and fourth layers until the value $a_i^{q-1}$ is computed for all $i$, and the circuit culminates in a binary tree of addition gates.}
\label{fig:F0circ}
\end{figure}

For our purposes there are $k+1$ relevant circuit layers, all of which consist entirely of multiplication gates.
Layers 1 through $k-1$ all contain $2n$ gates. 
 Number the gates from $0$ to $2n-1$ in the natural way. In what follows, we will abuse notation and use $p$ to refer to both a gate number as well as its
 binary representation.
 
An even-numbered gate $p$ at layer $i$ has both in-wires connected to gate $p$ at layer $i+1$,
while an odd-numbered gate $p$ has one in-wire connected to gate $p$ and another connected to gate $p-1$. 
Thus, the connectivity information of the circuit is a simple function of the binary representation $p$ of each gate at layer $i$. If the low-order bit $p_{s_i}$ of $p$ is $0$ (i.e., it is an even-numbered gate), 
then both in-neighbors at layer $i+1$ of gate $p$ have binary representation $p$. If the low-order bit $p_{s_i}$ is 1 (i.e., it is an odd-numbered gate), then the first in-neighbor of gate $p$ has binary representation 
$p$, and the second has binary representation $(p_{-s_i}, 0)$, where $p_{-s_i}$ denotes $p$ with the coordinate $p_{s_i}$ removed.

Invoking Lemma \ref{prop:big}, the following proposition is easily verified.

\begin{proposition}
\label{prop:F0main}
Let $C$ be the circuit described above. For layers $i \in \{1, \dots, k-1\}$,
$\tV_i(z) = \sum_{p \in \{0, 1\}^{s_i}} g_z^{(i)}(p)$
where 
$$g_z^{(i)}(p) = \beta_{s_i}(z, p) \left((1-p_{s_i}) \tV_{i+1}(p_{-s_i}, 0) \cdot \tV_{i+1}(p_{-s_i}, 0) + p_{s_i} \tV_{i+1}(p_{-s_{i}}, 1) \cdot \tV_{i+1}(p_{-s_i}, 0)\right),$$

where $p_{-s_i}$ denotes $p$ with the coordinate $p_{s_i}$ removed. 
\end{proposition}

\begin{remark} 
To check $\cP$'s claim in the final round of the sum-check protocol applied to $g_z^{(i)}$,
$\cV$ needs to know $\tV_{i+1}(r, 0)$ and $\tV_{i+1}(r, 1)$ for some random vector $r \in \mathbb{F}^{s_i-1}$. This is identical to the situation in the case of a binary tree of addition or multiplication gates,
where the ``Reducing to Verification of a Single Point'' step had an especially simple implementation.
\end{remark}

At layer $k$, an even-numbered gate $p$ has both in-wires connected to gate $p/2$ at layer $k+1$, while an odd-numbered gate $p$ has its unique in-wire connected to gate $(p-1)/2$ at layer $k+1$.
Thus, for a gate at layer $i=k$, if the the low-order bit $p_{s_i}$ of the gate's binary representation $p$ is $1$ (i.e., it is an odd-numbered gate), 
then both in-neighbors at layer $i+1$ of have binary representation $p_{-s_i}$. If the low-order bit $p_{s_i}$ is 0 (i.e., it is an even numbered gate), then the unique in-neighbor of $p$ at layer $i+1$ has binary representation 
$p_{-s_i}$. 

Invoking Lemma \ref{prop:big}, the following is easily verified.

\begin{proposition}
\label{prop:F061}
Let $C$ be the circuit described above. For layer $i=k$,
$\tV_{i}(z) = \sum_{p \in \{0, 1\}^{s_i}} g_z^{(i)}(p)$
where 
$$g_z^{(i)}(p) =  \beta_{s_i}(z, p) \left((1-p_{s_i}) \tV_{i+1}(p_{-s_i}) \cdot \tV_{i+1}(p_{-s_i}) + p_{s_i} \tV_{i+1}(p_{-s_i})\right),$$
where $p_{-s_i}$ denotes $p$ with coordinate $p_{s_i}$ removed. 
\end{proposition}

 Finally, at layer $k+1$, each gate $p$ has both in-wires connected to gate $p$ at layer $k+2$ (which is the input layer). Thus:
 
\begin{proposition}
\label{prop:F062}
Let $C$ be the circuit described above. For layer $i=k+1$,
$\tV_{i}(z) = \sum_{p \in \{0, 1\}^{s_i}} g_z^{(i)}(p)$
where 
$$g_z^{(i)}(p) =  \beta_{s_i}(z, p) \tV_{i+1}(p) \cdot \tV_{i+1}(p).$$
\end{proposition}

\subsection{Reusing Work}
\label{sec:algs}
Recall that our analysis of the costs of the sum-check protocol in Section \ref{sec:sumcheckcosts} revealed that, when applying a sum-check
protocol to an $s_i$-variate polynomial $g_z^{(i)}$,  $\cP$ only needs to evaluate $g_z^{(i)}$ at $O(2^{s_i})$ points across all rounds of the protocol.
Our goal in this section is to show how
 $\cP$ can do this in time  $O(2^{s_i} + 2^{s_{i+1}}) = O(S_i + S_{i+1})$ for all of the polynomials $g_z^{(i)}$ described in Section \ref{sec:polynomials}. This is 
sufficient to ensure that $\cP$ takes $O(\sum_{i=1}^{d(n)} S_i) = O(S(n))$ time across all iterations of our circuit-checking protocol.

To this end, notice that all of the polynomials $g_z$ described in Propositions \ref{prop:bintreemult}-\ref{prop:F062} have the following
property: for any $r \in \mathbb{F}^{s_i}$, evaluating $g_z^{(i)}(r)$ can be done in constant time given $\beta(z, r)$ and the
evaluations of $\tilde{V}_{i+1}$ at a constant number of points. For example, consider the polynomial $g_z^{(i)}$ described in Proposition \ref{prop:F0main}:
$g_z^{(i)}(r)$ can be computed in constant time given $\beta_{s_i}(z, r)$, $\tV_{i+1}(r_{-s_i}, 0)$, and $\tV_{i+1}(r_{-s_i}, 1)$.

Moreover, the points at which $\cP$ must evaluate $g_z^{(i)}$ within the sum-check protocol are highly structured:
in round $j$ of the sum-check protocol,
the points are all of the form $(r_1, \dots, r_{j-1}, t, b_{j+1}, \dots, b_{s_i})$ with $t \in \{0, 1, \dots, \deg_j(g_z^{(i)})\}$ and $(b_{j+1}, \dots, b_{s_i}) \in \{0, 1\}^{s_i-j}$. 

\subsubsection{Computing the Necessary $\beta(z, p)$ Values}
\label{sec:betas}
\noindent \textit{Pre-processing.} We begin by explaining how $\cP$ can, in $O(2^{s_i})$ time, compute an array $C^{(0)}$ of length $2^{s_i}$ of all values $\beta(z, p) = \prod_{k=1}^{s_i} (p_k z_k + (1-p_k) (1-z_k))$ for $p \in \{0, 1\}^{s_i}$. 
$\cP$ can do this computation in preprocessing before the sum-check protocol begins, as this computation does not depend on any of $\cV$'s messages. 
Naively, computing all entries of $C^{(0)}$
would require $O(s_i 2^{s_i})$ time, as there are $2^{s_i}$ values to compute, and each involves $\Omega(s_i)$ multiplications. However, this can be improved using dynamic programming.

The dynamic programming algorithm proceeds in stages. In stage $j$, $\cP$ computes an array $C^{(0, j)}$ of length $2^j$. Abusing notation, we identify a number $p$ in $[2^j]$ with its binary representation in $\{0, 1\}^j$.
$\cP$ computes $$C^{0, j}[p] = \prod_{k=1}^{j}  (p_k z_k + (1-p_k) (1-z_k))$$ 
via the recurrence $$C^{0, j}[(p_1, \dots, p_j)] = C^{0, j-1}[(p_1, \dots, p_{j-1})] \cdot (p_j z_j + (1-p_j)(1-z_j)).$$ 
Clearly $C^{(0, s_i)}$ equals the desired array $C^{(0)}$, and the total number of multiplications
required over the entire procedure is $O(\sum_{j=1}^{s_i} 2^j ) = O(2^{s_i})$.
We remark that our dynamic programming procedure is similar to the method used by Vu \etal\ to reduce the verifier's runtime in the GKR protocol from $O(n \log n)$ to $O(n)$ in Lemma \ref{lemma:vu}.

\medskip
\noindent \textit{Overview of Online Processing.}
In round $j$ of of the sum-check protocol, $\cP$ needs to evaluate the polynomial $\beta(z, p)$ at $O(2^{s_i-j})$ points of the form $(r_1, \dots, r_{j-1}, t, b_{j+1}, \dots, b_{s_i})$ for $t \in [\deg_j(g_z^{(i)})]$ and
$( b_{j+1}, \dots, b_{s_i}) \in \{0, 1\}^{s_i-j}$. $\cP$ will do this using the help of intermediate arrays $C^{(j)}$ defined as follows.

Define $C^{(j)}$ to be the array of length $2^{s_i-j}$ such that for $(p_{j+1}, \dots, p_{s_i}) \in \{0, 1\}^{s_i-j}$:

$$C^{(j)}[(p_{j+1}, \dots, p_{s_i})] = \left( \prod_{k=1}^{j} (r_k z_k + (1-r_k)(1-z_k))\right) \cdot \left( \prod_{k=j+1}^{s_i}(p_k z_k + (1-p_k)(1-z_k))\right),$$

\medskip \noindent \textit{Efficiently Constructing $C^{(j)}$ Arrays.} Inductively, assume $\cP$ has computed the array $C^{(j-1)}$  in the previous round.
As the base case, we explained how $\cP$ can evaluate $C^{(0)}$ in $O(2^{s_i})$ time in pre-processing. Now observe that $\cP$ can compute
$C^{(j)}$ given $C^{(j-1)}$ in $O(2^{s_i-j})$ time using the following recurrence:

 \begin{equation} \label{eq:therecurrence} C^{(j)}[(p_{j+1}, \dots, p_{s_i})]  = z_j^{-1}C^{(j-1)}[(1, p_{j+1}, \dots, p_{s_i})] \cdot (r_j z_j + (1-r_j)(1-z_j)).\end{equation}

\begin{remark} Equation \eqref{eq:therecurrence} is only valid when $z_j \neq 0$. To avoid this issue, we can have $\cV$ choose $z_j$ at random from $\mathbb{F}^*$ rather than from $\mathbb{F}$, and this
will affect the soundness probability by at most an additive $O(d(n) \cdot \log S(n)/|\mathbb{F}|)$ term. 
\end{remark}

\begin{remark} Since computing multiplicative inverses in a finite field is not a constant-time operation, it is important to note that $z_j^{-1}$ only needs to be computed once when determining the entries of $C^{(j)}$, i.e., it need not be recomputed for each entry of $C^{(j)}$. 
Therefore, across all $s_i$ rounds of the sum-check protocol, only $\tilde{O}(s_i)$ time in total is required to compute these multiplicative inverses, which does not affect the asymptotic costs for $\cP$.
We discount the costs of computing $z_j^{-1}$ for the remainder of the discussion.
\end{remark}

Thus, at the end of round $j$ of the sum-check protocol, when $\cV$ sends $\cP$ the value $r_j$, $\cP$ can compute $C^{(j)}$ from $C^{(j-1)}$ using Equation \eqref{eq:therecurrence} in $O(2^{s_i-j})$ time.
 
\medskip \noindent \textit{Using the $C^{(j)}$ Arrays.}
Observe that given any point of the form $p=(r_1, \dots, r_{j-1}, t, b_{j+1}, \dots, b_{s_i})$  with
$(b_{j+1}, \dots, b_{s_i}) \in \{0, 1\}^{s_i-j}$, $\beta(z, p)$ can be evaluated in constant time using the array $C^{(j-1)}$, using the equality
$$\beta(z, p) = C^{(j-1)}[(1, p_{j+1}, \dots, p_{s_i})] \cdot z_j^{-1} \cdot (t z_j + (1-t) (1-z_j)).$$

As above, note that $z_j^{-1}$ can be computed just once and used for all points $p$, and this does not affect the asymptotic costs for $\cP$.

\medskip \noindent \textit{Putting Things Together.}
In round $j$ of the sum-check protocol, $\cP$ uses the array $C^{(j-1)}$ to evaluate the $O(2^{s_i-j})$ required $\beta(z, p)$ values in $O(2^{s_i-j})$ time. At the end of round $j$, $\cV$ sends $\cP$ the value $r_j$, and 
$\cP$ computes $C^{(j)}$ from $C^{(j-1)}$ in $O(2^{s_i-j})$ time. 
In total across all rounds of the sum-check protocol, $\cP$ 
spends $O(\sum_{j=1}^{s_i} 2^{s_i-j}) = O(2^{s_i})$ time to compute the $\beta(z, p)$ values.

\subsubsection{Computing the Necessary $\tilde{V}_{i+1}(p)$ Values}
\label{sec:Vs}

For concreteness and clarity, we restrict our presentation within this subsection to the polynomial $g_z^{(i)}$ described in Proposition \ref{prop:F0main}. 
Theorem \ref{thm:generaltheorem} abstracts this analysis into a general result
capturing a large class of wiring patterns. 

Recall that all of the polynomials $g_z^{(i)}$ described in Propositions \ref{prop:bintreemult}-\ref{prop:F062} have the following
property: for any $p \in \mathbb{F}^{s_i}$, evaluating $g_z^{(i)}(p)$ can be done in constant time given $\beta(z, p)$ and the
evaluations of $\tilde{V}_{i+1}$ at a constant number of points. We have already shown how $\cP$ can evaluate all of the necessary $\beta(z, p)$ values in $O(2^{s_i})$ time.
It remains to show how $\cP$ can evaluate all of the $\tilde{V}_{i+1}$ values in time $O(2^{s_i} + 2^{s_{i+1}})$. We remark that in the context of Proposition \ref{prop:F0main}, $s_{i} = s_{i+1}$; however, 
we still distinguish between these
two quantities throughout this subsection in order to ensure maximal consistency with the general derivation of Theorem \ref{thm:generaltheorem}.

Recall that the polynomial $g_z^{(i)}$ in Proposition \ref{prop:F0main} was defined as follows:

$$g_z^{(i)}(p) = \beta_{s_i}(z, p) \left((1-p_{s_i}) \tV_{i+1}(p_{-s_i}, 0) \cdot \tV_{i+1}(p_{-s_i}, 0) + p_{s_i} \tV(p_{-s_i}, 1) \cdot \tV(p_{-s_i}, 0)\right).$$

In round $j$ of the sum-check protocol, $\cP$ needs to evaluate $g_z$ at all points in the set 
$$S^{(j)} = \{(r_1, \dots, r_{j-1}, t, b_{j+1}, \dots, b_{s_i}): t \in \{0, \dots, \deg_j(g_z^{(i)})\} \text{ and } (b_{j+1}, \dots, b_{s_i}) \in \{0, 1\}^{s_i-j}\}.$$
By inspection of $g_z^{(i)}$, it suffices for $\cV$ to evaluate $\tV_{i+1}$ at the same set of points. To show how to accomplish this efficiently,
we exploit the following explicit expression for $\tV_{i+1}$. This expression was derived for the case $i+1=d$ in Equation \eqref{eq:linextnew} within Lemma \ref{lemma:streamingv};
we re-derive it here in the general case. 

For a vector $b \in\{0,1\}^{s_{i+1}}$ let $\chi_b(x_1, \dots, x_{s_{i+1}})=\prod_{k=1}^{s_{i+1}} \chi_{b_k}(x_k)$, where $\chi_{0}(x_k)=1-x_k$ and $\chi_1(x_k) = x_k$. 
With this definition in hand, we may write:

\begin{equation} \label{eq:linext} \tV_{i+1}(p_1, \dots, p_{s_{i+1}}) = \sum_{b \in \{0, 1\}^{s_{i+1}}} V_{i+1}(b) \chi_b(p_1, \dots p_{s_{i+1}}), \end{equation}

To see that Equation \eqref{eq:linext} holds,
notice that the right hand side of Equation \eqref{eq:linext} is a multilinear polynomial in the variables $(p_1, \dots, b_{p_{i+1}})$, and that
it agrees with $V_{i+1}$ at all points $p \in \{0, 1\}^{s_{i+1}}$. Hence, it must be the unique
multilinear extension of $V_{i+1}$.

The intuition behind our optimizations is the following. In round $j$ of the sum-check protocol,
there are $|S^{(j)}|$ points at which $\tilde{V}_{i+1}$ must be evaluated. Equation \eqref{eq:linext} can be exploited to show that each gate at layer $i+1$ of the circuit 
contributes to $\tilde{V}_{i+1}(p)$ for at most one point  $p \in S^{(j)}$; namely the point $p$ whose last $s_{i+1}-j$ coordinates agrees with those
of $p$. This observation alone is enough to achieve an $O(S_{i+1} \log S_{i})$ runtime for $\cP$ in total across all iterations of the sum-check protocol, because
there are $S_{i+1}$ gates at layer $i+1$, and only $s_i = \log S_i$ rounds of the sum-check protocol. 
However,  
we need to go further in order to shave off the last $\log S_i$ factor from $\cP$'s runtime. Essentially, what we do is group the gates at layer $i+1$ by the point $p \in S^{(j)}$ to which they
contribute. Each such group can be treated as a single unit, ensuring that the work $\cP$ has to do in any round of the sum-check protocol in order to evaluate
$\tV_{i+1}$ at all points in $S^{(j)}$ is proportional to $|S^{(j)}|$ rather than to $S_{i+1}$. Since the size of $S^{(j)}$ falls geometrically with $j$, our desired
time bounds follow.

\medskip \noindent 
\textit{Pre-processing.} $\cP$ will begin by computing an array $V^{(0)}$, which is simply defined to be the vector of gate values at layer $i+1$,
i.e., identifying a number $0 < j < S_{i+1} $ with its binary representation in $\{0, 1\}^{s_{i+1}}$, $\cP$
sets $V^{(0)}[(j_1, \dots, j_{s_{i+1}})] = V_{i+1}(j_1, \dots, j_{s_{i+1}})$ for each $(j_1, \dots, j_{s_{i+1}}) \in \{0, 1\}^{s_{i+1}}$. 
The right hand side of this equation is simply the value of the $j$th gate at layer $i+1$ of $C$. 
So $\cP$ can fill in the array $V^{(0)}$ when she evaluates the circuit $C$, before receiving any messages from $\cV$.

 \medskip
\noindent \textit{Overview of Online Processing.}
In round $j$ of of the sum-check protocol, $\cP$ needs to evaluate the polynomial $\tV_{i+1}$ at the $O(2^{s_{i}-j})$ points in the set $S^{(j)}$. 
$\cP$ will do this using the help of intermediate arrays $V^{(j)}$ defined as follows. 

Define $V^{(j)}$ to be the length $2^{s_{i+1}-j}$ array such that for $(p_{j+1}, \dots, p_{s_{i+1}}) \in \{0, 1\}^{s_{i+1}-j}$,

$$V^{(j)}[(p_{j+1}, \dots, p_{s_{i+1}})] = \sum_{(b_1, \dots, b_{j}) \in \{0, 1\}^{j}} V_{i+1}(b_1, \dots, b_{j}, p_{j+1}, \dots, p_{s_{i+1}}) \cdot \prod_{k=1}^{j} \chi_{b_k}(r_k),$$

\medskip \noindent \textit{Efficiently Constructing $V^{(j)}$ Arrays.} Inductively, assume $\cP$ has computed in the previous round the array $V^{(j-1)}$ of length $2^{s_{i+1}-j+1}$.


As the base case, we explained how $\cP$ can fill in $V^{(0)}$ in the process of evaluating the circuit $C$.  
Now observe that $\cP$ can compute
$V^{(j)}$ given $V^{(j-1)}$ in $O(2^{s_{i+1}-j})$ time using the following recurrence:

 $$V^{(j)}[(p_{j+1}, \dots, p_{s_{i+1}})]  = V^{(j-1)}[(0, p_{j+1}, \dots, p_{s_i})] \cdot \chi_0(r_j) + V^{(j-1)}[(1, p_{j+1}, \dots, p_{s_i})] \cdot \chi_1(r_j).$$

Thus, at the end of round $j$ of the sum-check protocol, when $\cV$ sends $\cP$ the value $r_j$, $\cP$ can compute $V^{(j)}$ from $V^{(j-1)}$ in $O(2^{s_{i+1}-j+1})$ time. 
 
\medskip \noindent \textit{Using the $V^{(j)}$ Arrays.}
We now show how to use the array $V^{(j-1)}$ to evaluate $\tV_{i+1}(p)$ in constant time for any point of the form $p=(r_1, \dots, r_{j-1}, t, b_{j+1}, \dots, b_{s_{i+1}})$  with
$(b_{j+1}, \dots, b_{s_{i+1}}) \in \{0, 1\}^{s_{i+1}-j}$. We exploit the following sequence of equalities:

\begin{align*}\!\!\!\!\!\!\!\!\!\!\!\!\!\!\!\!\!\!\!\! \tilde{V}_{i+1}(r_1, \dots, r_{j-1}, t, b_{j+1}, \dots, b_{s_i}) & =  \sum_{c \in \{0, 1\}^{s_{i+1}}} V_{i+1}(c) \chi_c(r_1, \dots, r_{j-1}, t, b_{j+1}, \dots, b_{s_{i+1}})\\
\!\!\!\!\!\!\!\!\!\!\!\!\!\!\!\!\!\!\!\! & =   \sum_{(c_1, \dots, c_j) \in \{0, 1\}^{j}} \sum_{(c_{j+1}, \dots, c_{s_{i+1}}) \in \{0, 1\}^{s_{i+1} -j}} V_{i+1}(c) \chi_c(r_1, \dots, r_{j-1}, t, b_{j+1}, \dots, b_{s_{i+1}}) \\
\!\!\!\!\!\!\!\!\!\!\!\!\!\!\!\!\!\!\!\! & =  \sum_{(c_1, \dots, c_j) \in \{0, 1\}^{j}} \sum_{(c_{j+1}, \dots, c_{s_{i+1}}) \in \{0, 1\}^{s_{i+1} -j}} V_{i+1}(c)  \left(\prod_{k=1}^{j-1} \chi_{c_k}(r_k) \right) \left(\chi_{c_j}(t) \right) \left(\prod_{k=j+1}^{s_{i+1}} \chi_{c_k}(b_k)\right)\\
\!\!\!\!\!\!\!\!\!\!\!\!\!\!\!\!\!\!\!\!  & = \sum_{(c_1, \dots, c_j) \in \{0, 1\}^{j}} V_{i+1}(c_{j+1}, \dots, c_{j}, b_{j+1},  \dots, b_{s_{i+1}})  \left(\prod_{k=1}^{j-1} \chi_{c_k}(r_k) \right) \cdot \chi_{c_j}(t)\\
\!\!\!\!\!\!\!\!\!\!\!\!\!\!\!\!\!\!\!\!  & = V^{(j-1)}[(0, b_{j+1}, \dots, b_{s_{i+1}})] \cdot \chi_{0}(t) +  V^{(j-1)}[(1, b_{j+1}, \dots, b_{s_{i+1}})] \cdot \chi_{1}(t).
 \end{align*}
 
 Here, the first equality holds by Equation \eqref{eq:linext}. The third holds by definition of the function $\chi_c$. The fourth holds because for Boolean values $b_k, c_k \in \{0, 1\}$, $\chi_{c_k}(b_k)=1$
 if $c_k=b_k$, and $\chi_{c_k}(b_k)=0$ otherwise. The final equality holds by definition of the array $V^{(j-1)}$.

\medskip \noindent \textit{Putting Things Together.}
In round $j$ of the sum-check protocol, $\cP$ uses the array $V^{(j-1)}$ to evaluate $\tV_{i+1}(p)$ for all $O(2^{s_i-j})$ points $p \in S^{(j)}$. This requires constant
time per point, and hence $O(2^{s_{i}-j})$ time across all points.
At the end of round $j$, $\cV$ sends $\cP$ the value $r_j$, and 
$\cP$ computes $V^{(j)}$ from $V^{(j-1)}$ in $O(2^{s_{i+1}-j})$ time. 
In total across all rounds of the sum-check protocol, $\cP$ 
spends $O(\sum_{j=1}^{s_i} 2^{s_i-j} + 2^{s_{i+1}-j}) = O(2^{s_i} + 2^{s_{i+1}})$ time to evaluate $\tV_{i+1}$ at the relevant points.
When combined with our $O(2^{s_i})$-time algorithm for computing all the relevant $\beta(z,p)$ values, we see $\cP$ 
takes $O(2^{s_i} + 2^{s_{i+1}}) = O(S_i + S_{i+1})$ time to run the entire sum-check protocol for iteration $i$ of our circuit-checking protocol.

\subsection{A General Theorem}
In this section we formalize a large class of circuits to which our refinements yield asymptotic savings relative to prior implementations
of the GKR protocol. 
Our protocol makes use of the following functions that capture the wiring structure of an arithmetic circuit $C$.

\begin{definition} Let $C$ be a layered arithmetic circuit of depth $d(n)$ and size $S(n)$ over finite field $\mathbb{F}$. For
every $i \in \{1, \dots, d-1\}$,  let $\text{in}_1^{(i)}: \{0, 1\}^{s_i} \rightarrow \{0, 1\}^{s_{i+1}}$ and 
$\text{in}_2^{(i)}: \{0, 1\}^{s_i} \rightarrow \{0, 1\}^{s_{i+1}}$ denote the functions
that take as input the binary label $p$ of a gate at layer $i$ of $C$, and output the binary label of the first and second in-neighbor of 
gate $p$ respectively. Similarly, let $\text{type}^{(i)}:  \{0, 1\}^{s_i} \rightarrow \{0, 1\}$ denote the function that takes as input 
the binary label $p$ of a gate at layer $i$ of $C$, and outputs 0 if $p$ is an addition gate, and 1 if $p$ is a multiplication gate.
\end{definition}

Intuitively, the following definition captures functions whose outputs are simple bit-wise 
transformations of their inputs.

\begin{definition} \label{def:regular} Let $f$ be a function mapping  $\{0, 1\}^{v}$ to $\{0, 1\}^{v'}$. Number the $v$ input bits from $1$ to $v$, and the $v'$ output bits from $1$ to $v'$.
Assume that one machine word contains $\Omega(v + v')$
bits.
We say that $f$
is \emph{regular} if $f$ can be evaluated on any input in constant time, and there is a subset of input bits $\mathcal{S} \subseteq [v]$ with $|\mathcal{S}| = O(1)$ such that:
\begin{enumerate}
\item Each input bit in $[v] \setminus \mathcal{S}$ affects $O(1)$ of the output bits of $f$. Moreover, 
given input $j \in [v] \setminus \mathcal{S}$, the set $\mathcal{S}_j$ of output bits affected by $x_j$ can be enumerated
in constant time.
\item Each output bit of $f$ depends on at most one input bit. 
\end{enumerate}
\end{definition}

Our protocol applied to $C$ proceeds in $d(n)$ iterations, where iteration $i$ consists an application of the sum-check protocol
to an appropriate polynomial derived from $\text{type}^{(i)}$, $\text{in}_1^{(i)},$ and  $\text{in}_2^{(i)},$ followed by a phase for ``reducing to
verification of a single point''.  
For any layer $i$ of $C$ such that $\text{in}_1^{(i)},$ $\text{in}_2^{(i)}$ and $\text{type}^{(i)}$ are all regular, 
we can show that $\cP$ can execute the sum-check protocol at iteration $i$ in $O(S_i + S_{i+1})$ time. 
To ensure that $\cP$ can execute the ``reducing to verification of a single point'' phase in $O(S_{i+1})$ time, we need to place one additional condition
on $\text{in}_1^{(i)}$ and $\text{in}_2^{(i)}$.

\begin{definition} \label{def:similar} We say that $\text{in}_1^{(i)}$ and $\text{in}_2^{(i)}$ are \emph{similar} if there is a set of output bits $\mathcal{T} \subseteq [s_{i+1}]$ with $|\mathcal{T}| = O(1)$ such that
for all inputs $x$, the $j$th output bit of $in_1^{(i)}$ equals the $j$th output bit of $in_2^{(i)}$ for all $j \in [s_{i+1}] \setminus \mathcal{T}$.
\end{definition}

We are finally in a position to state the class of circuits to which our refinements apply.

\begin{theorem} \label{thm:generaltheorem}
Let $C$ be an arithmetic circuit, and suppose that for all layers $i$ of $C$, $\text{in}_1^{(i)}$, $\text{in}_2^{(i)}$, and $\text{type}^{(i)}$ are regular.
Suppose moreover that $\text{in}_1^{(i)}$ is similar to $\text{in}_2^{(i)}$ for all but $O(1)$ layers $i$ of $C$. Then 
there is a valid interactive proof protocol $(\cP, \cV)$ for the function computed by $C$, with the following costs.
The total communication cost is $|\mathcal{O}| + O(d(n)\log S(n))$ field elements, where $|\mathcal{O}|$ is the number of outputs of $C$.  
The time cost to $\cV$ is $O(n \log n + d(n) \log S(n))$,
and $\cV$ can make a single streaming pass over the input, storing $O(\log(S(n)))$ field elements. The 
time cost to $\cP$ is $O(S(n))$.
\end{theorem}

The asymptotic costs of the protocol whose existence is guaranteed by Theorem \ref{thm:generaltheorem} are identical to those
of the implementation of the GKR protocol due to Cormode \etal\ in \cite{itcs}, except that in Theorem \ref{thm:generaltheorem} $\cP$ runs in time $O(S(n))$ rather than $O(S(n) \log S(n))$ as achieved by \cite{itcs}.
We defer the proof to Appendix \ref{app:general}.

\subsubsection{Applications}
\label{sec:applications}
Theorem \ref{thm:generaltheorem} applies to circuits computing functions from a wide range of applications, with the following implications.

\medskip
\noindent \textbf{\matmult.} 
Consider the following circuit $C$ of size $O(n^3)$ for multiplying two $n\times n$ matrices $A$ and $B$. 
Let the input gate labelled $(0, i, j)$ correspond to $A_{ij}$, and the input labelled $(1, i, j)$ correspond to $B_{ij}$. 
The layer of $C$ adjacent to the input consists of $n^3$ gates, where the gate labeled $(i, j, k) \in (\{0, 1\}^{\log n})^3$ computes $A_{ik} \cdot B_{kj}$. 
All subsequent layers constitute a binary tree of addition gates summing up the results and thereby computing $\sum_{k} A_{ik} B_{kj}$ for all $(i, j) \in [n] \times [n]$. 

For layers $i \in \{1, \dots, \log n\}$ of this circuit, $\text{in}_1^{(i)}$, $\text{in}_2^{(i)}$,  and $\text{type}^{(i)}$ are all regular, and moreover
$\text{in}_1^{(i)}$ is similar to $\text{in}_2^{(i)}$ (see Section \ref{sec:polyforbintree} for a careful treatment of this wiring pattern).
The remaining layer of the circuit, layer $i=\log n + 1$, is regular, though $\text{in}_1^{(\log n + 1)}$ and $\text{in}_2^{(\log n + 1)}$ are not similar. 
We obtain the following immediate corollary.

\begin{corollary} \label{corr:matmult} There is a valid interactive proof protocol for $n \times n$ \matmult\ with the following costs.
The total communication cost is $n^2 + O(d(n)\log n)$ field elements, where the $n^2$ term is required to specify the answer.
The time cost to $\cV$ is $O(n^2 \log n)$,
and $\cV$ can make a single streaming pass over the input in time $O(n^2 \log n)$ and storing $O(\log n)$ field elements. The 
time cost to $\cP$ is $O(n^3)$.
\end{corollary}

We note that the costs of Corollary \ref{corr:matmult} are subsumed 
by our special-purpose matrix multiplication protocol presented later in Theorem \ref{thm:finalfinalthm}. We included Corollary \ref{corr:matmult} to demonstrate
the applicability of Theorem \ref{thm:generaltheorem}.

\medskip
\noindent \textbf{\distinct.} Recall the circuit $C$ over field size $q=2^{k}-1$
described in Section \ref{sec:distinctpolys} that takes a vector $a \in \mathbb{F}^n$ as input and outputs the number of non-zero entries of $a$.
This circuit has $k+1$ relevant layers and consists entirely of multiplication gates. 
For any layer $i \in [k-1]$, an even-numbered gate $p$ at layer $i$ has both in-wires connected to gate $p$ at layer $i+1$,
while an odd-numbered gate $p$ at layer $i$ has one in-wire connected to gate $p$ at layer $i+1$ and another connected to gate $p-1$ (which has binary
representation $(p_{-s_i}, 0)$,  where $p_{-s_i}$ denotes the binary representation of $p$ with the coordinate $p_{s_i}$ removed). 
For these layers, $\text{in}_1^{(i)}$, $\text{in}_2^{(i)}$, and $\text{type}^{(i)}$ are all regular, and  $\text{in}_1^{(i)}$ is similar to $\text{in}_2^{(i)}$.

At layer $k$, an even-numbered gate $p$ is has both in-wires connected to gate $p/2$ at layer $k+1$, while an odd-numbered gate $p$ at layer $k$ has its unique in-wire connected to gate $(p-1)/2$ at layer $k+1$.
In the former case, both in-neighbors of gate $p$ have binary representation $p_{-s_i}$. In the latter case the unique in-neighbor of gate $p$ has binary representation $p_{-s_i}$.  
It is therefore easily seen that $\text{in}_1^{(k)}$, $\text{in}_2^{(k)}$, and $\text{type}^{(k)}$ are all regular, and $\text{in}_1^{(k)}$ is similar to $\text{in}_2^{(k)}$. 
Finally, at layer $k+1$, both in-wires for gate $p$ are connected to gate $p$ at layer $k+2$. It is easily seen that $\text{in}_1^{(k+1)}, \text{in}_2^{(k+1)}$, and $\text{type}^{(k+1)}$
are all regular, and  $\text{in}_1^{(k+1)}$ is similar to $\text{in}_2^{(k+1)}$.
With all layers of $C$ satisfying the requirements of Theorem \ref{thm:generaltheorem}, we obtain the following corollary. 

\begin{corollary} \label{corr:distinct} Let $q>\max\{m, n\}$ be a Mersenne Prime. 
There is a valid interactive proof protocol over the field $\mathbb{F}_q$ for \distinct\ with the following costs.
The total communication cost is $O(\log n \log q)$ field elements.
The time cost to $\cV$ is $O(m \log n)$,
and $\cV$ can make a single streaming pass over the input, storing $O(\log n)$ field elements. The 
time cost to $\cP$ is $O(n \log q)$.
\end{corollary}

To or knowledge, Corollary \ref{corr:distinct} yields the fastest known prover of any streaming interactive proof protocol for \distinct\ that also has
total communication and space usage for $\cV$ that is sublinear in both $m$ and $n$. The fastest result previously was the $O\left(n \cdot \log(n) \cdot \log(p)\right)$-time prover
obtained by the implementation of Cormode \etal\ \cite{itcs}. We remark however that for a data stream with $F_0$ distinct items, the prover in \cite{itcs}
actually can be made to run in time $O\left(n + F_0 \cdot \log(n) \cdot \log(p)\right)$, where the $O(n)$ term is due to the time required to simply observe the entire input stream.
Therefore, for streams where $F_0 = o(n/\log n)$, the implementation of \cite{itcs} achieves 
an asymptotically faster prover than implied by Corollary \ref{corr:distinct}.

\begin{remark} Cormode \etal\ in \cite[Section 3.2]{itcs} describe how to extend the GKR protocol to handle circuits with gates that compute more general operations than just addition and multiplication. 
At a high level, \cite{itcs} shows that gates computing any ``low-degree'' operation can be handled, and they demonstrate analytically and experimentally
that these more general gates can achieve cost savings for the \distinct\ problem. These same optimizations are also applicable in conjunction 
with our refinements. We omit further details for brevity, and did not implement these optimizations in conjunction with our refinements. 
\end{remark}

\medskip \noindent 
\textbf{Other Problems.}
In order to demonstrate its generality, we describe two other non-trivial applications of Theorem \ref{thm:generaltheorem}.

\begin{itemize}
\item \textit{Pattern Matching}. In the Pattern Matching problem, the input consists of a stream of text $T = (t_0, \dots, t_{n-1}) \in [n]^n$ and pattern $P = (p_0, \dots, p_{m-1}) \in [n]^m$. The pattern $P$ is said to occur at location $i$ in $T$ if, for every position $k$ in $P$, $p_k = t_{i+k }$. The pattern-matching problem is to determine the number of locations at which $P$ occurs in $T$. 
For example, one might want to determine the number of times a given phrase appears in a corpus of emails stored in the cloud.

Cormode \etal\ describe the following circuit $C$ for Pattern Matching over the finite field $\mathbb{F}_q$. The circuit first computes 
the quantity $I_i = \sum_{j=0}^m (t_{i+j} - p_j)^2$ for each $i \in [[n]]$, and then exploits Fermat's Little Theorem (FLT) by computing $M=\sum_{i=1}^{n-m} I_i^{q-1}$.
The number of occurrences of the pattern equals $n-m-M$.

Computing $I_i$  for each $i$ can be done in $\log m+2$ layers: the layer closest to the input computes $t_{i+k}-p_k$ for each pair $(i, k) \in [[n]] \times [[q]]$, the next layer squares
each of the results, and the circuit then sums
the results via a depth $\log m$-binary tree of addition gates. The total size of the circuit $C$ is $O(nm + n \log q)$, where the $nm$ term is due to the computation of the $I_i$ values,
and the $n \log q$ term is due to the FLT computation. The total depth of the circuit is $O(\log m + \log q) = O(\log q)$.

We have already demonstrated that Theorem \ref{thm:generaltheorem} applies to the squaring layer, the binary tree sub-circuit, and the FLT computation. The only remaining layer of the circuit 
is the one that computes $t_{i+k}-p_k$ for each pair $(i, k) \in [[n]] \times [[m]]$. Unfortunately, Theorem \ref{thm:generaltheorem} does \emph{not} apply to this layer of the circuit.
This is because the first in-neighbor of a gate with label $(i_1, \dots, i_{\log n}, k_1, \dots, k_{\log m}) \in \{0, 1\}^{\log n + \log m}$ has label equal to the binary representation of the integer $i + k$,
and a single bit $i_j$ can affect many bits  in the binary representation of $i+k$ (likewise, each bit in the binary representation of $i+k$ may be affected by many bits in the binary representation of $i$ and $k$).

However, in Appendix \ref{app:patternmatch}, we describe how to extend the ideas underlying Theorem \ref{thm:generaltheorem} to handle this wiring pattern. The extensions in
Appendix \ref{app:patternmatch} may be more broadly useful, as the wiring pattern analyzed there is an instance of a common paradigm, in that it interprets binary gate 
labels as a pair of integers and performs a simple arithmetic operation (namely addition)
on those integers. 

We also remark that, instead of going through the analysis of Appendix \ref{app:patternmatch}, 
a more straightforward approach is to simply apply the implementation of \cite{itcs} to this layer; the runtime for $\cP$ in the corresponding sum-check protocol is $O(nm \log n)$.
This does not affect the asymptotic costs of the protocol if $m$ is constant, since in this case $nm \log n = O(n\log q)$, and the total runtime of $\cP$ over all other layers of the circuit is $\Theta(n \log q)$. 

This analysis highlights the following point: our refinements can be applied to a circuit on a layer-by-layer basis, so they can still yield speedups even if some but not all layers of
a circuit are sufficiently ``regular'' for our refinements to apply.

A similar analysis applies to a closely related circuit that solves a more general problem known as Pattern Matching with Wildcards. We omit these details for brevity.

\item \textit{Fast Fourier Transform.}  
Cormode \etal\ \cite{itcs} also describe a circuit over $\mathbb{C}$ for computing the standard radix-two decimation-in-time FFT. At a high level,
this circuit works as follows. It proceeds in $\log n$ stages, where for $k = (k_1, \dots, k_n) \in \{0, 1\}^n$, the $k$Õth output of stage $i$ is recursively defined as $V_i(k_1, \dots, k_n) = V_{i-1}(k_1,k_{i-1},0, k_i, \dots, k_n)
+ e^{-2 \pi ki/n} V_{i-1}(k_1, \dots , k_{i-1}, 1, k_{i+1}, \dots, k_n)$. Theorem \ref{thm:generaltheorem} is easily seen to apply to the natural circuit executing this recurrence,
and our refinements would therefore shave a logarithmic factor off the runtime of $\cP$ applied to this circuit, relative to the implementation of \cite{itcs}
(since this circuit is defined over the infinite field $\mathbb{C}$, the protocol is only defined in a model where complex numbers can be communicated and operated on at unit cost).
\end{itemize}

\section{Experimental Results}
\label{sec:expts}
We implemented the protocols implied by Theorem \ref{thm:generaltheorem} as applied to circuits computing $\matmult$ and $\distinct$. 
These experiments serve as case studies to demonstrate the feasibility of Theorem \ref{thm:generaltheorem} in practice, and 
to quantify the improvements over prior implementations. While Section \ref{sec:finalopt} describes a specialized protocol for \matmult\
that is significantly more efficient than the protocol implied by Theorem \ref{thm:generaltheorem}, \matmult\
serves as an important case study for the costs of the more general protocol described in Theorem \ref{thm:generaltheorem}, and allows
for direct comparison with prior implementation work that also evaluated general-purpose protocols via their performance on the \matmult\ problem \cite{itcs, hotcloud, ginger, zaatar, allspice, pinocchio}.

Our comparison point is the implementation of Cormode \etal\ \cite{itcs},
with some of the refinements of Vu \etal\ \cite{allspice} included. In particular, our comparison point for matrix multiplication uses the refinement  of \cite{allspice} for circuits with multiple outputs described in Section \ref{sec:protocoloutline}. We did not include Vu \etal's optimization from Lemma \ref{lemma:vu} that reduced the runtime of $\cV$ from $O(n \log n)$ to $O(n)$, because this optimization
blows up the space usage of $\cV$ to $\Omega(n)$, while we want to use a smaller-space verifier for streaming applications such as \distinct. 

\subsection{Summary of Results}
The main takeaways of our experiments are as follows. When Theorem \ref{thm:generaltheorem} is applicable,
the prover in the resulting protocol is 200x-250x faster than the previous state of the art implementation of the GKR protocol. 
The communication costs and the number of rounds required by our protocols are also 2x-3x smaller than the previous state of the art.
The verifier in our implementation takes essentially the same amount of time as in prior implementations of the GKR protocol; this time
is much smaller than the time to perform the computation locally without a prover.

Most of the observed 200x speedup can be attributed directly to our improvements in protocol design over prior work: the circuit for 512x512 matrix multiplication
is of size $2^{28}$, and hence our $\log S$ factor improvement the runtime of $\cP$ likely accounts for at least a 28x speedup. 
The 3x reduction in the number of rounds accounts for another 3x speedup. The remaining speedup factor of roughly 2x may be due to a more streamlined
implementation relative to prior work, rather than improved protocol design per se.

We have both a serial implementation and a parallel implementation that leverages graphics processing units (GPUs). The prover in our parallel implementation
runs roughly 30x faster than the prover in our serial implementation.
The ability to leverage GPUs to obtain robust
speedups in our setting is not unexpected, as Thaler, Roberts, Mitzenmacher, and Pfister demonstrated 
substantial speedups for an earlier implementation of the GKR protocol using GPUs in \cite{hotcloud}. 

All of our code is available online at \cite{code}. All of our serial code was written in C++ and all experiments were compiled with g++ using the $-$O3 compiler optimization flag and run on a workstation with a 64-bit Intel Xeon architecture and 48 GBs of RAM. We implemented all of our GPU code in CUDA and Thrust \cite{thrust} with all compiler optimizations turned on, and 
ran our GPU implementation on an NVIDIA Tesla C2070 GPU with 6 GBs of device memory.

\subsection{Details}
\label{sec:experimentaldetails}
\medskip \noindent
\textit{Choice of Finite Field.} All of our circuits work over the finite field of size $q=2^{61}-1$. Several remarks are appropriate regarding our choice of field size.
This field was used in our earlier work \cite{itcs} because it supports fast arithmetic, as reducing an integer modulo $q$ can be done with a bit-shift, addition, and a bit-wise AND.
(The same observation applies to any field whose size equals a Mersenne Prime, including $2^{89}-1$, $2^{107}-1$, and $2^{127}-1$). 
Moreover, the field
is large enough that the probability a verifier is fooled by a dishonest prover is smaller than $1/2^{45}$ for all of the problems we consider (this probability is proportional to $\frac{d(n)\log S(n)}{q}$). 

The main potential issue with our choice of field size is that ``overflow'' can occur for problems such as matrix multiplication if the entries of the input matrices can be very large.
For example, with $512 \times 512$ matrix multiplication, if the entries of the input matrices $A, B$ are larger than $2^{26}$, an entry in the product matrix $AB$ can be 
as large as $2^{61}$, which is larger than our field size. If this is a concern, a larger field size is appropriate. 
(Notice that for a problem such \distinct, there is no danger of overflow
issues as long as the length of the stream is smaller than $2^{61}-2$, which is larger than any stream encountered in practice). 

A second reason to use larger field sizes is to handle floating-point or rational arithmetic as proposed by Setty \etal\ in \cite{ginger}.

All of our protocols can be instantiated over fields with more than $q=2^{61}-1$ elements, with an implementation using these fields experiencing a slowdown
proportional to the increased cost of  arithmetic over these fields.

\subsubsection{Serial Implementation}
\label{sec:serialexpts}

\medskip
\noindent \textbf{\matmult.}
The costs of our serial \matmult\ implementation are displayed in Table \ref{tab:matmult}.
The prover in our matrix multiplication implementation is about $250$x faster than the previous state of the art. 
For example, when multiplying two 512 x 512 matrices, our prover takes about 38 seconds, while our comparison implementation
takes over 2.5 hours. A C++ program that simply evaluates the circuit without an integrity guarantee takes 6.07 seconds, 
so our prover experiences less than a 7x slowdown to provide the integrity guarantee relative
to simply evaluating the circuit without such a guarantee.  

When multiplying two 512 x 512 matrices $A$ and $B$, the protocol requires 236 rounds, and the total communication cost of our protocol is 5.48 KBs (plus the amount of communication required
to specify the answer $AB$). The previous state of the art required 767 rounds and close to 18 KBs of communication (plus the amount of communication required to specify $AB$). 
Notice that specifying a 512x512 matrix using 8 bytes per entry requires 2 MBs, which is more than 500 times larger than the  5.48 KBs of extra communication required to verify the answer.

A serial C++ program performing 512 x 512 matrix multiplication over the integers with floating point arithmetic (without going
through the circuit representation of the computation) required 1.53 seconds, so our prover runs approximately 25 times slower
than a standard unverifiable matrix multiplication algorithm. A serial C++ program
performing the same multiplication over the finite field of size $2^{61}-1$ required 4.74 seconds, so our serial prover runs about 8 times slower than
an unverifiable matrix multiplication algorithm over the corresponding finite field.

Our verifier takes essentially the same amount of time as in prior work, as in both implementations the bulk of the work of the verifier is spent evaluating the 
low-degree extension of the input at a point. This is more than an order of magnitude faster than the 1.03 seconds required by a serial C++ program performing the multiplication
in an unverified manner over the integers, so the verifier is indeed saving time by using a prover (relative to doing the computation locally without a prover). 
We stress that the savings for the verifier would be larger at larger input sizes, as the time cost to the verifier in our implementation and the prior implementation of \cite{itcs}
is quasilinear in the input size, which is polynomially faster than all known matrix multiplication algorithms. Moreover, when streaming considerations are not an issue,
we could apply the refinement of Vu \etal\ from Lemma \ref{lemma:vu} to reduce $\cV$'s runtime from $O(n^2 \log n)$ to $O(n^2)$ and thereby further speed up the verifier.

\begin{table}
\small
\centering
\begin{tabular}{|c|c|c|c|c|c|c|}
\hline
Implementation & Problem Size & $\cP$ Time & $\cV$  Time & Rounds & Total Communication & Circuit Eval Time\\
\hline
Previous state of the art & 256 x 256 & 1054 s & 0.02 s & 623 & 14.6 KBs & 0.73 s \\
\hline
Theorem \ref{thm:generaltheorem} & 256 x 256 &  4.37 s & .02 s & 190 & 4.4 KBs & 0.73 s \\
\hline
Previous state of the art & 512 x 512 & 9759 s & 0.10 s  & 767 & 17.97 KBs & 6.07 s \\
\hline
Theorem \ref{thm:generaltheorem} & 512 x 512 & 37.85 s  & 0.10 s & 236 &  5.48 KBs & 6.07 s \\
\hline
\end{tabular}
\caption{Experimental results for $n\times n$ \matmult\ with our serial implementation. The Total Communication column does not count
the communication required to specify the answer, only the ``extra'' communication required to run the verification protocol. }
\label{tab:matmult}
\end{table}

\medskip
\noindent \textbf{\distinct.}
The costs of our serial \distinct\ implementation are displayed in Table \ref{tab:distinct}. The comparison of our implementation
with prior work is similar to the case of matrix multiplication. Our prover is roughly 200 times faster than the comparison implementation.
For example, when computing the number of non-zero entries of a vector of length $2^{20}$, our prover takes about 17 seconds, while our comparison implementation
takes about 57 minutes. A C++ program that simply evaluates the circuit without an integrity guarantee takes 1.88 seconds, 
so our prover experiences roughly a 10x slowdown to prove an integrity guarantee relative
to simply evaluating the circuit.  Our implementation required 1361 rounds and 40.76 KBs of total communication, compared to 3916 rounds and 91.3 KBs for the previous state of the art.
This is essentially a 3x reduction in the number of rounds, and a 2.25x reduction in the total amount of communication.

A C++ program that (unverifiably) computes the number of non-zero entries in a vector $x$ with $2^{20}$ entries takes less than .01 seconds,
and our prover implementation runs more than $1,700$ times longer than this.
The reason that the slowdown for the prover relative to an unverifiable algorithm is larger for \distinct\ than for \matmult\ is that \distinct\ is a ``less arithmetic''
problem, in the sense that the size of the arithmetic circuit we use for computing \distinct\ is more than 100x larger than the runtime
of an unverifiable serial algorithm for the problem. We stress however that, as pointed out in \cite{hotcloud}, when solving the \distinct\ problem in practice, an unverifiable algorithm would
first aggregate a data stream into its frequency-vector representation before determining the number of non-zero frequencies. In reporting a time bound of .01 seconds for unverifiably solving \distinct, we are not
taking the aggregation time cost into account. For sufficiently long data streams, the slow-down for our prover relative to an unverifiable algorithm would be much smaller than $1,700$x if we did take
aggregation time into account.

\begin{table}
\small
\centering
\begin{tabular}{|c|c|c|c|c|c|c|}
\hline
Implementation & $\cP$ Time & $\cV$  Time & Rounds & Total Communication & Circuit Eval Time\\
\hline
Previous state of the art & 3400.23 s  & 0.20 s & 3916 & 91.3 KBs & 1.88 s \\
\hline
Theorem \ref{thm:generaltheorem} &  17.28 s & 0.20 s & 1361 & 40.76 KBs & 1.88 s \\
\hline
\end{tabular}
\label{tab:distinct}
\caption{Experimental results for computing the number of non-zero entries of a vector of length $2^{20}$ with our serial implementation.}
  
\label{tab:distinct}
\end{table}

\subsubsection{Parallel Implementation}
\label{sec:gpu}
Our serial implementation demonstrates that $\cP$ experiences a 10x slowdown in order to evaluate the circuit with an integrity guarantee relative
to simply evaluating the circuit without such a guarantee. The purpose of this section is to demonstrate that parallelization
can further mitigate this slowdown. To this end, we implemented a parallel version of our prover in the context of the matrix multiplication protocol 
of Section \ref{sec:ourrefinements}. Our parallel implementation uses a graphics processing unit (GPU). 

The high-level idea behind our parallel implementation is the following. Each time we apply the sum-check protocol to a polynomial $g_z^{(i)}$, it suffices for $\cP$ to evaluate $g_z^{(i)}$
at a large number of points $r$ of the form $p=(r_1, \dots, r_{j-1}, t, b_{j+1}, \dots, b_{s_{i+1}})$ with $t \in \{0, \dots, \deg_j(g_z^{(i)})\}$ and $(b_{j+1}, \dots, b_{s_{i+1}}) \in \{0, 1\}^{s_{i+1}-j}$. 
We can perform each of these evaluations independently. 
Thus, we devote a single thread on the GPU to each value of $(b_{j+1}, \dots, b_{s_{i+1}}) \in \{0, 1\}^{s_{i+1}-j}$ and 
have that thread evaluate $g_z^{(i)}(r)$ at each of the $\deg_j(g_z^{(i)})+1$ points of the form $(r_1, \dots, r_{j-1}, t, b_{j+1}, \dots, b_{s_{i+1}})$ with the help of the $C^{(j-1)}$ and $V^{(j-1)}$ arrays
described in Section \ref{sec:ourrefinements}. The one remaining issue is that after each round $j$ of each invocation of the sum-check protocol, we need to 
update the arrays, i.e., we need to compute $C^{(j)}$ and $V^{(j)}$. To accomplish this,
we devote a single thread to each entry of $C^{(j)}$ and $V^{(j)}$.  

All steps of our parallel implementation achieve excellent memory coalescing, which likely plays a significant role in the large speedups we were able to achieve. For example, if two threads are updating adjacent entries of the array $V^{(j)}$, the only memory accesses that the threads need to perform are to
adjacent entries of the array $V^{(j-1)}$.

The results are shown in Table \ref{tab:gpu}: we obtained about a 30x speedup for the prover relative to our serial implementation. 
The reported prover runtime does count the time required to copy data between 
the host (CPU) and the device (GPU), but
does not count the time required to evaluate the circuit, which our implementation does in serial for simplicity.
While our implementation evaluates the circuit serially, this step can in principle be done in parallel one layer at a time, as these circuits have only logarithmic depth.
 Notice that when the circuit evaluation runtime is excluded, our parallel prover implementation
runs faster in the case of 512x512 matrix multiplication than the time required to evaluate the circuit sequentially.

It is possible that we would observe slightly larger speedups at larger input sizes, but our parallel implementation exhausts the memory of the GPU at inputs larger than 512x512.
This memory bottleneck was also experienced by Thaler, Roberts, Mitzenmacher, and Pfister \cite{hotcloud}, who used the GPU to obtain a parallel implementation of the protocol of Cormode \etal\ \cite{itcs},
and helps motivate the importance of the improved space usage of the special purpose \matmult\ protocol we give later in Theorem \ref{thm:finalfinalthm}.
For comparison, the GPU implementation of \cite{hotcloud} required 39.6 seconds for 256 x 256 matrix multiplication, which is about 175x slower than our parallel implementation.

We also mention that Thaler, Roberts, Mitzenmacher, and Pfister \cite{hotcloud}
demonstrate that equally large speedups via parallelization
are achievable for the (already fast) computation of the verifier. These results directly apply to our protocols as well, as the verifier's runtime in both
implementations is dominated by the time required to evaluate the MLE of the input at a random point \cite{itcs, hotcloud}.

\begin{table}
\small
\centering
\begin{tabular}{|c|c|c|c|c|c|c|}
\hline
Implementation & Problem Size & $\cP$ Time & Serial Circuit Eval Time\\
\hline
Theorem \ref{thm:generaltheorem}, Serial Implementation & 256 x 256 &  4.37 s & 0.73 s\\
\hline
Theorem \ref{thm:generaltheorem}, Parallel Implementation & 256 x 256 &  0.23 s & 0.73 s \\
\hline
Theorem \ref{thm:generaltheorem}, Serial Implementation & 512 x 512 & 37.85 s & 6.07 s  \\
\hline
Theorem \ref{thm:generaltheorem}, Parallel Implementation & 512 x 512 & 1.29 s  & 6.07 s  \\
\hline
\end{tabular}
\caption{Experimental results for $n\times n$ \matmult\ with our parallel prover implementation.}
\label{tab:gpu}
\end{table}

\section{Verifying General Data Parallel Computations}
\label{sec:dataparallel}
In this section, our goal is to extend the applicability of the GKR protocol.
While the GKR protocol applies in principle to any function computed by a small-depth circuit, this is not the case when fine-grained efficiency
considerations are taken into account. The implementation of Cormode \etal\ \cite{itcs} required the programmer to express a program
as an arithmetic circuit, and moreover this circuit needed to have a regular wiring pattern, in the sense that the verifier could
efficiently evaluate the polynomials $\tilde{\text{add}}_i$ and 
$\tilde{\text{mult}}_i$ at a point. If this was not the case,
the verifier would need to do an expensive (though data-independent) preprocessing phase to perform these evaluations. 
Moreover, even for circuits with regular wiring patterns, this implementation caused the prover to suffer an $O(\log(S(n)))$ factor blowup 
in runtime relative to evaluating the circuit without a guarantee of correctness. 
The results of Sections \ref{sec:ourrefinements} and \ref{sec:finalopt} asymptotically eliminate the blowup in runtime for the prover,  but they also only
apply when the circuit has a very regular wiring pattern. 

The implementation of Vu \etal\ \cite{allspice} allows the 
programmer to express a program in a high-level language, but compiles  these programs into potentially irregular circuits that require the verifier 
to incur the expensive preprocessing phase mentioned above, in order for the verifier to evaluate the polynomials $\tilde{\text{add}}_i$ and 
$\tilde{\text{mult}}_i$ at a point.
They therefore propose to apply their system in a ``batching'' model, where multiple instances of the same sub-computation are applied independently to different 
pieces of data. More specifically, their system applies the GKR protocol independently to each application of the computation, and
relies on the ability of the verifier to use a single $\tilde{\text{add}}_i$ and 
$\tilde{\text{mult}}_i$ evaluation for all instances of the
sub-computation, thereby amortizing the cost of this evaluation across the instances. To clarify, this use of a single $\tilde{\text{add}}_i$ and 
$\tilde{\text{mult}}_i$ evaluation for all instances as in \cite{allspice} is only sound if all of the instances are checked simultaneously. 
If the instances are instead verified one after the other, then $\cP$ knows $\cV$'s randomness in all but the first instance, and can use that knowledge to mislead $\cV$. 

The batching model of Vu \etal\ is identical to the data parallel setting we consider here.
However, a downside to the solution of Vu \etal\ is that the verifier's work, as well as the total communication cost of the protocol, 
grows linearly with the ``batch size'' -- 
the number of applications of the sub-computation that are being outsourced. 
We wish to develop a protocol whose costs to both the prover and verifier grow
much more slowly with the batch size. 

\subsection{Motivation}
 As discussed above, existing interactive proof protocols for circuit evaluation either apply only to circuits
 with highly regular wiring patterns or incur large overheads for the prover and verifier. 
 While we do not have a magic bullet for dealing with irregular wiring patterns, we 
 do wish to mitigate the bottlenecks of existing protocols by leveraging some general structure underlying many real-world computations.
Specifically, the structure we focus on exploiting is data-parallelism. 
 
 By data parallel computation, we mean any setting in which
 the same sub-computation is applied independently to many pieces of data, before possibly aggregating the results.
Crucially, we do not want to make significant assumptions on the sub-computation that is being applied
(in particular, we want to handle sub-computations computed by circuits with highly irregular wiring patterns), 
but we are willing to assume that the sub-computation is applied independently to many pieces of data.
See Figure \ref{fig:datapar} for a schematic of a data parallel computation.

\begin{figure}
\centering
\includegraphics[width=3in]{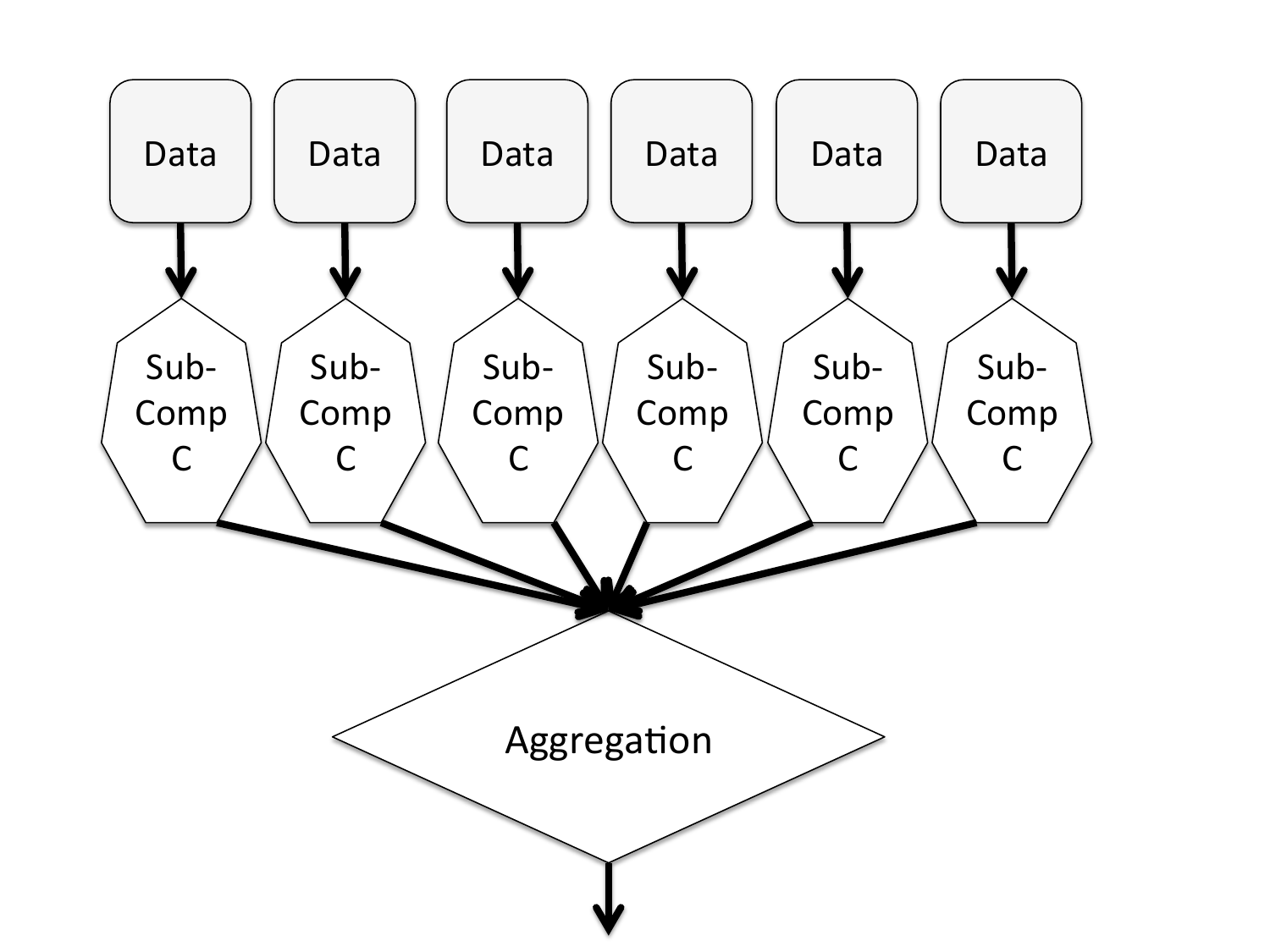}
\caption{Schematic of a data parallel computation.}
\label{fig:datapar}
\end{figure}

We have already seen a very simple example of a data parallel computation: the \distinct\ problem. The circuit $C$ from Section \ref{sec:ourrefinements} 
used to solve this problem takes as input a vector $a$ and computes $a_i^{q-1} \mod q$ for all $i$ (this is the data parallel phase of the computation), before summing the results
(this is the aggregation phase).
Notice that if the data stream consists of a sequence of words, then the \distinct\ problem becomes the word-count problem, a classic data parallel application.

By design, the protocol of this section also applies to more complicated data parallel computations. For example, it applies to arbitrary \emph{counting queries}
on a database. In a counting query, one applies some function independently to each row of the database and sums the results. 
For example, one may ask ``How many people in the database satisfy Property $P$?'' Our protocol allows one to verifiably outsource
such a counting query with overhead that depends minimally on the size of the database, but that necessarily depends on the complexity of the property $P$.

\subsection{Overview of the Protocol}
Let $C$ be a circuit of size $S(n)$ with an arbitrary wiring pattern, and let $C^*$ be a ``super-circuit'' that applies $C$ independently to $B$ different inputs 
before aggregating the results in some fashion. For example, in the case of a counting query, the aggregation phase simply sums the results of the data parallel phase.
We assume that the aggregation step is sufficiently simple that the aggregation itself can be verified 
using existing techniques, and we focus on verifying the data parallel part of the computation.


If we naively apply the GKR protocol to the super-circuit $C^*$, $\cV$ might have to perform an expensive pre-processing phase
to evaluate the wiring predicate of $C^*$ at the necessary locations -- this would require time $\Omega(B \cdot S)$. 
Moreover, when applying the basic GKR protocol to $C^*$, $\cP$ would require time $\Theta\left(B \cdot S \cdot \log(B \cdot S)\right)$. 
A different approach was taken by Vu et al \cite{allspice}, who applied the GKR protocol $B$ independent times, once for each copy of $C$.
This causes both the communication cost and $\cV$'s online check time to grow linearly with $B$, the number of sub-computations.

In contrast, our protocol achieves the best of both prior approaches. We observe that although each sub-computation $C$ can have a complicated wiring pattern, 
the circuit is maximally regular between sub-computations, as the sub-computations do not interact at all. Therefore,
each time the basic GKR protocol would apply the sum-check protocol to a polynomial
derived from the wiring predicate of $C^*$, we can instead use a simpler polynomial derived only from the wiring predicate of $C$.
By itself, this is enough to ensure that $\cV$'s pre-processing phase requires time only $O(S)$, rather than $O(B \cdot S)$ as in a naive
application of the basic GKR protocol.
That is, the cost of $\cV$'s pre-processing phase is essentially proportional to the cost of applying the GKR protocol only to $C$, not to the super-circuit $C^*$.

Furthermore, by combining this observation with the methods of Section \ref{sec:ourrefinements}, we can bring the runtime of $\cP$ down to $O(B \cdot S \cdot \log S)$. That is, the blowup in runtime suffered
by the prover, relative to performing the computation without a guarantee of correctness, is just a factor of $\log S$ --
 the same as it would be if the prover had run the basic GKR protocol on a single instance of the sub-computation.

\subsection{Technical Details}
\subsubsection{Notation} 
Let $C$ be an arithmetic circuit over $\mathbb{F}$ of depth $d$ and size $S$ with an arbitrary wiring pattern, and let 
 $C^*$ be the circuit of depth $d$ and size $B \cdot S$ obtained by laying $B$ copies of $C$ side-by-side, where $B=2^b$ is a power of 2. 
 We assume that the in-neighbors of all of the $S_i$ gates at layer $i$ can be enumerated in $O(S_i)$ time. 
 We will use the same notation as in Section \ref{sec:ourrefinements}, using $^*$'s to denote quantities referring to $C^*$.
For example, layer $i$ of $C$ has size $S_i=2^{s_i}$ and gate values specified by the function $V_i$, while layer $i$ of $C^*$ has size $S_i^* = 2^{s_i^*}$ and
gate values specified by the function $V_i^*$. We denote the length of the input to $C^*$ by $n^*=Bn$.

\subsubsection{Main Theorem} 
Our main theorem gives a protocol for compute 
$\tilde{V}^*_1(z)$, for any point $z \in \mathbb{F}^{s_1^*}$.
The idea is that the verifier would first apply 
simpler techniques (such as the protocol of Theorem \ref{thm:generaltheorem}) to the aggregation phase of the 
computation to obtain a claim about  $\tilde{V}^*_1(z)$,
and then use our main theorem to verify this claim.
Hence, in principle $\cV$ need not look at the entire output
of the data parallel phase, only the output of the aggregation
phase, which we anticipate to be much smaller.

\begin{theorem} \label{thm:dataparallel}
For any point $z \in \mathbb{F}^{s_1^*}$, there is a valid interactive proof protocol for computing $\tilde{V}^*_1(z)$ with the following costs. 
$\cV$ spends $O(S)$ time in a pre-processing phase, and $O(n^*\!\log n^*\!+\!d\!\cdot\!\log(B\!\cdot\!S))$ time in an online verification phase,
where the $n^* \log n^*$ term is due to the time required to evaluate the multilinear extension of the input to $C^*$ at a point.
$\cP$ runs in total time $O(S \cdot B \cdot \log S)$. The total communication is $O(d \cdot \log(B \cdot S))$ field elements. 
\end{theorem}

\begin{proof}
Consider layer $i$ of $C^*$. 
Let $p = (p_1, p_2) \in \{0, 1\}^{s_i} \times \{0, 1\}^b$ be the label of a gate at layer $i$ of $C^*$, where $p_2$ 
specifies which ``copy'' of $C$ the gate is in, while $p_1$ designates the label of the gate within the copy.
Similarly, let $\omega = (\omega_1, \omega_2) \in  \{0, 1\}^{s_{i+1}} \times \{0, 1\}^b$ and $\gamma = (\gamma_1, \gamma_2) \in  \{0, 1\}^{s_{i+1}} \times \{0, 1\}^b$ be the labels of two gates at layer $i+1$.

It is straightforward to check that for all $(p_1, p_2)  \in \{0, 1\}^{s_i} \times \{0, 1\}^b$,\\ 
$$V^*_i(p_1, p_2) = \sum_{\omega_1 \in \{0, 1\}^{s_{i+1}}} \sum_{\gamma_1 \in \{0, 1\}^{s_{i+1}}}  h^{(i)}(p_1, p_2, \omega_1, \gamma_1),$$
where

\begin{align*}
h^{(i)}(p_1, p_2, \omega_1, \gamma_1) = 
\end{align*}
\begin{align*}
\big(\tilde{\text{add}}_i(p_1, \omega_1, \gamma_1) \left(\tilde{V}^*_{i+1}(\omega_1, p_2) +\tilde{V}^*_{i+1}(\gamma_1, p_2)\right) +
\tilde{\text{mult}}_i(p_1, \omega_1, \gamma_1) \left(\tilde{V}^*_{i+1}(\omega_1, p_2) \cdot \tilde{V}^*_{i+1}(\gamma_1, p_2)\big)\right). &\end{align*}

Essentially, this equation says that an addition (respectively, multiplication) gate $p=(p_1, p_2) \in \{0, 1\}^{s_i + b}$ is connected to gates $\omega=(\omega_1, \omega_2) \in \{0, 1\}^{s_{i+1} + b}$ and $\gamma=(\gamma_1, \gamma_2) \in \{0, 1\}^{s_{i+1} + b}$ if and only if $p, \omega,$ and $\gamma$ are all in the same copy of $C$, and $p$ is connected to $\omega$ and $\gamma$ within the copy.

Lemma \ref{prop:big} then implies that for any $z \in \mathbb{F}^{s_i^*}$,

$$\tilde{V}^*_i(z) = \sum_{(p_1, p_2, \omega_1, \gamma_1) \in \{0, 1\}^{s_i} \times \{0, 1\}^b \times \{0, 1\}^{s_{i+1}} \times  \{0, 1\}^{s_{i+1}}} \beta_{s_i^*}(z, (p_1, p_2)) \cdot h^{(i)}(p_1, p_2, \omega_1, \gamma_1).$$

Thus, in iteration $i$ of our protocol, we apply the sum-check protocol to the polynomial $g_z^{(i)}$ 
given by $g_z^{(i)}(p_1, p_2, \omega_1, \gamma_1) = \beta_{s_i^*}(z, (p_1, p_2))  \cdot h^{(i)}(p_1, p_2, \omega_1, \gamma_1)$. 
The communication costs of this protocol are immediate.

\medskip 
\noindent \textbf{Costs for $\cV$.}
In order to run her part of the sum-check protocol of iteration $i$, $\cV$ only needs to perform the required checks
on each of $\cP$'s messages.  $\cV$'s check requires $O(1)$ time in each round of the sum-check protocol except the last.  
In the last round of the sum-check protocol, $\cV$ must evaluate the polynomial $g_z^{(i)}$ at a single
point. This requires evaluating $\beta_{s_i^*}$, $\tilde{\text{add}}_i$, $\tilde{\text{mult}}_i$, and $\tilde{V}^*_{i+1}$ at a constant number of points.
The $\tilde{V}^*_{i+1}$ evaluations are provided by $\cP$ in all iterations $i$ of the protocol except the last, while the $\beta_{s_i^*}$ evaluation can be done in $O(\log(B \cdot S))$ time.

The $\tilde{\text{add}}_i$ and $\tilde{\text{mult}}_i$ computations can be done in pre-processing in time
$O(S_i)$ by enumerating the in-neighbors of each of the $S_i$ gates at layer $i$ \cite{itcs, allspice}. 
Adding up the pre-processing time across all iterations $i$ of our protocol, $\cV$'s pre-processing time is $O(\sum_i S_i) = O(S)$ as claimed. 

In the final iteration of the protocol, $\cP$ no longer provides the $\tilde{V}^*_{i+1}$ evaluation for $\cV$; instead, $\cV$ must 
evaluate the multilinear extension of the input at a point on her own. This can be done in a streaming manner using space $O(\log n^*)$ in time $O(n^* \log n^*)$.
The time cost for $\cV$ in the online phase follows.

\medskip 
\noindent \textbf{Costs for $\cP$.} It remains to show that $\cP$ can perform the required computations in iteration $i$ of the protocol in time $O((S_i + S_{i+1}) \cdot B \cdot \log(S))$.
To this end, notice $g_z^{(i)}$ is a polynomial in $v :=s_i+2s_{i+1}+b$ variables. 
We order the sum in this sum-check protocol so that the $s_i + 2s_{i+1}$ variables in $p_1$, $\omega_1$, and $\gamma_1$ are bound first in arbitrary order,
 followed by the variables of $p_2$. 
$\cP$ can compute the prescribed messages in the first $s_i + 2s_{i+1}=O(\log S)$ rounds exactly as in the implementation of Cormode \etal\ \cite{itcs}. They show that
 each gate at layers $i$ and $i+1$ of $C^*$ contributes to exactly one term in the sum defining $\cP$'s message in any given round of the sum-check protocol, and moreover
 the contribution of a given gate can be determined in $O(1)$ time.
Hence the total time devoted required by $\cP$ to handle these rounds is $O(B \cdot (S_i + S_{i+1}) \cdot \log S)$. It remains to show
 how $\cP$ can compute the prescribed messages in the final $b$ rounds of the sum-check protocol while investing $O(\left(S_i + S_{i+1}\right)\cdot B)$ across 
 all rounds of the protocol. 

Recall that in order to compute $\cP$'s message in round $j$ of the sum-check protocol applied to the $v$-variate polynomial $g_z^{(i)}$, it suffices for $\cP$ to evaluate $g_z^{(i)}$ at 
 $2^{v-j}$ points of the form $(r_1, \dots, r_{j-1}, t, b_{j+1}, \dots, b_{v})$, with $t \in \{0, \dots, \deg_j(g_z^{(i)})\}$ and $(b_{j+1}, \dots, b_{v}) \in \{0, 1\}^{v-j}$.
Each of these evaluations of $g_z^{(i)}$ can be computed in $O(1)$ time given the evaluations of 
$\beta_{s^*_i}$, $\tilde{\text{add}_i}$,  $\tilde{\text{mult}_i}$, and $\tilde{V}^*_{i+1}$ at the relevant points.

Notice that once the variables in $p_1$, $\omega_1$, and $\gamma_1$ are bound to specific values, say $r_1^{(p)}$, $r_1^{(\omega)}$, and $r_1^{(\gamma)}$,
 $\tilde{\text{add}}_i(p_1, \omega_1, \gamma_1)$ and  $\tilde{\text{mult}}_i(p_1, \omega_1, \gamma_1)$ are themselves bound to specific values, namely
  $\tilde{\text{add}}_i(r_1^{(p)}, r_1^{(\omega)}, r_1^{(\gamma)})$ and $\tilde{\text{mult}}_i(r_1^{(p)}, r_1^{(\omega)},  r_1^{(\gamma)})$. So $\cP$
  only needs to evaluate these polynomials once, and
both of these evaluations can be computed by $\cP$
in $O(S_i)$ time. Thus,  the $\tilde{\text{add}_i}$,  $\tilde{\text{mult}_i}$
 evaluations in the last $b$ rounds require just $O(S_i)$ time in total.

 $\cP$ can evaluate the function $\beta_{s^*_i}$ at the relevant points exactly as in the proof of Theorem \ref{thm:generaltheorem} using the $C^{(j)}$ arrays to ensure
 that this computation is done quickly.
 The array $C^{(0)}$ has size $2^{s_i^*} = O(S_i \cdot B)$, and $C^{(j-1)}$ gets updated to $C^{(j)}$ whenever a variable in $p_1$ or $p_2$ becomes bound.
 This ensures that across all rounds of the sum-check protocol, the $\beta_{s_i^*}$ evaluations require $O(S_i \cdot B)$ time in total.
 
 Likewise, the $\tilde{V}^*_{i+1}$ evaluations can be handled exactly as in Theorem \ref{thm:generaltheorem}, using the the $V^{(j)}$ arrays to ensure that this computation is done quickly.
 The array $V^{(0)}$ has size $2^{s_{i+1}^*} = O(S_{i+1} \cdot B)$, and $V^{(j-1)}$ gets updated to $V^{(j)}$ whenever a variable in $\omega_1$ becomes bound (and similarly
 for the variables in $\gamma_1$).
 This ensures that across all rounds of the sum-check protocol, the $\tilde{V}^*_{i+1}$ evaluations take $O((S_i + S_{i+1}) \cdot B)$ in total. 
 
 \medskip
\textbf{Reducing to Verification of a Single Point.} After executing the sum-check protocol at layer $i$ as described above, 
$\cV$ is left with a claim about $\tilde{V}_{i+1}(\omega_1, p_2)$ and $\tilde{V}_{i+1}(\gamma_1, p_2)$, for $\omega_1, \gamma_1 \in \mathbb{F}^{s_i}$,
and $p_2 \in \mathbb{F}^b$.
 This requires $\cP$ to send $\tilde{V}_{i+1}(\ell(t))$ for a canonical line $\ell(t)$ that passes through  $(\omega_1, p_2)$ and $(\gamma_1, p_2)$.
It is easily seen that $\tilde{V}_{i+1} (\ell(t))$ is a univariate polynomial of degree at most $s_i$. Here, we are exploiting the fact that 
the final $b$ coordinates of $(\omega_1, p_2)$ and $(\gamma_1, p_2)$ are equal.

Hence $\cP$ can specify $\tilde{V}_{i+1} (\ell(t))$ by sending $\tilde{V}_{i+1}(\ell(t_j))$ for $O(s_i)$ many points $t_j \in \mathbb{F}$. Using
the method of Lemma \ref{lemma:vu}, $\cP$ can evaluate $\tilde{V}_{i+1}$ at each point $\ell(t_j)$ in $O(S_{i+1})$ time, and hence can perform all
$\tilde{V}_{i+1}(\ell(t_j))$ evaluations in $O(S_{i+1} \cdot s_i)=O(S_{i+1} \cdot \log S)$ time in total.
 This ensures that across all iterations of our protocol, $\cP$ devotes at most $O(S \cdot B \cdot \log S)$ time to the ``reducing to verification of a single point'' 
 phase of the protocol. This completes the proof.

\end{proof}  

In practice we would expect the results of the data parallel phase of computation represented by the super-circuit $C^*$ to be aggregated in some fashion.
We assume this aggregation step is amenable to verification via other techniques. In the case of counting queries, 
the aggregation step simply sums the outputs of the data parallel step, which can be handled via Theorem \ref{thm:generaltheorem}, or slightly more
efficiently via Proposition \ref{prop:finalopt} described below in Section \ref{sec:finalopt}.
More generally, if this aggregation step is computed by a circuit $C'$ of 
size $O(S \cdot B \cdot \log S/\log B)$ such that $\cV$ can efficiently evaluate the multilinear extension of the wiring predicate of $C'$, 
then we can simply apply the basic GKR protocol to $C'$ with asymptotic costs smaller than those of the protocol described in
Theorem \ref{thm:dataparallel}. 
This application of the GKR protocol to $C'$ ends with a claim about the value of $\tilde{V}^*_1(z)$ for some $z \in \mathbb{F}^{s_1^*}$.
The verifier can then invoke the protocol of Theorem \ref{thm:dataparallel} to verify this claim.

We stress that the protocol of Theorem \ref{thm:dataparallel} can be applied if there are multiple data parallel stages interleaved with aggregation stages.



\section{Extensions}
\label{sec:finalopt}
In this section we describe two final optimizations that are much more specialized than Theorems \ref{thm:generaltheorem} and \ref{thm:dataparallel},
but have a significant effect in practice when they apply.
In particular, Section \ref{sec:optmatmult} culminates in a protocol for matrix multiplication that is of interest in its own right. It is 
hundreds of times faster than the protocol implied by Theorem \ref{thm:generaltheorem} and studied experimentally in Section \ref{sec:expts}.

\subsection{Binary Tree of Addition Gates}
\label{sec:bintreeopt}
Cormode \etal\ \cite{gkr} describe an optimization that applies to any circuit $C$ with a single output that culminates in a binary tree of addition gates;
at a high level, they directly apply a single sum-check protocol to the entire binary tree, thereby treating the entire tree as a single addition gate with
very large fan-in.
In contrast, the optimization described here applies to circuits with multiple outputs and allows the binary tree of addition gates to occur anywhere in the circuit,
not just at the layers immediately preceding the output.

At first blush, our optimization might seem quite specialized since it only applies to circuits with a specific wiring pattern.
However, this is one of the most commonly occurring wiring patterns, as evidenced by its appearance within the circuits computing \matmult, 
\distinct, Pattern Matching, and counting queries.
Notice that our optimization also applies to verifying multiple independent instances of any problem with a single output whose circuit ends with a binary tree of sum-gates, 
such as verifying the number of distinct items in multiple distinct data streams, or posing multiple separate counting queries to a database. 
This is because, similar to Theorem \ref{thm:dataparallel}, one can lay the circuits for each of the individual problem instances side-by-side and treat the result as a single
``super-circuit'' culminating in a binary tree of addition gates with multiple outputs.

The starting point for our optimization is the observation of Vu \etal\ \cite{allspice} mentioned in Section \ref{sec:protocoloutline}: 
in order to verify that $\cP$ has correctly evaluated a circuit with many output gates,
$\cP$ may simply send $\cV$ the (claimed) values of all output gates, thereby specifying a function $V'_1 : \{0, 1\}^{s_1} \rightarrow \mathbb{F}$ claimed
to equal $V_1$.
$\cV$ can pick a random point $z \in \mathbb{F}^{s_1}$ and evaluate $\tilde{V}'_1(z)$ on her own in $O(S_1)$ time. An application of the 
Schwartz-Zippel Lemma (Lemma \ref{lemma:schwartzzippel}) implies that it is safe for $\cV$ to
believe that $V_1$ is as claimed as long as $\tilde{V}_1(z)=\tilde{V}'_1(z)$. 
Our protocol as described in Section \ref{sec:ourrefinements} would then proceed in iterations, with one iteration per layer of the circuit and one application of the sum-check
protocol per iteration. This would ultimately reduce $\cP$'s claim about the value of $\tilde{V}_1(z)$ to a claim about $\tilde{V}_d(z')$ for some $z' \in \mathbb{F}^{s_d}$,
where $d$ is the input layer of the circuit.

Instead, our final refinement uses a single sum-check protocol to directly reduce $\cP$'s claim about $\tilde{V}_1(z)$
to a claim about $\tilde{V}_d(z')$ for some random points $z' \in \mathbb{F}^{s_d}$. 

\begin{proposition} \label{prop:finalopt}
Let $C$ be a depth-$d$ circuit consisting of a binary tree of addition gates, $2^k$ inputs, and $2^{k-d}$ outputs. For any points
$z \in \mathbb{F}^{k-d}$,
$\tilde{V}_1(z) = \sum_{p \in \{0, 1\}^{k}} g_z(p)$, where 

$$g_z(p) = \tilde{V}_d(z, p_{k-d+1}, \dots, p_{k}).$$
\end{proposition}
\begin{proof}
At layer $i$ of $C$, the gate with label $p \in \{0, 1\}^{s_i}$ is the sum of the gates with labels $(p, 0)$ and $(p, 1)$ at layer $i+1$. 
It is then straightforward to observe that the for any $p \in \{0, 1\}^{k-d}$, the $p$th output gate has value 

\begin{equation} \label{finalopeq} V_1(p_1, \dots, p_{k-d}) = \sum_{(p_{k-d+1}, \dots, p_{d}) \in \{0, 1\}^{d}} \tilde{V}_d(p_1, \dots, p_{k-d}, p_{k-d+1}, \dots, p_k). \end{equation}
Notice that the right hand side of  Equation \eqref{finalopeq} is a multilinear polynomial in the variables $(p_1, \dots, p_{k-d})$ that agrees with
$V_1(p_1, \dots, p_{k-d})$ at all Boolean inputs. Hence, the right hand side is the (unique) multilinear extension $\tilde{V}_1$ of the function $V_1: \{0, 1\}^{k-d} \rightarrow \{0, 1\}$. 
The theorem follows.
\end{proof}

In applying the sum-check protocol to the polynomial $g_z$ in Proposition \ref{prop:finalopt}, it is straightforward to use the methods of Section \ref{sec:Vs}
to implement the honest prover in time $O(2^k)$. We omit the details for brevity.

\begin{table}
\small
\centering
\begin{tabular}{|c|c|c|c|c|c|c|}
\hline
Implementation & Problem Size & $\cP$ Time & $\cV$  Time & Rounds & Total Communication & Circuit Eval Time\\
\hline
Theorem \ref{thm:generaltheorem} & 256 x 256 &  4.37 s & 0.02 s & 190 & 4.4 KBs & 0.73 s \\
\hline
Proposition \ref{prop:finalopt} & 256 x 256 & 2.52 s & 0.02 s & 35 & 0.76 KBs & 0.73 s \\
\hline
Theorem \ref{thm:generaltheorem} & 512 x 512 & 37.85 s  & 0.10 s & 236 &  5.48 KBs & 6.07 s \\
\hline
Proposition \ref{prop:finalopt}  & 512 x 512 & 22.98 s &  0.10 s & 39 & 0.86 KBs & 6.07 s \\
\hline
\end{tabular}
\caption{Experimental results for $n\times n$ \matmult, with and without the refinement of Section \ref{sec:bintreeopt}. As in Table \ref{tab:matmult},
the Total Communication column does not count the $n^2$ field elements required to specify the answer.}
\label{tab:finalopt}
\end{table}

\medskip
\noindent \textbf{Experimental Results.} Let $C$ be the circuit for naive matrix multiplication described in Section \ref{sec:applications}. To demonstrate the efficiency gains implied by Proposition \ref{prop:finalopt}, we modified our \matmult\ implementation of Section \ref{sec:serialexpts} to use the protocol
of Proposition \ref{prop:finalopt} to verify the sub-circuit of $C$ consisting of a binary tree of addition gates. The results are shown in Table \ref{tab:finalopt}.
Our optimizations in this section shave $\cP$'s runtime by a factor of 1.5x-2x, the total number of rounds by a factor of more than 5, and the total communication (not counting 
the cost of specifying the output of the circuit) by a factor of more than 5. 

\subsection{Optimal Space and Time Costs for \matmult}
\label{sec:optmatmult}
We describe a final optimization here on top of Proposition \ref{prop:finalopt}. While this optimization is specific to the \matmult problem, 
its effects are substantial and 
the underlying observation may be more broadly applicable. 

Suppose we are given an unverifiable algorithm
for $n \times n$ matrix multiplication that requires time $T(n)$ and space $s(n)$. Our refinements reduce the prover's runtime from $O(n^3)$
in the case of Sections \ref{sec:ourrefinements} and \ref{sec:bintreeopt}
to $T(n) + O(n^2)$, and lowers $\cP$'s space requirement to $s(n) + o(n^2)$. That is, in the protocol the prover sends the correct output and performs just $O(n^2)$
more work to provide a guarantee of correctness on top. It is irrelevant what algorithm the prover uses to 
arrive at the correct output -- in particular, algorithms much more sophisticated than naive matrix multiplication are permitted.
This runtime and space usage for $\cP$ are optimal even up to the leading constant assuming matrix multiplication cannot be computed in $O(n^2)$ time.

The final protocol is extremely natural, as it consists of a single invocation of the sum-check protocol.
We believe this protocol is of interest in its own right. The proof and technical details are in Section \ref{sec:matmultspecialpurposedetails}.



\begin{theorem} \label{thm:finalfinalthm}
There is a valid interactive proof protocol for $n \times n$ matrix multiplication over the field $\mathbb{F}_q$ with the following costs.
The communication cost is $n^2 + O(\log n)$ field elements. The runtime of the prover is $T(n)+O(n^2)$ and the space usage is $s(n) + o(n^2)$, where $T(n)$ and $s(n)$
are the time and space requirements of any (unverifiable) algorithm for $n \times n$ matrix multiplication. 
The verifier can make a single streaming pass over the input as well as over the claimed output in time $O(n^2 \log n)$, storing $O(\log n)$ field elements. 
\end{theorem}

Using the observation of Vu \etal\ described in Lemma \ref{lemma:vu}, the runtime of the verifier can be brought down to $O(n^2)$ at the cost of increasing $\cV$'s space usage
to $O(n^2)$. Furthermore, by Remark \ref{remark:vu}, the runtime of the verifier can be brought down to $O(n^2)$ while maintaining the streaming property if the input matrices are presented in row-major order.

The prover's runtime in Theorem \ref{thm:finalfinalthm} is within an additive low-order term of any unverifiable algorithm for matrix multiplication;
this is essential in many practical scenarios where even a 2x slowdown is too steep a price to pay for verifiability. Notice also 
that the space usage bounds in Theorem \ref{thm:finalfinalthm} are in stark contrast to protocols based on circuit-checking: the prover in a general circuit-checking protocol
may have to store the entire circuit, and this can result in space requirements that are much larger than those of an unverifiable algorithm
for the problem. For example, naive matrix multiplication requires time $O(n^3)$, but only $O(n^2)$ space,
while the provers in our \matmult\ protocols of Sections \ref{sec:ourrefinements} and \ref{sec:bintreeopt} require both space and time $O(n^3)$.
As implementations of interactive proofs become faster, the prover is likely to run out of space long before she runs out of time.

\subsubsection{Comparison to Prior Work}
It is worth comparing Theorem \ref{thm:finalfinalthm} to a well-known protocol due to Freivalds \cite{freivalds}. Let $D^*$ denote the 
claimed output matrix. In Freivalds' algorithm,
the verifier stores a random vector $x \in \mathbb{F}^n$, and computes $D^*x$ and $ABx$, accepting if and only if $ABx = D^*x$. 
Freivalds showed that this is a valid protocol. In both Freivalds' protocol and that of Theorem \ref{thm:finalfinalthm}, the prover
runs in time $T(n) + O(n^2)$ (in the case of Freivalds' algorithm, the $O(n^2)$ term is 0), and the verifier runs in linear or quasilinear time.

We now highlight several properties of our protocol that are not achieved by prior work.

\medskip \noindent
\textbf{Utility as a Primitive.} A major advantage of Theorem \ref{thm:finalfinalthm} relative to prior work
is its utility as a primitive that can be used to verify more complicated computations. 
This is important as many algorithms repeatedly invoke matrix multiplication as a subroutine. 
For concreteness, consider the problem of computing $A^{2^k}$ via repeated squaring. 
By iterating the protocol of Theorem \ref{thm:finalfinalthm} $k$ times, we obtain a valid interactive proof protocol for computing $A^{2^k}$
with communication cost $n^2 + O(k\log(n))$. The $n^2$ term is due simply to specifying the output $A^{2^k}$,
and can often be avoided in applications -- see for example the diameter protocol described two paragraphs hence. 
The $i$th iteration of the protocol for computing $A^{2^k}$ reduces a claim about an evaluation of the multilinear extension of $A^{2^{k-i+1}}$
to an analogous claim about $A^{2^{k-i}}$.
Crucially, the prover in this protocol never
needs to send the verifier the intermediate matrices $A^{2^{k'}}$ for $k' < k$. 
In contrast, applying Freivalds' algorithm 
to this problem would require $O(kn^2)$ communication, as $\cP$ must specify each of the intermediate matrices $A^{2^i}$. 

The ability to avoid having $\cP$ explicitly send intermediate matrices is especially important in settings where an algorithm
repeatedly invokes matrix multiplication, but 
the desired output of the algorithm is smaller than the size of the matrix. In these cases, it is not necessary for $\cP$ to send
\emph{any} matrices; $\cP$ can instead send just the desired output, and $V$ can use Theorem \ref{thm:finalfinalthm} to 
check the validity of the output with only a polylogarithmic amount of additional communication. This is analogous to how the verifier in the GKR protocol
can check the values of the output gates of a circuit without ever seeing the values of the ``interior'' gates of the circuit.

As a concrete example illustrating the power of our matrix multiplication protocol, consider the fundamental problem of computing the diameter of an unweighted (possibly directed) graph $G$ on $n$ vertices.
Let $A$ denote the adjacency matrix of $G$, and let $I$ denote the $n \times n$ identity matrix. Then it is easily verified that the diameter of $G$ is the least positive number $d$
such that $(A+I)^d_{ij} \neq 0$ for all $(i, j)$. We therefore obtain the following natural protocol for diameter. $\cP$ sends the claimed output $d$ to $V$, as well as an $(i, j)$ 
such that $(A+I)^{d-1}_{ij} = 0$. To confirm that $d$ is the diameter of $G$, it suffices for $\cV$ to check two things: first, that all entries of $(A+I)^d$ are non-zero, and second that $(A+I)^{d-1}_{ij}$ is indeed non-zero.

The first task is accomplished by combining our matrix multiplication protocol of Theorem \ref{thm:finalfinalthm} with our \distinct\ protocol from
Theorem \ref{thm:generaltheorem}. 
Indeed, let $d_j$ denote the $j$th bit in the binary representation of $d$. Then $(A+I)^d = \prod_{j}^{\lceil\log d \rceil} (A+I)^{d_j 2^j}$, so
computing the number of non-zero entries of $(A+I)^d$ can be computed via a sequence of $O(\log d)$ matrix multiplications, followed by a \distinct\ computation. 
The second task, of verifying that $(A+I)^{d-1}_{ij}=0$, is similarly accomplished using $O(\log d)$ invocations of the matrix multiplication protocol of Theorem \ref{thm:finalfinalthm} -- since $\cV$ is only interested in one entry of $(A+I)^{d-1}$,  $\cP$
need not send the matrix $(A+I)^{d-1}$ in full, and
the total communication here is just $\polylog(n)$. 

$\cV$'s runtime in this diameter protocol is $O(m \log n)$, where $m$ is the number of edges in $G$. $\cP$'s runtime in the above diameter protocol matches the best known unverifiable diameter algorithm up to a low-order additive term \cite{apsp, diameterref}, and the communication is just $\polylog(n)$. 
We know of no other protocol achieving this. 

As discussed above, the fact that $\cP$'s slowdown is a low-order additive term is critical in the many settings in which even a 2x slowdown to achieve verifiability is unacceptable.
Moreover, for a graph with $n=1$ million nodes, the total communication cost of the above protocol is on the order of KBs -- in contrast, if $\cP$ had to send the matrices $(I+A)^{d}$ or $(I+A)^{d-1}$ 
explicitly (as required in prior work e.g. Cormode \etal\ \cite{esa}), the communication cost would be at least $n^2=10^{12}$ words, which translates
to terabytes of data. 

\eat{consider the problem of computing the shortest path between two nodes $s$ and $t$ in an unweighted graph. 
$\cP$ can claim that the shortest $s$-$t$ path has length $\ell$, and this can be confirmed by computing $A^{\ell-1}_{s, t}$ and $A^{\ell}_{s, t}$ via repeated squaring of the adjacency matrix, and ensuring that the former equals 0 while the latter is larger than zero. If the actual shortest $s$-$t$ path is desired, rather than just its cost, $\cP$ can specify a candidate path $P$ of length $\ell$, and $\cV$ can use an appropriate invocation of the sum-check protocol to quickly check that all of the edges in $P$ do appear in the graph. 

On a graph $G$ with $n=1$ million nodes, each invocation of Freivald's algorithm applied to a (dense) $n \times n$ matrix would require $n^2=10^{12}$ words of communication, which translates
to terabytes of data in practice. In contrast, the protocol just described would require $O(n)$ communication (the communication
would be dominated by simply specifying the correct path), which would be on the order of megabytes in practice. }

\medskip \noindent
\textbf{Small-Space Streaming Verifiers.} 
In Freivalds' algorithm, $\cV$ has the store the random vector $x$, which requires $\Omega(n)$ space. There are 
methods to reduce $\cV$'s space usage by generating $x$ with limited randomness: Kimbrel and Sinha \cite{kimbrelsinha} show how to reduce $\cV$'s space to $O(\log n)$, but their solution does not work if $\cV$ must make a streaming pass over arbitrarily ordered input. Chakrabarti \etal\ \cite{icalp09} extend the method of Kimbrel and Sinha to work with a streaming verifier, but this requires $\cP$ to play back the input matrices $A, B$ in a special order, increasing proof length to $3n^2$. Our protocol works with a streaming verifier using $O(\log n)$ space, and our proof length is $n^2+O(\log n)$, where the $n^2$ term is due to specifying $AB$ and can be avoided in applications such as the diameter example considered above.

\subsubsection{Protocol Details}
\label{sec:matmultspecialpurposedetails}
The idea behind the optimization is as follows. All of our earlier circuit-checking protocols only make use of the multilinear
extension $\tilde{V}_i$ of the function $V_i$ mapping gate labels at  layer $i$ of the circuit to their values. 
In some cases, there is something to be gained by using a higher-degree extension
of $V_i$, and this is precisely what we exploit here. 
By using a higher-degree extension of the gate
values in the circuit, we are able to apply the sum-check protocol to a polynomial that differs from the one used in Section \ref{sec:ourrefinements}. In particular,
the polynomial we use here avoids referencing the $\beta_{s_i}$ polynomial used in Section \ref{sec:ourrefinements}. Details follow.

When multiplying matrices $A$ and $B$ such that $AB=D$, let $A(i, j)$, $B(i,j)$ and $D(i, j)$ denote functions from $\{0, 1\}^{\log n} \times  \{0, 1\}^{\log n} \rightarrow \mathbb{F}_q$ that map input $(i, j)$ to $A_{ij}$, $B_{ij}$,
and $D_{ij}$
respectively. Let $\tilde{A}$, $\tilde{B}$, and $\tilde{D}$ denote their multilinear extensions. 

\begin{lemma} \label{finalmatmultlemma}
For all $(p_1, p_2) \in \mathbb{F}^{\log n} \times  \mathbb{F}^{\log n}$, 
$$\tilde{D}(p_1, p_2) = \sum_{p_3 \in \{0, 1\}^{\log n}} \tilde{A}(p_1, p_3) \cdot \tilde{B}(p_3, p_2)$$
\end{lemma}
\begin{proof} 
For all $(p_1, p_2) \in \{0, 1\}^{\log n} \times \{0, 1\}^{\log n}$, the right hand side is easily seen to equal $D(p_1, p_2)$,
using the fact that $D_{ij} = \sum_k A_{ik} B_{kj}$ and the fact that $\tilde{A}$ and $\tilde{B}$ agree with the functions $A(i, j)$ and $B(i, j)$ at
all Boolean inputs. Moreover, the right hand side
 is a multilinear polynomial in the variables of $(p_1, p_2)$. Putting these facts together implies that the right hand side  is the unique multilinear extension of the function $D(i, j)$.
 \end{proof}
 
 Lemma \ref{finalmatmultlemma} implies the following valid interactive proof protocol for matrix multiplication:
 $\cP$ sends a matrix $D^*$  claimed to equal the product $D=AB$. $\cV$ evaluates $\tilde{D}^*(r_1, r_2)$ 
 at a random point $(r_1, r_2) \in \mathbb{F}^{\log n} \times \mathbb{F}^{\log n}$. By the Schwartz-Zippel lemma,
 it is safe for $\cV$ to believe $D^*$ is as claimed, as long as $\tilde{D}^*(r_1, r_2) = \tilde{D}(r_1, r_2)$ (formally, if $D^* \neq D$, then $\tilde{D}^*(r_1, r_2) \neq \tilde{D}(r_1, r_2)$ with probability $1-2\log n/q$).
 In order to check that $\tilde{D}^*(r_1, r_2) = \tilde{D}(r_1, r_2)$, we invoke a sum-check protocol on the polynomial 
 $g_{r_1, r_2}(p_3) =  \tilde{A}(r_1, p_3) \cdot \tilde{B}(p_3, r_2)$.
  
$\cV$'s final check in this protocol requires her to compute $g_{r_1, r_2}(r_3)$ for a random point $r_3 \in \mathbb{F}^{\log n}$.
$\cV$ can do this by evaluating both of $\tilde{A}(r_1, r_3)$ and $\tilde{B}(r_3, r_2)$ with a single 
streaming pass over the input, and then multiplying the results. 

The prover can be made to run in time $T(n) + O(n^2)$ across all rounds of the sum-check protocol using the $V^{(j)}$ arrays
described in Section \ref{sec:ourrefinements} to quickly evaluate $\tilde{A}$ and $\tilde{B}$ at all of the necessary points.
The $V^{(j)}$ arrays are initialized in round 0 to equal the input matrices themselves, and there is no need for $\cP$ to maintain
an ``uncorrupted'' copy of the original input (though in practice this may be desirable). Thus, the $V^{(j)}$ arrays
can be computed using the storage $\cP$ initially devoted to the inputs, and $\cP$ needs to store just $O(1)$ additional field elements over
the course of the protocol ($\cP$ does not even need to store the messages sent by $\cV$, as $\cP$ need not refer to the $j$th message once the array $V^{(j)}$ is computed).
The claimed $s(n) + o(n^2)$ space usage bound for $\cP$ follows.
 
\begin{remark} Let $C$ be the circuit for naive matrix multiplication described in Section \ref{sec:ourrefinements}. 
Notice that the $3 \log n$-variate polynomial $h(p_1, p_2, p_3) = \tilde{A}(p_1, p_3) \cdot \tilde{B}(p_3, p_2)$ extends the function $V_i$ mapping gate labels at layer $i=\log n$ of $C$
to their values. However, $h$
is not the multilinear extension of $\tilde{V}_i$, as $h$ has degree two in the variables of $p_3$. 

Informally, Theorem \ref{thm:finalfinalthm} cannot be said to perform  ``circuit checking'' on $C$, since it is not necessary for $\cP$ to evaluate all of the gates in $C$; indeed, the prover
in Theorem \ref{thm:finalfinalthm} can run
in sub-cubic time using fast matrix multiplication algorithms. However, the use of a low-degree extension of the gate values at layer $\log n$ of $C$ allows one to view the protocol of Theorem \ref{thm:finalfinalthm}
as a direct extension of the circuit-checking methodology.
\end{remark}

\begin{remark} Consider the problem of computing a matrix power $M^{2^k}$ via repeated squaring. 
We may apply the protocol of Theorem \ref{thm:finalfinalthm} in $k$ iterations, with the $i$th iteration applied to inputs $A=B=M^{2^{k-i}}$.
The $i$th iteration of this protocol reduces a claim about an evaluation of the multilinear extension of $M^{2^{k-i+1}}$
to an analogous claim about the multilinear extension of $M^{2^{k-i}}$ at two points of the form $(r_1, r_3)$, $(r_3, r_2) \in \mathbb{F}^{\log n \times \log n}$. 
We can further reduce the claims about $(r_1, r_3)$, $(r_3, r_2)$ to a claim about a single point exactly as in the ``Reducing to Verification of a Single Point'' step of the GKR protocol.
We then move onto iteration $i+1$. Notice in particular that the verifier  only needs to observe the output matrix $M^{2^k}$ and the input matrix $M$ to run this protocol;
in particular, $\cP$ does not need to explicitly send the intermediate matrices  $M^{2^{k-i}}$ to $\cV$.
\end{remark}

We implemented the protocol just described (our implementation is sequential). 
The results are shown in Table \ref{tab:finaltab}, where the column labelled ``Additional Time for $\cP$'' denotes the time required to compute $\cP$'s
prescribed messages after $\cP$ has already computed the correct answer. We report the naive matrix multiplication time both when the computation is done using standard multiplication of 64-bit integers,
as well as when the computation is done using finite field arithmetic over the field with $q=2^{61}-1$ elements. The reported verifier runtime is
for the $O(n^2 \log n)$ time reported in Theorem \ref{thm:finalfinalthm}. The verifier's runtime could be improved using Lemma \ref{lemma:vu} at the cost of increasing $\cV$'s space usage to $O(n)$, but 
we did not implement this optimization. Moreover, if the input matrices are presented in row-major order, then the observation of Vu \etal\ described in Remark \ref{remark:vu} improves $\cV$'s runtime with no increase in space usage.

The main takeaways from Table \ref{tab:finaltab} are that  the verifier does indeed save substantial time relative to performing matrix multiplication locally,
and that the runtime of the prover is hugely dominated by the time required simply to
compute the answer. 

\begin{table}
\small
\centering
\begin{tabular}{|c|c|c|c|c|c|c|c|}
\hline
Implementation & Problem Size  & Naive Matrix Multiplication Time & Additional Time for $\cP$ & $\cV$ Time & Rounds\\
\hline
Theorem \ref{thm:finalfinalthm} & $2^{10} \times 2^{10}$ & 2.17 s over $\mathbb{Z}$ & 0.03 s & 0.67 s & 11 \\
						&				        & 9.11 s over $\mathbb{F}_q$ &  & & \\
\hline
Theorem \ref{thm:finalfinalthm} & $2^{11} \times 2^{11}$ & 18.23 s over $\mathbb{Z}$ & 0.13 s & 2.89 s & 12 \\
						&				        &   73.65 s over $\mathbb{F}_q$  & & &\\
\hline
\end{tabular}
\caption{Experimental results for the $n \times n$ \matmult\ protocol of Theorem \ref{thm:finalfinalthm}.}
\label{tab:finaltab}
\end{table}

\section{Conclusion}
\label{sec:conclusion}
We believe our results substantially advance the goal of achieving a truly practical general purpose implementation of interactive proofs. 
The $O(\log S(n))$ factor overhead in the runtime of the prover within prior implementations of the GKR protocol
is too steep a price to pay in practice, and 
our refinements (formalized in Theorem \ref{thm:generaltheorem}) remove this logarithmic factor overhead for circuits with regular wiring patterns.
Our experiments demonstrate that this protocols yields
a prover that is less than 10x slower than a C++ program that simply evaluates the circuit, and that our protocols
are highly amenable to parallelization. Exploiting similar ideas, we have also extended the reach of 
prior interactive proof protocols by describing an efficient protocol (formalized in Theorem \ref{thm:dataparallel}) for general data parallel computation,
and given a protocol for matrix multiplication in which the prover's overhead (relative to \emph{any} unverifiable algorithm) is just a low-order additive term.
The latter is a powerful primitive for verifying the many algorithms that repeatedly invoke matrix multiplication. 
A major message of our results is that the more structure that exists in a computation, the more efficiently it can be verified,
and that this structure exists in many real-world computations.

We believe two directions in particular are worthy of future work. The first direction is to build a full-fledged system implementing our protocol for data parallel computation.
Our vision is to combine our protocol with a high-level programming language allowing the programmer to easily specify data parallel computations, analogous 
to frameworks such as MapReduce. Any such program could be automatically compiled in the manner of Vu \etal\ \cite{allspice} into a circuit, 
and our protocol could be run automatically on that circuit. 
The second direction is to further enable such a compiler to automatically take advantage of our other refinements, which are targeted at computations that are not necessarily
data parallel. These refinements apply to a circuit on a layer-by-layer basis, so they may yield substantial speedups in practice even if they apply only to a subset of the layers of a circuit.

\medskip
\noindent \textbf{Acknowledgements.} The author is grateful to Frank McSherry for raising the question of outsourcing general data parallel computations,
and to 
Michael Mitzenmacher and Graham Cormode for discussions and feedback that greatly improved the quality of this manuscript.

\appendix
\section{Proof of Theorem \ref{thm:generaltheorem}}
\label{app:general}

\begin{proof} 
Consider layer $i$ of the circuit $C$. Since $\text{in}_1^{(i)}$ and $\text{in}_2^{(i)}$ are regular,
there is a subset of input bits $\mathcal{S}_i \subseteq [v]$ with $|\mathcal{S}_i| = c_i$ for some constant $c_i$ such that
each input bit in $[v] \setminus \mathcal{S}$ affects $O(1)$ of the output bits of $\text{in}_1^{(i)}$ and $\text{in}_2^{(i)}$. Number the input variables so that the numbers $\{1, \dots, c_i\}$
correspond to variables in $\mathcal{S}_i$. 

Let $\rho \in \{0, 1\}^{c_i}$ be an assignment to the variables in $\mathcal{S}$,
and let $I_{\rho}: \{0, 1\}^{s_i} \rightarrow \{0, 1\}$ denote the indicator function for $\rho$. For example,
if $c_i=3$ and $\rho=(1,0, 1)$, then $I_{\rho}(x) = 1$ if $x_1 = 1, x_2 = 0$, and $x_3 = 1$, and $I_{\rho}(x)=0$ otherwise.
Let $\tilde{I}_\rho$ denote the multilinear extension of $I_{\rho}$. In the previous example, $\tilde{I}_{\rho} = x_1 (1-x_2) x_3$.
Finally, let $\text{in}^{(i)}_{1, \rho}$ and $\text{in}^{(i)}_{2, \rho}$  denote the functions $\text{in}^{(i)}_{1}$ and $\text{in}^{(i)}_{2}$ with the variables in $\mathcal{S}_i$
fixed to the assignment $\rho$, and for $k \in \{1, 2\}$, let $b_{\rho, k, j}$ denote the $j$th output bit of $\text{in}^{(i)}_{k, \rho}$.

By regularity, for each assignment $\rho \in \{0, 1\}^{c_i}$ to the variables in $\mathcal{S}_i$, the $j$th output bit $b_{\rho, k, j}$ of $\text{in}^{k}_{\rho}$ depends on only one variable $x_{q(\rho, k, j)} \in[s_i] \setminus \mathcal{S}_i$
for some function $q(\rho, k, j)$. 
Let $\tilde{b}_{\rho, k, j}(x_{q(\rho, k, j)}) : \mathbb{F} \rightarrow \mathbb{F}$ denote
the multilinear extension of the function $b_{\rho, k, j}(x_{q(\rho, k, j)}) : \{0, 1\} \rightarrow \{0, 1\}$. If $b_{\rho, k, j}$ is not identically 0 or identically 1, then either $\tilde{b}_{\rho, k, j}(x_{q(\rho, k, j)}) = x_{q(\rho, k, j)}$ or $\tilde{b}_{\rho, k, j} = 1-x_{q(\rho, k, j)}$.

For any $\rho \in \{0, 1\}^{s_i}$, define $\tilde{\text{in}}_{1, \rho}^{(i)}$ to be the concatenation of the $\tilde{b}_{\rho, 1, j}$ functions for all $j \in [s_{i+1}]$. 
Under this definition, $\tilde{\text{in}}_{1, \rho}^{(i)}$ is a collection of $s_{i+1}$ linear polynomials, where each of the polynomials depends on a single variable, and we may
 view $\tilde{\text{in}}_{1, \rho}^{(i)}$ as a single function mapping $\mathbb{F}^{s_i}$ to $\mathbb{F}^{s_{i+1}}$. 
We define  $\tilde{\text{in}}_{2, \rho}^{(i)}$ and $\tilde{\text{type}}_{\rho}^{(i)}$ 
analogously to $\tilde{\text{in}}_{1}$.

Now let 

\begin{flalign*} & W^{(i)}(p) = \end{flalign*}
\begin{flalign*}  \!\!\!\!\!\!\!\!\!\!\!\!\!\!\!\!\!\!\!\!\!\!\!\!\!\!\!\! \sum_{\rho \in L^{(i)}} \tilde{I}_\rho(p) \cdot \left(\tilde{\text{type}}_{\rho}^{\left(i\right)}\left(p\right)\cdot \tV_{i+1}\left(\tilde{\text{in}}_{1, \rho}^{\left(i\right)}\left(p\right)\right) \cdot \tV_{i+1}\left(\tilde{\text{in}}_{2, \rho}^{\left(i\right)}\left(p\right)\right)+ \left(1-\tilde{\text{type}_{\rho}}^{\left(i\right)}\left(p\right)\right) \left(\tV_{i+1}\left(\tilde{\text{in}}_{1, \rho}^{\left(i\right)}\left(p\right)\right) + \tV_{i+1}\left(\tilde{\text{in}}_{2, \rho}^{\left(i\right)}\left(p\right)\right)\right)\right).\end{flalign*}

It is easily checked that for all $p \in \{0, 1\}^{s_i}$, $V_i\left(p\right) =W^{(i)}(p).$
Lemma \ref{prop:big} then implies that 
$\tV_i(z) = \sum_{p \in \{0, 1\}^{s_i}} g_z^{(i)}(p), $
where $g_z^{(i)}(p) = \beta_{s_i}(z, p) \cdot W^{(i)}(p)$.
Our protocol follows precisely the description of Section \ref{sec:youroutline}, with $\cP$ and $\cV$ applying the sum-check protocol to the polynomial $g_z^{(i)}$ at iteration $i$.

\medskip
\noindent \textbf{Communication Costs and Costs to $\cV$.} 
Notice that our polynomial $g_z^{(i)}(p)= \beta(z, p) \cdot W^{(i)}(p)$ has degree $O(1)$ in each variable. Indeed, $\beta(z, p)$ has degree 1 in each variable. Moreover, $W^{(i)}(p)$ is a sum of polynomials
that each have degree $O(1)$ in each variable, and hence $W^{(i)}(p)$ itself has degree $O(1)$ in each variable. 

This latter fact can be seen by observing that for each assignment $\rho \in \{0,1\}^{c_i}$ to the variables in $\mathcal{S}_i$, it holds that $\tilde{I}_\rho(p)$,
 $\tilde{\text{type}}_{\rho}^{\left(i\right)}\left(p\right)$, $\tV_{i+1}\left(\tilde{\text{in}}_{1, \rho}^{\left(i\right)}\left(p\right)\right)$ and $\tV_{i+1}\left(\tilde{\text{in}}_{2, \rho}^{\left(i\right)}\left(p\right)\right)$
 all have constant degree in each variable. That $\tV_{i+1}\left(\tilde{\text{in}}_{1, \rho}^{\left(i\right)}\left(p\right)\right)$ and $\tV_{i+1}\left(\tilde{\text{in}}_{2, \rho}^{\left(i\right)}\left(p\right)\right)$
have constant degree in each variable follows from the facts that $\tV_{i+1}$ is a multilinear polynomial, and that each input variable $j \in [s_i] \setminus \mathcal{S}_i$ affects at most a constant number of outputs
for $\tilde{\text{in}}_{1, \rho}$ and $\tilde{\text{in}}_{2, \rho}$ by Property 1 of Definition \ref{def:regular}.

Since $g_z^{(i)}(p)$ has degree $O(1)$ in each variable, the claimed communication cost and the costs to the verifier follow immediately by summing the corresponding costs of the sum-check protocols over all iterations $i \in \{1, \dots, d(n)\}$ (see Section \ref{sec:sumcheck}). 

\medskip
\noindent \textbf{Time Cost for $\cP$.} 
It remains to demonstrate how $\cP$ can compute her prescribed messages when applying the sum-check protocol to the polynomial $g_z^{(i)}$ in time $O(S_i + S_{i+1})$. 
It will follow that $\cP$'s runtime over all $d(n)$ invocations of the sum-check protocol is $O(\sum_{i=1}^{d(n)} S_i) = O(S(n))$.

As in our analysis of Section \ref{sec:algs}, it suffices to show how $\cP$ can quickly evaluate $g_z^{(i)}$ at all points in $S^{(j)}$, where $S^{(j)}$ consists of all
points of the form
$p=(r_1, \dots, r_{j-1}, t, p_{j+1}, \dots, p_{s_i})$ with $t \in \{0, 1, \dots, \deg_j(g_z^{(i)})\}$ and $(p_{j+1}, \dots, p_{s_i}) \in \{0, 1\}^{s_i-j}$. 
As $g_z^{(i)}(p) =  \beta_{s_i}(z, p) \cdot W^{(i)}(p)$, it suffices for $\cP$ to evaluate $\beta_{s_i}(z, \cdot)$ and $W(\cdot)$ at all such points $p$. 
The $\beta_{s_i}(z, \cdot)$ computations can be done in $O(S_i)$ total time across all iterations
of the sum-check protocol, exactly as in Section \ref{sec:betas}. 

To see how $\cP$ can efficiently evaluate all of the $W^{(i)}(p)$ values efficiently, notice that for any fixed point $p \in \mathbb{F}^{s_i}$, $W^{(i)}(p)$ can be computed
efficiently given $\tilde{\text{type}}^{(i)}_{\rho}(p)$, $\tilde{V}_{i+1}(\tilde{\text{in}}_{1, \rho} (p))$, and $\tilde{V}_{i+1}(\tilde{\text{in}}_{2, \rho} (p))$ for all $\rho \in \{0, 1\}^{c_i}$.
As $|\mathcal{S}_i|=c_i=O(1)$, modulo a constant-factor blowup in runtime it suffices to explain how to 
perform these evaluations for a fixed restriction $\rho \in \{0, 1\}^{c_i}$ to the variables in $\mathcal{S}_i$. 

It is easy to see that $\tilde{\text{type}}^{(i)}_{\rho}(p)$ can be evaluated in constant time, since 
this function depends on only 1 input variable $x_{q(\rho, 3, 1)}$. All that remains is to show how $\cP$ can evaluate  $\tilde{V}_{i+1}(\tilde{\text{in}}_{1, \rho} (p))$ quickly;
the case for $\tilde{V}_{i+1}(\tilde{\text{in}}_{2, \rho} (p))$ is similar.

To this end, we follow the approach of Section \ref{sec:Vs}. 

\medskip \noindent 
\textit{Pre-processing.} $\cP$ will begin by computing an array $V^{(0)}$, which is simply defined to be the vector of gate values at layer $i+1$
i.e., identifying a number $0 < j < S_{i+1} $ with its binary representation in $\{0, 1\}^{s_{i+1}}$, $\cP$
sets $V^{(0)}[(j_1, \dots, j_{s_{i+1}})] = V_{i+1}(j_1, \dots, j_{s_{i+1}})$ for each $(j_1, \dots, j_{s_{i+1}}) \in \{0, 1\}^{s_{i+1}}$. 
The right hand side of this equation is simply the value of the $j$th gate at layer $i+1$ of $C$. 
So $\cP$ can fill in the array $V^{(0)}$ when she evaluates the circuit $C$, before receiving any messages from $\cV$.

 \medskip
\noindent \textit{Overview of Online Processing.}
Assume without loss of generality that the output bits of $\tilde{\text{in}}_{1, \rho} (p)$ are labelled in increasing order of the input bits they are affected by. So for example
if $p_1$ affects 2 output bits of $\tilde{\text{in}}_{1, \rho}$ and $p_2$ affects 3 output bits, then the bits affected by $p_1$ are labelled 1 and 2 respectively,
while the bits affected by $p_2$ are labelled 3, 4, and 5.

In round $j$ of of the sum-check protocol, $\cP$ needs to evaluate the polynomial $\tV_{i+1}$ at the $O(2^{s_{i+1}-j})$ points in the sets $\tilde{\text{in}}_{1, \rho}(S^{(j)})$ and
$\tilde{\text{in}}_{2, \rho}(S^{(j)})$. 
$\cP$ will do this using the help of intermediate arrays as follows.

\medskip \noindent \textit{Efficiently Constructing $V^{(j)}$ Arrays.} Let $a_{j-1}$ denote the total  number of output bits affected by the first $j-1$ input variables. 
Inductively, assume $\cP$ has computed in the previous round an array $V^{(j-1)}$ of length $2^{s_{i+1}-a_{j-1}}$, such that for each $p = (p_{a_{j-1}+1}, \dots, p_{s_{i+1}}) \in \{0, 1\}^{s_{i+1}-a_{j-1}}$,
the $p$th entry of $V^{(j-1)}$ equals

$$V^{(j-1)}[(p_{a_{j-1}+1}, \dots, p_{s_{i+1}})] = \sum_{(c_1, \dots, c_{a_{j-1}}) \in \{0, 1\}^{a_{j-1}}} V_{i+1}(c_1, \dots, c_{a_{j-1}}, p_{a_{j-1}+1}, \dots, p_{s_{i+1}}) \cdot \prod_{k=1}^{j-1} \chi_{c_k}(\tilde{b}_{\rho, 1, k}(r_{q(\rho, 1, k)})),$$
where recall that $q(\rho, 1, k)$ is the input bit that output bit $k$ of $\text{in}_{1, \rho}$ depends on.
As the base case, we explained how $\cP$ can fill in $V^{(0)}$ in the process of evaluating the circuit $C$.  

Let $x_1, \dots, x_{s_i}$ denote the input variables to $\text{in}_1$, and let $b_1, \dots, b_{s_{i+1}}$ denote the outputs of $\text{in}_1$. 
Intuitively, at the end of round $j$ of the sum-check protocol, $\cP$ must ``bind'' input variable $x_j$ to value $r_j \in \mathbb{F}$. 
This has the effect of binding the output variables affected by $x_j$, since each such output variable depends only on $x_j$. 
For illustration, suppose the variable $x_1$ affects output variable $b_1$; specifically, suppose that $b_1 = 1-x_1$. Then binding $x_1$ to value $r_1$ has the effect of 
binding $b_1$ to value $1-r_1$.  $V^{(j)}$ is obtained from $V^{(j-1)}$ by taking this into account. We formalize this as follows.

Assume that variable $x_j$ affects only one output variable $b_{\rho, 1, a_{j-1}+1}$, and thus $a_j = a_{j-1}+1$; if this is not the case, we
can compute $V^{(j)}$ by applying the following update once for each output variable affected by $x_j$. Observe that $\cP$ can compute
$V^{(j)}$ given $V^{(j-1)}$ in $O(2^{s_{i+1}-a_{j-1}})$ time using the following recurrence:

 $$V^{(j)}[(p_{a_{j}+1}, \dots, p_{s_{i+1}})]  = V^{(j-1)}[(0, p_{a_{j}+1}, \dots, p_{s_{i+1}})] \cdot \chi_0(\tilde{b}_{\rho, 1, a_j}(r_j)) + V^{(j-1)}[(1, p_{a_{j}+1}, \dots, p_{s_{i+1}})] \cdot \chi_1(\tilde{b}_{\rho, 1, a_{j}}(r_j)).$$

Thus, at the end of round $j$ of the sum-check protocol, when $\cV$ sends $\cP$ the value $r_j$, $\cP$ can compute $V^{(j)}$ from $V^{(j-1)}$ in $O(2^{s_{i+1}-a_{j-1}})$ time. 
 
\medskip \noindent \textit{Using the $V^{(j)}$ Arrays.}
We now show how to use the array $V^{(j-1)}$ to evaluate $\tilde{V}_{i+1}(\tilde{\text{in}}_{1, \rho}(p))$ in constant time for any point $p$ of the form $p=(r_1, \dots, r_{j-1}, t, p_{j+1}, \dots, p_{s_i})$  with
$(p_{j+1}, \dots, p_{s_i}) \in \{0, 1\}^{s_i-j}$.  In order to ease notation in the following derivation,
 we make the simplifying assumption that  $\tilde{b}_{\rho, 1, k}(x_{q(\rho, 1, k)}) = x_{q(\rho, 1, k)}$ for all output bits $k \in [s_{i+1}]$. The derivation when this assumption does not hold is similar.

We exploit the following sequence of equalities:

\begin{align*}\!\!\!\!\!\!\!\!\!\!\!\!\!\!\!\!\!\!\!\!\!\!\!\!\!\!\!\!\!\!\!\!\!\!\!\!\!\!\!\!\!\!\!\!\!\!\!\!\!\!\!\!\!\!\!\!\!\!\!\! & \tilde{V}_{i+1}(\tilde{\text{in}}_{1, \rho}(p))  
  =   \sum_{c \in \{0, 1\}^{s_{i+1}}} V_{i+1}(c) \chi_c(\tilde{\text{in}}_{1, \rho}(p))\\
\!\!\!\!\!\!\!\!\!\!\!\!\!\!\!\!\!\!\!\!\!\!\!\!\!\!\!\!\!\!\!\!\!\! & =   \sum_{(c_1, \dots, c_{a_{j-1}}) \in \{0, 1\}^{a_{j-1}}} \sum_{(c_{a_{j-1}+1}, \dots, c_{s_{i+1}}) \in \{0, 1\}^{s_{i+1}- a_{j-1}}} \!\!\!\!\!\!\!\!\!\!\!\!\!\!\!\!\!\!\!\!V_{i+1}(c) \chi_c(\tilde{\text{in}}_{1, \rho}(p)) \\
\!\!\!\!\!\!\!\!\!\!\!\!\!\!\!\!\!\!\!\!\!\!\!\!\!\!\!\!\!\!\!\!\!\!\!\!\!\!\!\!\!\!\!\!\!\!\!\!\!\!\!\!\!\!\!\!\!\!\!\!\!\!\!\!\!\!\!\! & = \sum_{(c_1, \dots, c_{a_{j-1}}) \in \{0, 1\}^{a_{j-1}}} \sum_{(c_{a_{j-1}+1}, \dots, c_{s_{i+1}}) \in\{0, 1\}^{s_{i+1}- a_{j-1}}} \!\!\!\!\!\!\!\!\!\!\!\!\!\!\!\!\!\!\!\!V_{i+1}(c)  \left(\prod_{k=1}^{a_{j-1}} \chi_{c_k}(\tilde{b}_{\rho, 1, k}(r_{q(\rho, 1, k)})) \right) 
\left(\prod_{k = a_{j-1}+1}^{a_j} \chi_{c_k}(\tilde{b}_{\rho, 1, k}(t)) \right) \left(\prod_{k=a_j+1}^{s_{i+1}} \chi_{c_k}(p_{q(\rho, 1, k)})\right)\\
\!\!\!\!\!\!\!\!\!\!\!\!\!\!\!\!\!\!\!\!\!\!\!\!\!\!\!\!\!\!\!\!\!\! & = \sum_{(c_1, \dots, c_{a_{j}}) \in \{0, 1\}^{a_j}} V_{i+1}(c_{j+1}, \dots, c_{a_j}, p_{q(\rho, 1, a_j+1)},  \dots, p_{q(\rho, 1, s_{j+1})})  \left(\prod_{k=1}^{a_{j-1}} \chi_{c_k}(r_k) \right) \cdot 
\left( \prod_{k=a_{j-1}+1}^{a_j} \chi_{c_k}(t)\right)\\
\!\!\!\!\!\!\!\!\!\!\!\!\!\!\!\!\!\!\!\!\!\!\!\!\!\!\!\!\!\!\!\!\!\! & = \sum_{(p_{a_{j-1}+1}, \dots, p_{a_j}) \in \{0, 1\}^{a_j - a_{j-1}}} V^{(j-1)}[(p_{q(\rho, 1, a_{j-1}+1)},  \dots, p_{q(\rho, 1, s_{j+1})})] \cdot \prod_{k=a_{j-1}+1}^{a_j} \chi_{p_k}(t) \end{align*}
 
 Here, the first equality holds by Equation \eqref{eq:linext}. The third holds by definition of the functions $\chi_c$ and $\tilde{\text{in}}_1$, as well as the assumption that 
 $\tilde{b}_{\rho, 1, k}(x_{q(\rho, 1, k)}) = x_{q(\rho, 1, k)}$ for all $k \in [s_{i+1}]$. The fourth holds because for Boolean values $c_k, p_{q(\rho, 1, k)} \in \{0, 1\}$, $\chi_{c_k}(p_{q(\rho, 1, k)})=1$
 if $c_k=p_{q(\rho, 1, k)}$, and $\chi_{c_k}(p_{q(\rho, 1, k)})=0$ otherwise. The final equality holds by definition of the array $V^{(j-1)}$. 
 
 The final expression above can be computed with $O(2^{a_{j} - a_{j-1}})$ time given the array $V^{(j-1)}$. Since 
 $a_j - a_{j-1}$ is constant by Property 1 of Definition \ref{def:regular}, $O(2^{a_{j} - a_{j-1}}) = O(1)$.

\medskip \noindent \textit{Putting Things Together.}
In round $j$ of the sum-check protocol, $\cP$ uses the array $V^{(j-1)}$ to evaluate $\tV_{i+1}(\tilde{\text{in}}_1(p))$ for all $O(2^{s_i-j})$ points $p \in S^{(j)}$, which requires
constant time per point and hence $O(2^{s_{i}-j})$ time over all points in $S^{(j)}$.
At the end of round $j$, $\cV$ sends $\cP$ the value $r_j$, and 
$\cP$ computes $V^{(j)}$ from $V^{(j-1)}$ in $O(2^{s_{i+1}-a_{j-1}})$ time.
By ordering input variables in such a way that $a_{j} > a_{j-1}$ for all $j$, we ensure that
 in total across all rounds of the sum-check protocol, $\cP$ 
spends $O(\sum_{j=1}^{s_i} 2^{s_i - j } +2^{s_{i+1}-j}) = O(2^{s_i} + 2^{s_{i+1}})$ time to evaluate $\tV_{i+1}$ at the relevant points.
When combined with our $O(2^{s_i})$-time algorithm for computing all the relevant $\beta(z,p)$ values, we see that $\cP$ 
takes $O(2^{s_i} + 2^{s_{i+1}}) = O(S_i + S_{i+1})$ time to run the entire sum-check protocol for iteration $i$ of our circuit-checking protocol.

\medskip
\textbf{Reducing to Verification of a Single Point.} After executing the sum-check protocol at layer $i$ as described above, 
$\cV$ is left with a claim about $\tilde{V}_{i+1}(\omega_1)$ and $\tilde{V}_{i+1}(\omega_2)$ from two points $\omega_1, \omega_2 \in \mathbb{F}^{s_{i+1}}$.
If $i$ is a layer for which $\text{in}_1^{(i)}$ and  $\text{in}_2^{(i)}$ are similar (see Definition \ref{def:similar}), we run the reducing to verification of a single point phase exactly
as in the basic GKR protocol. This requires $\cP$ to send $\tilde{V}_{i+1} (\ell(t))$ for a canonical line $\ell(t)$ that passes through the points $\omega_1$ and $\omega_2$.
Because $\text{in}_1^{(i)}$ and  $\text{in}_2^{(i)}$ are similar, it is easily seen that $\tilde{V}_{i+1} (\ell(t))$ is a univariate polynomial of constant degree.
Hence $\cP$ can specify $\tilde{V}_{i+1} (\ell(t))$ by sending $\tilde{V}_{i+1}(\ell(t_j))$ for $O(1)$ many points $t_j \in \mathbb{F}$. Using
the method of Lemma \ref{lemma:vu}, $\cP$ can evaluate $\tilde{V}_{i+1}$ at each point $\ell(t_j)$ in $O(S_{i+1})$ time, and hence can perform all
$\tilde{V}_{i+1}(\ell(t_j))$ evaluations in $O(S_{i+1})$ time in total.
  
Let $c=O(1)$ be the number of layers $i$ for which $\text{in}_1^{(i)}$ and  $\text{in}_2^{(i)}$  are not similar. At each such layer $i$, 
we skip the ``reducing to verification at a single point'' phase of the protocol. Each time we do this, it  
doubles the number of points $\omega \in \mathbb{F}^{s_{i+1}}$ that must be considered at the next iteration. However, we only
skip the ``reducing to verification at a single point'' phase $c$ times, and thus at all layers $i$ of the circuit,
$\cV$ needs to check $\tilde{V}_{i}(\omega_j)$ for at most $2^c=O(1)$ points. This affects $\cP$'s and $\cV$'s runtime by at most a $2^c=O(1)$ factor,
and the $O(S)$ time bound for $\cP$, and the $O(n \log n + d(n) \log S(n))$ time bound for $\cV$ follow.
\end{proof}

\section{Analysis for Pattern Matching}
\label{app:patternmatch}

Let $C$ be the circuit for pattern matching described in Section \ref{sec:applications}. Our goal in this appendix is to handle the layer of the circuit adjacent to the
input layer.  Call this layer $\ell$. Layer $\ell$ computes $t_{i+k}-p_k$ for each pair $(i, k) \in [[n]] \times [[q]]$. We want to show how to use a sum-check protocol
to reduce a claim about the value of $\tilde{V}_{\ell}(z)$ for some $z \in \mathbb{F}^{s_{\ell}}$ to a claim about $\tilde{V}_{\ell+1}(r)$ for some $r \in \mathbb{F}^{s_{\ell+1}}$,
while ensuring that $\cP$ runs in time $O(S_{\ell}) = O(nm)$. 

The idea underlying our analysis here is the following. The reason Theorem \ref{thm:generaltheorem} does not apply to layer $\ell$ is that
the first in-neighbor of a gate with label $p=(i_1, \dots, i_{\log n}, k_1, \dots, k_{\log m}) \in \{0, 1\}^{\log n + \log m}$ has label equal to the binary representation of the integer $i + k$,
and a single bit $i_k$ can affect many bits  in the binary representation of $i+k$ (likewise, each bit in the binary representation of $i+k$ may be affected by many bits in the binary representation of $i$ and $k$).
In order to ensure that each bit of $p$ affects only a single bit of $y=\text{in}^{(\ell)}_1(p)$, we introduce $\log n$ dummy variables $(c_1, \dots, c_{\log n})$ and force the $j$th dummy variable $c_j$ to 
have value equal to the $j$th \emph{carry bit} when adding numbers $i$ and $k$ in binary.
Now each bit of $p$ affects only one output bit, and each output bit $y_j$ is only
affected by at most three ``input bits'': $i_j, k_j,$ and $c_j$ if $j \leq \log m$, and just $i_j$ and $c_j$ if $j > \log m$.

To this end, let $\phi : \{0, 1\}^4 \rightarrow \{0, 1\}$ be the function that evaluates to 1 on input $(i_1, k_1, c_0, c_1)$ if and only if $c_1=0$ and $i_1+k_1+c_0 < 2$ or $c_1= 1$ and $i+k+c_0 \geq 2$. That
is, $\phi$ outputs 1 if and only if $c_1$ is equal to the carry bit when adding $i_1, k_1,$ and $c_0$. Let $\tilde{\phi}$ be the multilinear
extension of $\phi$. Notice $\tilde{\phi}$ can be evaluated at any point $r \in \mathbb{F}^4$ in $O(1)$ time.

Now let $(i, k, c)$ denote a vector in $\mathbb{F}^{\log n} \times \mathbb{F}^{\log m} \times \mathbb{F}^{\log n}$,
and define $$\Phi(i, k, c) := \prod_{j=1}^{\log n} \tilde{\phi}(i_j, k_j, c_{j-1}, c_j),$$
where it is understood that $c_{-1}=0$ and $k_j=0$ for $j > \log m$.

For any Boolean vector $(i, k, c) \in \{0, 1\}^{\log n} \times \{0, 1\}^{\log m} \times \{0, 1\}^{\log n}$,
it is easily verified that $\Phi(i, k, c)=1$ if and only if for all $j$, $c_j$ equals the $j$th \emph{carry bit} when adding numbers $i$ and $k$ in binary.

Finally, let $\gamma : \{0, 1\}^3 \rightarrow \{0, 1\}$ be the function that evaluates to 1 on input $(i_1, k_1, c_1)$ if and only if $i_1 + k_1 + c_1 = 1 \mod 2$. Let $\tilde{\gamma}$ be
the multilinear extension of $\gamma$.  Notice $\tilde{\gamma}$ can be evaluated at any point $r \in \mathbb{F}^3$ in $O(1)$ time.

Now consider the following $\log n + \log m$-variate polynomial over the field $\mathbb{F}$:

$$W^{(\ell)}(i, k) = \sum_{(c_1, \dots, c_{\log n}) \in \{0, 1\}^{\log n}} \Phi(i, k, c) \cdot \left(  \tilde{T}(\tilde{\gamma}(i_1 + k_1 + c_{0}), \dots,\tilde{\gamma}(i_{\log n} + k_{\log n} + c_{\log n - 1})) -  \tilde{P}(k_1, \dots, k_{\log m})\right),$$

where again it is understood that $c_{-1}=0$ and $k_j=0$ for $j > \log m$. Here, $\tilde{T}$ is the multilinear extension of the input $T$, viewed as a function from $\{0, 1\}^{\log n}$ to $[n]$,
and $\tilde{P}$ is the multilinear extension of the input pattern $P$, viewed as a function from $\{0, 1\}^{\log m}$ to $[n]$. 

It can be seen that for all Boolean vectors $(i, k) = \{0, 1\}^{\log n} \times \{0, 1\}^{\log m}$, $W^{(\ell)}(i, k) = V_{\ell}(i, k)$. 
This is because for any $(i, k) \in \{0, 1\}^{\log n} \times \{0, 1\}^{\log m}$,  $\Phi(i, k, c)$ will be zero for all $c$ except
the $c$ consisting of the correct carry bits for $i$ and $k$, and for this input $c$,  $\tilde{T}(\tilde{\gamma}(i_1 + k_1 + c_{0}), \dots,\tilde{\gamma}(i_{\log n} + k_{\log n} + c_{\log n - 1}))$
will equal $T(i+k)$ when interpreting $i, k$ as integers in the natural way. 

Lemma \ref{prop:big} then implies that for all $z \in \mathbb{F}^{\log n + \log m}$,
$$\tilde{V}_{\ell}(z) = \sum_{(i, k) \in \{0, 1\}^{\log n} \times \{0, 1\}^{\log m}} \beta_{\log n + \log m}(z, (i, k)) \cdot W^{(\ell)}(i, k)$$
$$\!\!\!\!\!\!\!\!\!\!\!= \sum_{(i, k, c) \in \{0, 1\}^{\log n} \times \{0, 1\}^{\log m} \times \{0, 1\}^{\log n}}\!\!\!\!\!\!\!\!\!\!\!\!\!\!\!\!\!\!\!\!\!\!\!\!\!\!\!\!\!\!\!\!\!\!\!\!\!\!\!\! \beta_{\log n + \log m}(z, (i, k)) \cdot \Phi(i, k, c) \cdot  \left(  \tilde{T}(\tilde{\gamma}(i_1 + k_1 + c_{0}), \dots,\tilde{\gamma}(i_{\log n} + k_{\log n} + c_{\log n - 1}))-  \tilde{P}(j_1, \dots, j_{\log m})\right).$$ 

Therefore, in order to reduce a claim about $\tilde{V}_{\ell}(z)$ to a claim about   $\tilde{T}(r_1)$ and $\tilde{P}(r_2)$ for random
vectors $r_1 \in \mathbb{F}^{\log n}$ and $r_2 \in \mathbb{F}^{\log m}$, it suffices to apply the sum-check protocol to the 
$2 \log n + \log m$-variate polynomial $$g_z(i, k, c) =  \beta_{\log n + \log m}(z, (i, k)) \cdot \Phi(i, k, c) \cdot  \left(  \tilde{T}(\tilde{\gamma}(i_1 + k_1 + c_{0}), \dots,\tilde{\gamma}(i_{\log n} + k_{\log n} + c_{\log n - 1})) -  \tilde{P}(j_1, \dots, j_{\log m})\right).$$

It remains to show how to extend the techniques underlying Theorem \ref{thm:generaltheorem} to allow $\cP$ to compute all of the required messages in this sum-check protocol in $O(nm)$ time. For brevity, we restrict ourselves to a sketch of the 
techniques.  

The first obvious complication is that the sum defining $\cP$'s message in
a given round of the sum-check protocol has as many as $2^{2 \log n + \log m} = \Omega(mn^2) > nm$ terms.
Fortunately, the $\Phi$ polynomial ensures that almost all of these terms are zero: when considering 
any Boolean setting of the variables $i_j, k_j$, and $c_{j-1}$, the only setting of $c_j$ that $\cP$ must consider is
the one corresponding to the carry bit of $i_j + k_j+ c_{j-1}$ i.e., the unique setting of $c_j$ such that $\phi(i_j, k_j, c_{j-1}, c_j)=1$.  
This ensures that at round $3j$, $3j+1$, and $3j+2$ of the sum-check protocol applied to $g_z$, $\cP$ must only evaluate $g_z$ at $O(2^{\log n + \log m - j})$ terms,
which is falling geometrically quickly with $j$. 

We now turn to explaining how $\cP$ can evaluate $g_z$ at all necessary points in round $3j$, $3j+1$ and $3j+2$ in total time $O(2^{\log n + \log m - j})$. 
To accomplish this, it is sufficient for $\cP$ to evaluate $\beta_{\log n + \log m}$ at the necessary points, as well as $\Phi$, $\tilde{T}$, and $\tilde{P}$ at the necessary points.
The $\beta_{\log n + \log m}$ evaluations are handled exactly as in Theorem \ref{thm:generaltheorem} i.e., by using $C^{(j)}$ arrays
(but these arrays only get updated every time a variable $i_j$ or $k_j$ gets bound within the sum-check protocol;
no update is necessary when a variable $c_j$ gets bound).
The $\tilde{P}$ evaluations are also handled exactly as in Theorem \ref{thm:generaltheorem}, using $V^{(j)}$ arrays that
only need to be updated when a variable $k_j$ gets bound.

The $\tilde{T}$ evaluations require some additional explanation on top of the analysis of Theorem \ref{thm:generaltheorem}.
We want $\cP$ to be able to use $V^{(j)}$ arrays as in Theorem \ref{thm:generaltheorem} to evaluate $\tilde{T}$ at the necessary points
in constant time per point, but we need to make sure that $\cP$ can compute array $V^{(j)}$ from $V^{(j-1)}$ in time that 
 falls geometrically quickly with $j$. In order to do this,
it is essential to choose a specific ordering for the sum in the sum-check protocol. 

Specifically, we write the sum as:

$$\sum_{i_1} \sum_{k_1} \sum_{c_1} \sum_{i_2} \sum_{k_2} \sum_{c_2} \dots \sum_{i_{\log n}} \sum_{c_{\log n}} g_z(i, k, c).$$

This ensures that, e.g., $(i_1, k_1, c_1)$ are the first three variables in the sum-check protocol to become bound to random values in $\mathbb{F}$. 
The reason we must do this is so that every 3 rounds, another value  $\tilde{\gamma}(i_j + k_j + c_{j-1})$ feeding into $\tilde{T}$ becomes bound to a specific value
(and moreover the outputs of $\tilde{\gamma}(i_{j'} + k_{j'} + c_{j'-1})$ are unaffected by the bound variables for all $j' > j)$.
This is precisely the property we exploited in the protocol of Theorem \ref{thm:generaltheorem} to ensure that the $V^{(j)}$ arrays
there halved in size every round, and that $V^{(j)}$ could be computed from $V^{(j-1)}$ in time proportional to its size. So we can use
$V^{(j)}$ arrays to efficiently perform the $\tilde{T}$ evaluations, updating the arrays every time another  value  $\tilde{\gamma}(i_j + k_j + c_{j-1})$ feeding into $\tilde{T}$ becomes bound to a specific value.

Finally, the $\Phi$ evaluations can be handled as follows. Consider for simplicity round $3j$ of the protocol.
Recall that $\cP$ only needs to evaluate $\Phi$ at points for which $\phi_{j'}(i_{j'}, k_{j'}, c_{j'-1}, c_{j'})=1$ for all $j' > j$. 
Thus, for all $j' > j$, $\phi_{j'}$ does not affect the product defining $\Phi$.  
So in order to evaluate $\Phi$ at the relevant points, it suffices for $\cP$ to evaluate the $\phi_{j'}$s for $j' \leq j$. 
Now at round $3j$ of the protocol, all triples $(i_{j'}, k_{j'}, c_{j'})$ for $j' < j$ are already bound, say to the values $(r^{(i)}_{j'}, r^{(k)}_{j'}, r^{(c)}_{j'})$, 
and hence all the $\phi_{j'}$ functions
for $j' < j$ are themselves already bound to specific values. So in order to quickly determine the contribution of the $\phi_{j'}$s for $j' < j$ 
to the product defining $\Phi$, it suffices for $\cP$ to maintain the quantity $\prod_{j' < j} \phi_{j'}(r^{(i)}_{j'}, r^{(k)}_{j'}, r^{(c)}_{j'})$ 
over the course of the protocol, which takes just $O(\log n)$ time in total.
Finally, the contribution of $\phi_j$ to the product defining $\Phi$ can be computed in constant time per point. 
This completes the proof that $\Phi$ can be evaluated by $\cP$ at all of the necessary points in $O(1)$ time per point over all rounds of the sum-check protocol, and
completes the proof of the theorem.

\end{document}